\documentclass[12pt,a4paper]{article}
\usepackage[T1]{fontenc}
\usepackage[top=1in,left=1in,right=1in,bottom=1.4in]{geometry}
\usepackage{graphicx}
\usepackage{xspace}
\usepackage{anyfontsize}
\usepackage[fontsize=11.8pt]{fontsize}
\usepackage{helvet}

\usepackage{mathpazo,newpxtext}
\usepackage[x11names,dvipsnames]{xcolor}
\definecolor{myblue}{HTML}{15348e}
\definecolor{linkblue}{HTML}{0050bf}
\usepackage{titlesec}
\titleformat{\section}
    {\color{myblue}\Large\bfseries}
    {\thesection}
    {1em}
    {}
\titleformat{\subsection}
    {\color{myblue}\large\bfseries}
    {\thesubsection}
    {1em}
    {}
\titleformat{\subsubsection}
    {\color{myblue}\normalsize}
    {\thesubsubsection}
    {1em}
    {}
\titleformat{\paragraph}[runin]
  {\bfseries\normalsize\color{myblue}}
  {}
  {0em}
  {}
  [. \hspace{0.5em}]

\titlespacing*{\section}
  {0pt}    {2.5ex plus 1ex minus .2ex} {1.5ex plus .2ex}           

\titlespacing*{\subsection}
  {0pt}    {1.5ex plus .5ex minus .2ex} {1ex plus .2ex}           

\usepackage[bottom]{footmisc}
\usepackage{amsmath,amssymb,bm,mathtools}
\allowdisplaybreaks
\usepackage{amsthm}
\usepackage{thmtools}
\definecolor{theorems}{rgb}{.6,.1,0}
\declaretheoremstyle[headfont=\color{theorems}\bfseries,bodyfont=\itshape]{colortheorem}
\declaretheoremstyle[headfont=\color{theorems}\bfseries]{colorothers}

\declaretheorem[style=colortheorem]{theorem}
\declaretheorem[style=colortheorem]{proposition}
\declaretheorem[style=colortheorem]{corollary}
\declaretheorem[style=colortheorem]{lemma}
\declaretheorem[style=colortheorem]{observation}
\declaretheorem[style=colortheorem]{claim}
\colorlet{colorrm}{Black}
\declaretheoremstyle[headfont=\color{theorems}\bfseries,bodyfont=\normalfont]{colorrm}
\declaretheorem[style=colorrm]{example}
\declaretheorem[style=colorrm]{remark}

\newtheoremstyle{continuedexample}
  {\topsep}       {\topsep}       {}      {}              {\color{colorrm}\bfseries}     {.}             {.5em}          {\thmnote{#3}}  \theoremstyle{continuedexample}
\newtheorem*{continuedexampleinner}{}
\theoremstyle{plain}

\makeatletter
\renewenvironment{proof}[1][\proofname]{\par\pushQED{\qed}\normalfont
  \topsep6\p@\@plus6\p@\relax
  \trivlist
  \item[\hskip\labelsep
        \bfseries #1\@addpunct{.}]\ignorespaces
}{\popQED\endtrivlist\@endpefalse
}
\makeatother

\usepackage{subcaption}
\usepackage{tikz}
\usepackage{pgfplots}
\usetikzlibrary{arrows.meta,automata}
\pgfplotsset{compat=1.18}
\usepackage{enumitem}
\setlist[itemize]{itemsep=3pt}
\setlist[enumerate]{itemsep=3pt,label=(\alph*)}
\usepackage{natbib}
\usepackage[setpagesize=false,bookmarks=true,colorlinks=true,
    linkcolor=linkblue,urlcolor=linkblue,citecolor=linkblue,
    breaklinks=true]{hyperref}

\usepackage[capitalise]{cleveref}
\crefname{fig}{Figure}{Figures}
\crefname{tab}{Table}{Tables}
\crefname{assumption}{Assumption}{Assumptions}
\crefname{appendix}{Appendix}{Appendices}
\Crefname{appendix}{Appendix}{Appendices}
\crefname{observation}{Observation}{Observations}
\crefname{footnote}{Footnote}{Footnotes}
\crefname{claim}{Claim}{Claims}
\crefformat{footnote}{#2#1#3}
\crefformat{enumi}{#2#1#3}
\Crefformat{enumi}{#2#1#3}

\usepackage{etoolbox}
\newcommand{\addQEDstyle}[2]{\AtBeginEnvironment{#1}{\pushQED{\qed}\renewcommand{\qedsymbol}{#2}}\colorlet{QEDColor}{theorems}
  \AtEndEnvironment{#1}{\popQED}
}
\addQEDstyle{example}{{\color{QEDColor}$\blacklozenge$}}
\addQEDstyle{remark}{{\color{QEDColor}$\blacklozenge$}}
\renewcommand{\qedsymbol}{$\blacksquare$}

\makeatletter
\patchcmd{\maketitle}
 {\def\@makefnmark}
 {\def\@makefnmark{}\def\useless@macro}
 {}{}
\makeatother

\newcommand{\diag}[1]{\operatornamewithlimits{diag}[#1]}

\newcommand{\argmax}{\operatornamewithlimits{arg\,max}}

\newcommand{\sgn}{\operatorname{sgn}}

\providecommand{\Mult}{\operatorname{Mult}}
\providecommand{\Var}{\operatorname{Var}}
\providecommand{\trace}{\operatorname{trace}}
\providecommand{\la}{\langle}
\providecommand{\ra}{\rangle}
\providecommand{\Supp}{\operatorname{support}}
\providecommand{\Int}[1]{\operatorname{int}(#1)}

\providecommand{\OL}[1]{\overline{#1}}
\providecommand{\WH}[1]{\widehat{#1}}

\providecommand{\BbRe}{\mathbb{R}_{\ge0}}
\providecommand{\BbRg}{\mathbb{R}_{> 0}}
\providecommand{\BbZe}{\mathbb{Z}_{\ge 0}}

\providecommand{\br}{\operatorname{br}}
\providecommand{\BR}{\operatorname{BR}}

\providecommand{\TlL}{{\widetilde{L}}}
\providecommand{\TlR}{{\widetilde{R}}}
\providecommand{\TlF}{{\widetilde{F}}}
\providecommand{\TlG}{{\widetilde{G}}}
\providecommand{\Tlomega}{{\widetilde{\omega}}}
\providecommand{\Olu}{{\overline{u}}}

\providecommand{\Fix}{{\operatorname{FP}}}
\providecommand{\SfC}{{\mathsf{C}}}

\providecommand{\Htv}{\widehat{v}}
\providecommand{\Htq}{\widehat{q}}
\providecommand{\Htsigma}{\widehat{\sigma}}

\providecommand{\UlHeading}[1]{\medskip\noindent{{\textbf{#1}}.\hspace{.5em}}}

\providecommand{\Toz}{\downarrow0}

\usepackage{quoting}
\usepackage{setspace}

\providecommand{\Vt}[1]{\bm{#1}}

\providecommand{\Rme}{\mathrm{e}}

\providecommand{\Rmi}{\mathrm{i}}

\providecommand{\RmH}{\mathrm{H}}

\providecommand{\ClF}{\mathcal{F}}

\providecommand{\ClL}{\mathcal{L}}

\providecommand{\BbC}{\mathbb{C}}

\providecommand{\BbE}{\mathbb{E}}

\providecommand{\BbP}{\mathbb{P}}

\providecommand{\BbR}{\mathbb{R}}

\providecommand{\Htq}{\hat{q}}

\providecommand{\Htv}{\hat{v}}

\providecommand{\Htsigma}{\hat{\sigma}}

\providecommand{\Brx}{\bar{x}}

\providecommand{\Tlf}{\tilde{f}}

\providecommand{\Tlu}{\tilde{u}}

\providecommand{\Tlx}{\tilde{x}}

\providecommand{\TlF}{\tilde{F}}
\providecommand{\TlG}{\tilde{G}}

\providecommand{\TlL}{\tilde{L}}

\providecommand{\TlR}{\tilde{R}}

\providecommand{\Tlomega}{\tilde{\omega}}

\providecommand{\Dtx}{\dot{x}}
\providecommand{\Dty}{\dot{y}}

\providecommand{\Is}{\equiv}

\begin{document}

\setstretch{1.25}

\title{\large\bfseries{Sampling Logit Equilibrium and Endogenous Payoff Distortion}}

\author{\shortstack[c]{{Minoru Osawa}\\[.8em]\small \itshape Institute of Economic Research, Kyoto University}\\\small\ttfamily \href{mailto:osawa.minoru@gmail.com}{\color{myblue} osawa.minoru@gmail.com}}
\date{\vskip 1em\normalsize\today}
\maketitle 

\begingroup
\renewcommand\thefootnote{}
\footnotetext{\hspace{-2.15em}Website: {\ttfamily\href{https://www.minoruosawa.com/}{www.minoruosawa.com}}. I thank Srinivas Arigapudi, Chiaki Hara, Jonathan Newton, Daisuke Oyama, Yannick Viossat, and Dai Zusai for their comments. I also thank the editor and the referees for their valuable comments and suggestions.
I am grateful to the late Professor Sandholm for his beautiful book, ``\emph{Population Games and Evolutionary Dynamics},'' which has been a constant source of inspiration.
Financial support from JSPS KAKENHI JP22K04353 and JP26K07706 is gratefully acknowledged. The author acknowledges the use of ChatGPT 5.6 Sol and Claude Fable 5 in the course of this research. All mathematical arguments and proofs in the final manuscript were checked and written by the author.} 
\endgroup

\vspace{-2em}

\vspace{-1.5em}

\renewcommand{\abstractname}{}
\begin{abstract}
\noindent 
We introduce \emph{sampling logit equilibrium} (SLE) for population games in which agents infer payoffs from a finite sample of $k$ opponents' actions and respond according to a logit choice rule. We find that sampling error systematically distorts incentives. Actions whose inferred payoffs are relatively more variable earn a \emph{variance premium}, while nonlinear payoffs generate a \emph{curvature premium} through Jensen effects. For large samples, an SLE is approximated by a logit equilibrium of a \emph{virtual game} whose payoffs include these premiums. For linear games, when sampling error and logit noise vanish at comparable rates, their interaction around completely mixed Nash equilibria is reduces to a Gaussian--Gumbel limiting model. We also establish exact finite-$k$ uniqueness and stability results. In particular, in two-action coordination games with a $1/k$-dominant action, the unique SLE converges to the risk-dominant equilibrium as logit noise vanishes.
\end{abstract}
\bigskip

{\small


 }

\clearpage 

\section{Introduction}

Decision making in strategic environments often departs from perfect rationality through two distinct channels. 
One is \emph{stochastic choice}: even when agents observe payoffs correctly, their responses may be probabilistic due to idiosyncratic shocks or cognitive noise. 
The other is an \emph{informational constraint}: agents frequently evaluate actions using only limited observations of the environment, forming payoff assessments from a small sample of interactions. 
These mechanisms introduce noise in different ways. 
Stochastic choice perturbs the decision rule, whereas finite sampling distorts the evaluation of payoffs themselves.

These channels are often studied separately. 
Quantal-response models combine full-information payoff evaluation with stochastic choice \citep{McKelvey-Palfrey-GEB1995}, whereas sampling models typically combine limited observation with deterministic best response \citep{Osborne-Rubinstein-GEB2003}. 
This separation leaves open a natural question: how does behavior change when agents both observe only a finite sample of the environment and respond stochastically to the resulting payoff signals?

This study analyzes this interaction in a simple large-population framework combining finite sampling with logit stochastic choice. 
Under the $(k,\eta)$-\emph{sampling logit choice rule}, each agent draws $k$ independent samples of opponents' actions from the population, evaluates the resulting sample-based payoffs, and then selects an action according to a logit rule with noise level $\eta>0$. 
A \emph{sampling logit equilibrium} (SLE) is defined as a fixed point of this process.

The framework connects naturally to existing models in two limiting cases. 
For fixed $\eta>0$, sampling logit choice approaches the standard logit choice rule as the sample size $k$ grows. 
For fixed $k$, reducing $\eta$ concentrates choice on best responses to each sampled population state, linking the model to the \emph{sampling equilibrium} of \citet{Osborne-Rubinstein-GEB2003} and the sampling best-response dynamics studied by \citet{Oyama-Sandholm-Tercieux-TE2015}. 
By varying $(k,\eta)$, the model therefore connects fully informed stochastic choice with limited-information deterministic best response.

Our first main result shows that finite sampling and stochastic choice systematically distort incentives, and that these distortions have a simple deterministic representation. 
When the sample size $k$ is sufficiently large, the SLE of a game can be approximated by the logit equilibrium of a \emph{virtual game} with modified payoffs. 
The modification consists of two components. 
A \emph{variance premium} favors actions with greater variability in inferred payoffs relative to competing actions. 
A \emph{curvature premium} captures Jensen-type shifts produced when sampling variation passes through nonlinear payoff functions. 
Together, they show that finite sampling does not merely add randomness to behavior but systematically alters the incentives agents face.

The virtual-game representation is accurate when sampling noise is small relative to logit noise. 
In this case, however, sampling can only modify full-information logit behavior at the margin. 
Our second main result characterizes the critical regime in which the two frictions remain of the same order, so that neither reduces to a perturbation of the other.
Specifically, we assume that the noise level falls as the sample size grows, and $\sqrt{k}\eta$ converges to a positive constant $\tau$. 
When we consider this limit in linear population games, the SLE near a completely mixed Nash equilibrium converges to that equilibrium at rate $1/\sqrt{k}$, and we can characterize its scaled displacement as the solution of an exact limiting model with Gaussian--Gumbel noise structure.
The Gaussian component is the limiting sampling error in the inferred population state, correlated across actions through the payoff matrix, whereas the Gumbel component represents logit choice noise.
Also, the large-$\tau$ expansion of the limiting model recovers the variance premium discussed above. 

Beyond the asymptotic results for large $k$, we establish several exact results in selected small-$k$ environments. 
When each agent observes one opponent, the equilibrium is unique and globally attracting in any population game. 
When each agent observes two opponents, the same conclusion holds for two-action games. 
Moreover, in the class of two-action coordination games satisfying the \emph{$1/k$-dominance} condition, the equilibrium is unique for every $\eta>0$ and converges to the risk-dominant Nash equilibrium as logit noise vanishes, in line with earlier selection results for sampling-based dynamics \citep{Sandholm-IJGT2001,Kreindler-Young-GEB2013,Oyama-Sandholm-Tercieux-TE2015}. 
These results illustrate how finite sampling can sharpen equilibrium selection when agents operate under limited information.

\subsection*{Related literature}

The present study relates to four strands of literature: extensions of logit QRE, models of finite sampling and limited observation, stochastic evolutionary dynamics, and the virtual-payoff approach.

First, the framework extends the \emph{quantal response equilibrium} (QRE) paradigm, which incorporates stochastic choice into equilibrium analysis through random utility models \citep{McKelvey-Palfrey-GEB1995,McKelvey-Palfrey-EE1998,Goeree-Holt-Palfrey-EE2005}. 
In QRE, agents evaluate expected payoffs under full information and randomness arises from idiosyncratic payoff shocks. 
The present study instead introduces stochasticity through finite sampling of the environment. 
Payoff evaluations are based on a small number of observations, so randomness enters through distorted payoff estimates rather than purely idiosyncratic payoff noise. 
Experimental comparisons of stationary concepts have shown that models incorporating sampling often fit observed play better than QRE \citep{Selten-Chmura-AER2008}, suggesting that sampling may represent an important source of behavioral noise.

Second, the analysis builds on the literature studying limited observation in games. 
Early contributions include \citet{Osborne-Rubinstein-AER1998,Osborne-Rubinstein-GEB2003}, who introduced the concept of sampling equilibrium. 
Related work examines strategic behavior when agents observe only partial information about the environment or other players' actions \citep[e.g.,][]{Spiegler-ET2006,Spiegler-REStud2006,Salant-Cherry-ECTA2020,Danenberg-Spiegler-JPEM2025}. 
Learning dynamics under such informational constraints have also been studied extensively \citep[e.g.,][]{Sethi-GEB2000,Sandholm-IJGT2001,Oyama-Sandholm-Tercieux-TE2015,Heller-Mohlin-REStud2018,Mantilla-Sethi-Cardenas-JPET2020,Sawa-Wu-GEB2023,Arigapudi-etal-AEJMicro2024,Heller-Arigapudi-DP2026}. 
The framework here is particularly close to \citet{Oyama-Sandholm-Tercieux-TE2015}, who analyze dynamics in which agents occasionally sample opponents and best respond to the resulting empirical distribution. 
The present study replaces deterministic best responses with logit responses, yielding a smoother dynamic and enabling a tractable approximation and interpretation in terms of a virtual game.

Third, the model connects to the literature on \emph{stochastic evolutionary dynamics} and stochastic stability in games \citep{Foster-Young-TPB1990,Young-ECTA1993,Kandori-Mailath-Rob-ECTA1993,Ellison-ECTA1993}. 
In these models, random experimentation or errors drive transitions between states, and equilibrium selection is studied through the long-run behavior of the induced Markov process. 
A complementary stochastic-approximation approach is provided by \citet{Benaim-Weibull-ECTA2003}, who derive deterministic game dynamics as the finite-time expected behavior of stochastic revision processes in finite populations. 
Logit choice has been widely used as a modeling device in finite-population stochastic evolutionary models \citep[e.g.,][]{Blume-GEB1993,Blume-GEB1995,Alos-Ferrer-Netzer-GEB2010,Marden-Shamma-GEB2012}. 
The study most closely related to the present analysis is \citet{Kreindler-Young-GEB2013}, who analyze a finite-population model in which agents observe a random sample of opponents before making a logit choice. 
In the canonical two-action coordination game, their partial-information mean-field dynamic coincides exactly with the learning dynamic associated with the sampling logit choice rule. 
Their focus is finite-population hitting times in this two-action environment. 
The present study formulates a static equilibrium concept for general $n$-action population games and derives its virtual-payoff decomposition. 

Finally, the analysis relates to the \emph{virtual payoff} approach for representing stochastic choice models through deterministic payoff perturbations. 
In large-population games, \citet{Hofbauer-Sandholm-JET2007} show that stationary behavior under stochastic choice can often be represented as if agents were responding to deterministically perturbed payoffs. 
Related ideas appear in earlier discussions of logit-based models of product differentiation \citep{Anderson-dePalma-Thisse-Book1992}. 
We derive an explicit virtual payoff representation that captures the distortions generated by finite sampling under stochastic choice. 
In equilibrium, agents behave \emph{as if} they were responding to these modified incentives.

\section{Model}
\label{sec:model}
\subsection{Population games}
We focus on a large-population game played by a single homogeneous population. 
There is a unit mass of anonymous agents, each of whom chooses a pure action, where $S \Is \{1,2,\hdots, n\}$ denotes the common, finite set of available actions. 
The simplex $X \Is \{x \in \BbRe^n : \sum_{i\in S} x_i = 1\}$ represents the set of \emph{population states} where $x_i \in [0,1]$ is the fraction of agents playing action $i$. 
A state at which all agents play action $i$ is called a \emph{pure population state} and denoted by $e_i$. 
The function $F:X \to \BbR^n$ describes a game's payoffs, with $F_i(x)$ being the payoff obtained at state $x$ by agents playing action $i \in S$. 
Where $S$ is understood, $F$ identifies a \emph{population game}. 

The most basic instances of population games are those generated by random matching in normal form games. 
Suppose that, given the population state, agents are randomly matched to play a normal form game with payoff matrix $A = [a_{ij}] \in \BbR^{n \times n}$, where $a_{ij}$ is the payoff for an agent playing action $i \in S$ matched to another agent playing action $j \in S$. 
Then, the expected payoffs are given by $F(x) = Ax$, which we may call a \emph{linear population game}. 
We can identify a linear population game with its base payoff matrix $A$. 
    
Below, we introduce four choice rules in population games relevant for our discussion. 
For each rule, the associated game dynamic is defined as the expected motion of the population state \citep[see][Ch.4]{Sandholm-Book2010}. 

\subsection{Best response}
Given a population game $F$, the pure and mixed best response correspondences $\br:X \Rightarrow S$ and $\BR: X \Rightarrow X$ are respectively defined as
\begin{align}
    & \br(x) \Is \argmax_{i\in S} F_i(x)
    \quad \text{and } \\
    & \BR(x) \Is \argmax_{y \in X} \la y, F(x) \ra
    = \left\{y \in X : \Supp(y) \subseteq \br(x) \right\}, 
\end{align}
where we define $\la u, v \ra \Is \sum_{i} u_i v_i$ for vectors. 
A population state $x \in X$ is a \emph{Nash equilibrium} of $F$ if every agent is playing a pure action that is optimal given the others' behavior, in which case $x \in \BR(x)$. 
The \emph{best response dynamic} \citep{Gilboa-Matsui-ECTA1991,Hofbauer-1995} is defined as the following differential inclusion:
\begin{align}
    \Dtx \in \BR(x) - x. 
    \tag{\textsf{BRD}}
    \label{eq:BRD}
\end{align}
See \citet[Appendix A.1]{Oyama-Sandholm-Tercieux-TE2015} for a concise summary of differential inclusions. 
This approach assumes each agent has perfect information about the current population state and is able to choose an optimal action, which can be demanding in some contexts. 

\subsection{Best response under finite sampling}

\citet{Oyama-Sandholm-Tercieux-TE2015} considers an alternative large-population model that imposes milder informational requirements. 
When an agent receives a revision opportunity, the agent first observes $k \ge 1$ independent samples of opponents' play from the population. 
The set of possible outcomes of samples of size $k$ is $Z^k \Is \{ z \in \BbZe^n : \sum_{i\in S} z_i = k \}$. 
The \emph{empirical population state} according to a sample $z \in Z^k$ is $w = \frac{1}{k} z \in X$. 
Under uniform random sampling with replacement, the probability that $z \in Z^k$ is drawn at $x\in X$ follows the multinomial distribution $\Mult(k\!\mid\!x)$. 
That is, the probability mass of drawing $z\in Z^k$ at $x \in X$ is 
\begin{align}
    & M^k(z\!\mid\! x)= \binom{k}{z_1, z_2, \cdots, z_n} \cdot 
    x_1^{z_1} \cdot x_2^{z_2} \cdot \cdots \cdot x_n^{z_n}
    \quad \in [0,1]. 
\label{eq:mult}
\end{align}
The empirical population state $w$ is the maximum likelihood estimator of the population state under this sampling model. 
The agent then evaluates payoffs based on $w$ and plays a best response. 
For each $k$, the \emph{$k$-sampling best response correspondence} $\BR^k: X \Rightarrow X$ is 
\begin{align}
    \BR^k(x) \Is \BbE\big[\!\BR(\tfrac{1}{k} z)\big] = \sum_{z \in Z^k} M^k(z\!\mid\! x)\cdot \BR(\tfrac{1}{k} z). 
\end{align}
More precisely, $y \in \BR^k(x)$ if and only if $y = \sum_{z \in Z^k} M^k(z\!\mid\! x)\cdot \alpha(z)$ where $\alpha(z) \in \BR(\frac{1}{k} z)$ for each $z \in Z^k$. 
A $k$-\emph{sampling equilibrium} \citep{Osborne-Rubinstein-GEB2003} is a fixed point of $\BR^k$ satisfying $x \in \BR^k(x)$. 
The \emph{$k$-sampling best response dynamic} is defined as the following differential inclusion: 
\begin{align}
    \Dtx \in \BR^k(x) - x. 
    \tag{\textsf{SBRD}}
    \label{eq:sBRD}
\end{align}
While demanding less information about the population state, this model still assumes agents are rational responders.\footnote{Sampling equilibrium is a special case of \emph{sampling equilibrium with statistical inference} (SESI) \citep{Salant-Cherry-ECTA2020}. \cite{Sawa-Wu-GEB2023} extensively discusses Bayesian large-population dynamics corresponding to SESI.} 

\subsection{Logit choice} 
\label{sec:logit} 

We next recall the logit choice and the logit dynamic. 
For a payoff vector $u \in \BbR^n$, the \emph{$\eta$-logit map} $\Lambda^\eta: \BbR^n \to \Int{X}$ with \emph{noise level} $\eta > 0$ is defined as follows: 
\begin{align}
    & \Lambda_i^\eta (u) 
    \Is 
    \frac
    {\exp\left(\eta^{-1} u_i \right)}
    {\sum_{j \in S} \exp\left(\eta^{-1} u_j\right)}
    \quad \in (0,1)
    & \forall i \in S. 
    \label{eq:Lambda}
\end{align}
Then, at each $x \in X$, the \emph{$\eta$-logit choice rule} $P^\eta : X \to \Int{X}$ with noise level $\eta > 0$ is
\begin{align}
    & P_i^\eta (x) 
    \Is 
    \Lambda_i^\eta(F(x))
    =
    \frac
    {\exp\left(\eta^{-1} F_i(x) \right)}
    {\sum_{j \in S} \exp\left(\eta^{-1} F_j(x)\right)}
    \quad \in (0,1)
    & \forall i \in S. 
    \label{eq:P}
\end{align}
A population state $x$ is a \emph{$\eta$-logit equilibrium} if it is a fixed point of the rule \eqref{eq:P}: $x = P^\eta(x)$. 
Logit equilibrium is by far the most common formulation of QRE. 
The large-population \emph{$\eta$-logit dynamic} \citep[][Ch.4]{Fudenberg-Levine-Book1998} is defined by\footnote{Logit revision in finite-population settings yields a Markov chain or process rather than a differential equation, as considered in, e.g., \cite{Blume-GEB1993,Blume-GEB1995,Blume-GEB2003,Hofbauer-Sandholm-JET2007,Alos-Ferrer-Netzer-GEB2010,Marden-Shamma-GEB2012}, among others. Another foundation for the logit dynamic can be found in stochastic fictitious play \citep{Hofbauer-Sandholm-ECTA2002}.} 
\begin{align}
    \Dtx = P^\eta(x) - x. 
    \tag{\textsf{LD}}
    \label{eq:LD}
\end{align}

For any $\eta > 0$, logit equilibria are positive since $P^\eta$ is positive. 
The $\eta$-logit equilibria for a game $F$ coincide with the Nash equilibria for the deterministically distorted game $F+H$, where $H_i(x) \Is -\eta \log(x_i)$. 
In this sense, $F+H$ is a ``virtual'' payoff for the $\eta$-logit equilibrium problem \citep[][Appendix]{Hofbauer-Sandholm-JET2007}.\footnote{More specifically, requiring $x^* \in X$ and $\lambda = F_i(x^*)+H_i(x^*)$ for all $i\in S$ implies $\lambda = \eta \log \sum_{j\in S} \exp(\eta^{-1} F_j(x^*))$ and $x^* = P^\eta(x^*)$. See \cite{Behrens-Murata-JUE2021} for an application in spatial economics, where $H_i(x)$ represents material congestion externalities incurred by the households residing in \emph{location} $i\in S$.}

\subsection{Logit choice under finite sampling}
Our framework is a natural synthesis of sampling best response and logit choice. 
Assume that, after observing a sample $z\in Z^k$ and adopting $w = \frac{1}{k} z$ as the estimate for the population state, each agent follows the $\eta$-logit choice rule $P^\eta$ instead of the mixed best response $\BR$. 
The induced aggregate choice rule $L^{k,\eta}: X \to \Int{X}$ defines the \emph{$(k,\eta)$-sampling logit choice rule}:  
\begin{align}
    L^{k,\eta}(x) \Is \BbE \big[P^\eta(\tfrac{1}{k} z)\big] = 
    \sum_{z \in Z^k} M^k(z\!\mid\! x)\cdot P ^\eta (\tfrac{1}{k} z). 
    \label{eq:sL-rule}
\end{align}
A \emph{$(k,\eta)$-sampling logit equilibrium} (SLE) is a fixed point of $L^{k,\eta}$ satisfying $x = L^{k,\eta}(x)$. 
The \emph{$(k,\eta)$-sampling logit dynamic} can be defined as
\begin{align}
    \Dtx = L^{k,\eta} (x) - x. 
    \tag{\textsf{SLD}}
    \label{eq:sLD}
\end{align} 

Brouwer's fixed point theorem implies that an SLE exists for any $1 \le k < \infty$ and $\eta > 0$. 
The dynamic \eqref{eq:sLD} admits a unique global solution for every initial population state in $X$ as $L^{k,\eta}$ is continuous and globally Lipschitz on $X$. 
All SLE are necessarily positive because $L^{k,\eta}$ is strictly positive. 
For example, a pure population state cannot be a fixed point of $L^{k,\eta}$ because $L^{k,\eta}(e_i) = P^\eta(e_i) > \Vt0$ for any $i \in S$ under any $k \ge 1$ and $\eta > 0$.

\begin{figure}[t!]
    \centering
    \begin{subfigure}[c]{.49\textwidth}
        \centering
        \includegraphics[width=.8\hsize]{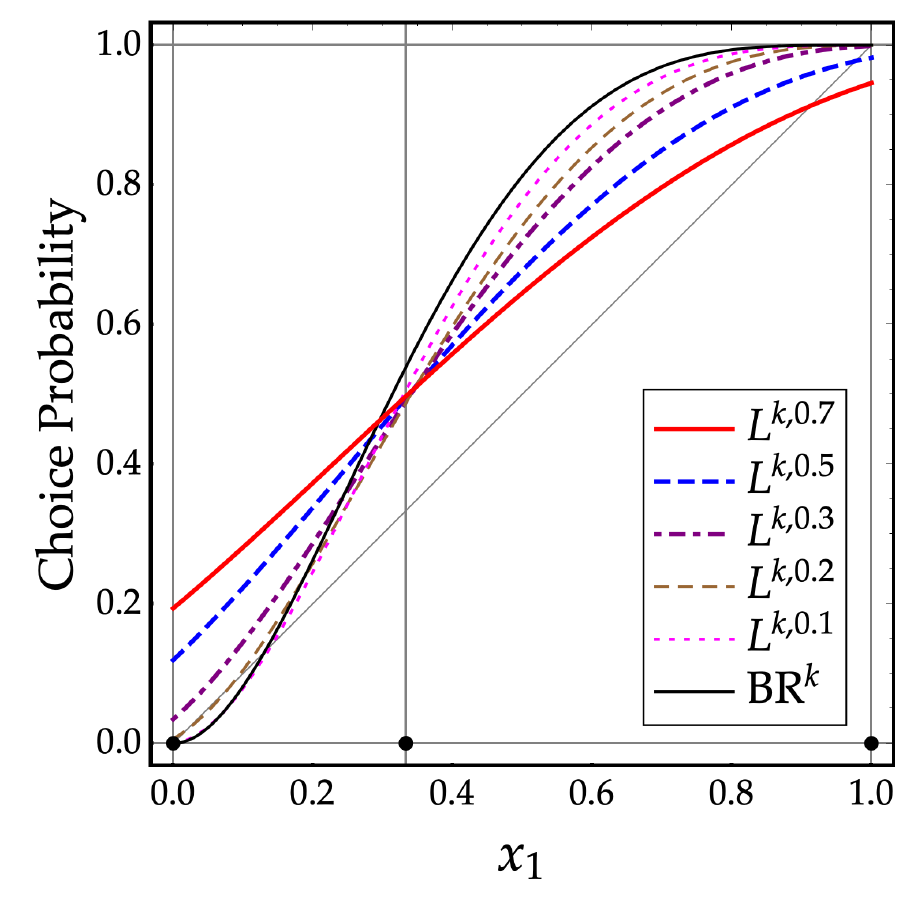}
        \caption{$L_1^{k,\eta}$ for different $\eta$ and $\BR_1^k$ ($k = 5$)}
        \label{fig:brs}
    \end{subfigure}
    \begin{subfigure}[c]{.49\textwidth}
        \centering
        \includegraphics[width=.8\hsize]{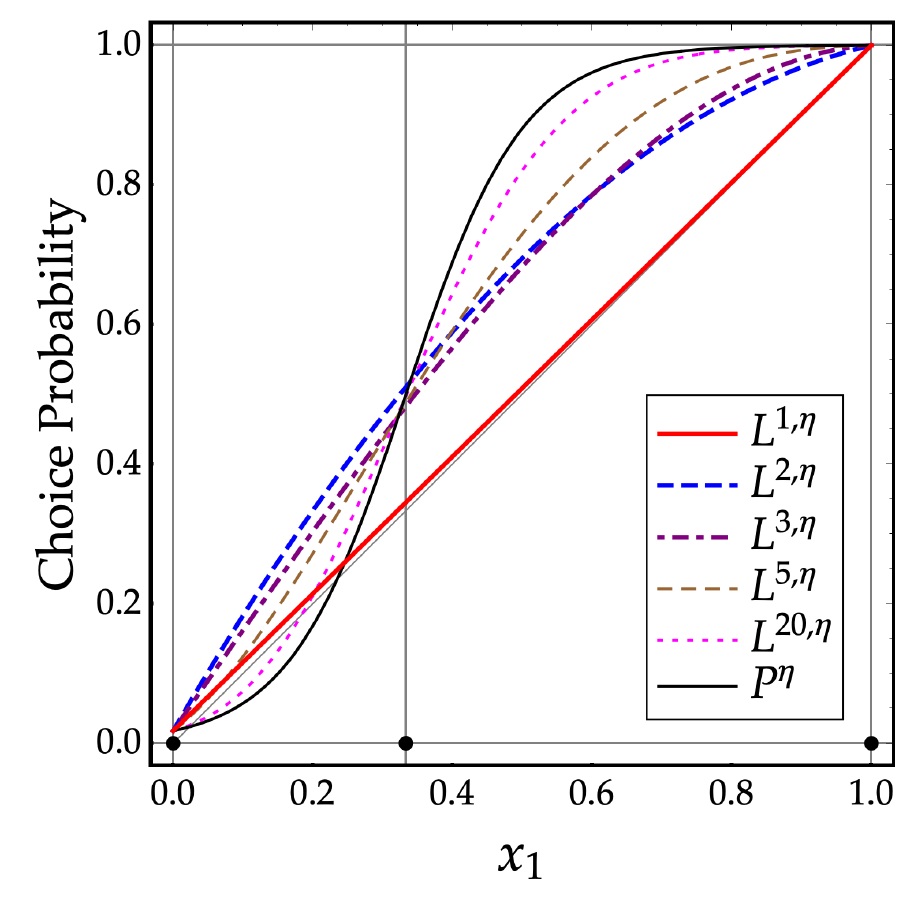}
        \caption{$L_1^{k,\eta}$ for different $k$ and $P_1^{\eta}$ ($\eta = 0.25$)}
        \label{fig:logits}
    \end{subfigure}

    \caption{Choice probability of action $1$ in the game \eqref{eq:A.2x2.coord-ex} under different rules. Black markers on the $x_1$ axis indicate the Nash equilibrium values.} 
    \label{fig:choice-probs}
\end{figure}
\begin{example}
\Cref{fig:choice-probs} compares the choice probability of action $1$ under different choice rules for the following linear two-action population game: 
\begin{align}
    &  
    A = 
    \begin{bmatrix}
        2 & 0 \\ 
        0 & 1 
    \end{bmatrix}.   
    \label{eq:A.2x2.coord-ex}
\end{align}
The Nash equilibria satisfy $x_1 \in \{0,1/3, 1\}$. 
This is a special instance of the coordination game \eqref{eq:A.2x2.coord} to be discussed in \cref{sec:small-k-selection}. 
For a fixed $k$, $L^{k,\eta}$ tends toward $\BR^k$ as $\eta$ decreases. 
Likewise, for a fixed $\eta$, it tends toward $P^\eta$ as $k$ increases. 
For each curve, the intersections with the 45-degree line are the fixed points for the corresponding case. 
\end{example}

\section{Large \texorpdfstring{$k$}{k} approximation}
\label{sec:approximation}

For fixed $\eta > 0$, as the sample size $k$ grows, the sampling logit rule approaches the full-information logit rule. 
This section examines the role of sampling in producing systematic distortion in this regime. 
A second-order approximation identifies the two endogenous premiums discussed in the introduction, variance and curvature premiums, and translates them into a virtual payoff representation. 
Section~\ref{sec:joint-limit} considers the critical scale at which sampling and logit noise remain comparable, and Section~\ref{sec:examples} considers the complementary regime in which $k$ is fixed and small. 

Below, we suppress $k$ and $\eta$ from notation whenever they are understood from the context: $L = L^{k,\eta}$, $P = P^\eta$, and so on. 
We use Landau's order notation to summarize approximation errors and convergence rates. 
In a stated limit, $f=O(g)$ means that $|f/g|$ remains bounded, whereas $f=o(g)$ means that $f/g\to0$, so $f$ becomes negligible relative to $g$. 
For vector-valued functions, absolute value is replaced by any fixed norm, and a statement that is uniform over $X$ applies to the corresponding supremum over $X$.

\subsection{Approximation and virtual payoffs}

We employ the \emph{delta method} to approximate the expectation over samples for large $k$ \citep[see, e.g.,][Ch.3]{vanderVarrt-Book2000}. 
For the empirical population state $w = \frac{1}{k}z$, we have $\BbE[w]=x$ and $\Var[w]=\frac{1}{k}\Sigma(x)$, where $\Sigma(x) \Is \diag{x}-xx^\top$ is the covariance matrix for $\Mult(k\!\mid\!x)$. 
The delta method then allows us to approximate $L(x) = \BbE[P(w)]$ by $x$, $\Sigma(x)$, and the derivatives of $P$ and $F$. 

We introduce some convenient notation. 
For a collection $\{y_i\}_{i\in S}$ of scalar-, vector-, or matrix-valued functions defined on $X$, let their logit-weighted average $\OL y$ and their corresponding deviations $\{\WH y_i\}_{i\in S}$ be 
\begin{align}
    & \OL{y}(x) \Is \sum_{l\in S} P_l(x) y_l(x) 
    \quad \text{and} \quad 
    \widehat{y}_i(x) \Is y_i(x) - \OL{y}(x)
    & \forall x \in X. 
    \label{eq:bar-hat} 
\end{align}
Note that the logit weights are state dependent. 
By construction, hatted variables satisfy $\sum_{i\in S}P_i(x)\,\WH y_i(x) = \sum_{i\in S}P_i(x)\bigl(y_i(x)-\OL y(x)\bigr) =0$. 
Also, given matrices $A$ and $B$, let $\la A, B \ra \Is \trace[A^\top B] = \sum_{k,l} a_{kl} b_{kl}$ where $a_{kl}$ and $b_{kl}$ denote the entries in the $k$th row and $l$th column of $A$ and $B$, respectively. 

With this notation, the second-order delta method yields the following approximation of $L$. All proofs are in \cref{app:derivations}. 
\begin{theorem}[Approximated choice rule]
\label{thm:approximation}
Assume that $F$ admits a $C^2$ extension to an open neighborhood of $X$. 
The $(k,\eta)$-sampling logit choice rule is approximated by the following corrected $\eta$-logit rule: 
\begin{align}
    & \TlL_i(x) 
    = \left(1 + \widehat v_i(x) + \widehat q_i(x) \right) P_i(x) 
    && \forall i \in S,\ x \in X,  
    \label{eq:TlL}
\end{align}
where the functions $v: X \to \mathbb R_{\ge 0}^n$ and $q: X \to \mathbb R^n$ are defined by  
\begin{align}
    v_i(x) 
    &= 
    \frac{1}{2k\eta^2} 
    \; 
    \WH{F_i'\,}(x)^\top \Sigma(x)\; \WH{F_i'\,}(x)
    \quad
    \text{and}
    \quad 
    q_i(x) 
    = 
    \frac{1}{2k\eta} 
    \langle F_i''(x), \Sigma(x) \rangle
    \label{eq:def.delta} 
\end{align} 
with $\Sigma(x) = \diag{x} - x x ^\top$. 
The approximation formula satisfies the following: 
\begin{enumerate}
\item \label{enum:validity}
For any $\eta > 0$, there is $k^* < \infty$ such that $\TlL(x) \in X$ for all $x \in X$ and all $k > k^*$. 
That is, the approximated rule $\TlL$ is a valid choice rule for sufficiently large $k$. 
\item \label{enum:error-bound} 
If $F$ admits a $C^4$ extension to an open neighborhood of $X$, then
\begin{align*}
    L
    =
    \TlL
    +O\!\left(\frac{1}{k^2}\sum_{m=1}^4\eta^{-m}\right)
\end{align*}
uniformly on $X$. 
More explicitly, for any fixed norm, there is a constant $C>0$, independent of $k$ and $\eta$, such that
\begin{align}
& 
    \Delta
    \Is
    \sup_{x\in X}
    \|L(x)-\TlL(x)\|
    \le
    \begin{cases}
        Ck^{-2}\eta^{-4} & \text{if } 0 < \eta \le 1 \\ 
        Ck^{-2}\eta^{-1} & \text{if } \eta > 1
    \end{cases}
&
    \forall k\ge1.
    \label{eq:theorem-error-bound}
\end{align}
\end{enumerate}
\end{theorem}

\Cref{sec:origins-of-premiums} discusses the intuition behind the correction terms $v$ and $q$ in detail. 

The approximation is a large-$k$ result.  
For fixed $\eta$, $\TlL-P=O(k^{-1})$ and $L-\TlL=O(k^{-2})$, uniformly on $X$, so the approximation can be made as accurate as desired by taking $k$ large. 
For $0<\eta\le1$, part \cref{enum:error-bound} gives an error bound of order $k^{-2}\eta^{-4}$, whereas the nominal scale for the correction terms is $(k\eta^2)^{-1}$, so that their ratio is of order $(k\eta^2)^{-1}$. 
Hence, along sequences for which $k\eta^2\to\infty$,
\begin{align*}
    \sup_{x\in X}\|L(x)-\TlL(x)\|
    =o\bigl((k\eta^2)^{-1}\bigr),
\end{align*}
meaning that the remainder becomes negligible relative to the nominal scale of the correction terms. 
The proof of part \cref{enum:validity} also shows that one may choose $k^*(\eta)=O(\eta^{-2})$ as $\eta\Toz$ for $\TlL$ to be a valid choice rule. 
For example, $\eta^{-2}=100$ at $\eta=0.1$ means that samples in the hundreds are the relevant starting point for the approximation to be uniformly accurate on $X$. 

The trade-off implicit in this scaling has a natural interpretation.
In unfamiliar or complex environments, logit noise is plausibly high, so the approximation can become informative with a more modest sample size.
In familiar environments, by contrast, agents may respond more precisely, corresponding to a lower $\eta$, but repeated play or widely available information may also provide the much larger effective sample required to compensate.
Thus, the two empirically plausible combinations, large $\eta$ with smaller $k$, and small $\eta$ with larger $k$, are consistent with the restriction underlying the approximation.

Provided that $k$ is sufficiently large, \cref{thm:approximation} yields an interpretation of SLE based on \emph{virtual payoffs} capturing agents' effective decision biases in aggregate, in the spirit of \citet[Appendix]{Hofbauer-Sandholm-JET2007}. 
\begin{observation}[Virtual payoff representations]
\label{obs:TlL-interpretation}
Fix $\eta > 0$. Suppose $k$ is sufficiently large that $\min_{x \in X}\min_{i\in S} \{ 1 + \WH{v_i}(x) + \WH{q_i}(x) \} > 0$, ensuring that $\TlL(x) \in X$ for all $x \in X$ and that the logarithm defining $\TlG$ below is well defined.
Then, equilibria under the approximated choice rule $\TlL$ are equivalently represented as: 
\begin{enumerate}[label={(\alph*)},leftmargin=2em]
\item \label{enum:virtual-game-1}
the $\eta$-logit equilibria for the ``virtual'' population game $\TlF = F + \TlG$, where 
\begin{align}
    & \TlG_i(x) 
    = \eta \log \left(1 + \WH{v}_i(x) + \WH{q}_i(x) \right) 
    & \forall i \in S, x \in X 
\end{align}
is the deterministic payoff distortion that encapsulates the role of sampling noise, or 
\item \label{enum:virtual-game-2}
the Nash equilibria for the ``virtual'' population game $F + \TlG + H$, where $\TlG$ is the same as above and $H_i(x) \Is - \eta \log (x_i)$ is a deterministic representation of logit noise. 
\qedhere 
\end{enumerate}
\end{observation}

To see \cref{obs:TlL-interpretation}~\ref{enum:virtual-game-1}, we can rearrange 
\begin{align}
    \TlL_i(x) 
    & = 
    \frac
        {\exp\big(\eta^{-1} \TlF_i(x)\big)}
        {\sum_{j \in S} \exp\big(\eta^{-1} \TlF_j(x)\big)} 
    & \forall i \in S, x \in X
    \label{eq:TlL-logit}
\end{align}
as long as $\TlG$ is well defined. 
Thus, an equilibrium under $\TlL$, namely $x = \TlL(x)$, is nothing but an $\eta$-logit equilibrium for the virtual population game $\TlF$. 
\cref{obs:TlL-interpretation}~\ref{enum:virtual-game-2} then follows from known results on logit choice. 

Before interpreting $v$ and $q$, we examine the closeness of fixed points of the approximated rule $\TlL$ to true SLE. 
Naturally, when $\TlL$ is sufficiently close to $L$, the equilibria under the two rules should be close to each other. 

Some definitions are in order to state this observation rigorously. 
For a continuous choice rule $c:X\to X$, let
$\Fix(c)\Is\{x\in X:x=c(x)\}$ denote its fixed-point set: $\Fix(L)$ is the set of SLE, whereas
$\Fix(\TlL)$ is the set of stationary points under $\TlL$. 
For a fixed norm $\|\cdot\|$, write $d(x,A)\Is\inf_{a\in A}\|x-a\|$ and let
$d_{\RmH}$ denote the corresponding Hausdorff distance $d_{\RmH}(A,B) = \max\{\sup_{x \in A}d(x,B), \sup_{y\in B} d(y, A) \}$ for sets $A, B$. 
We call $x \in \Fix(\TlL)$ \emph{regular} if $I-\TlL_T'(x)$ is nonsingular on the tangent space $T= \{h\in\BbR^n:\sum_{i\in S} h_i=0\}$ of the affine hull of $X$, where $\TlL_T'$ is interpreted as the derivative of $\TlL$ restricted to $T$.
\begin{theorem}[Accuracy of approximated equilibria]
\label{thm:equilibrium-distortion-general}
Fix $\eta>0$. 
Take $k$ sufficiently large that $\TlL$ is a valid choice rule. 
Set $\Delta = \sup_{x\in X}\|L(x)-\TlL(x)\|$.
Then, the following statements hold.
\begin{enumerate}
    \item 
    \label{enum:equilibrium-bound-1} 
    There is
    $\Tlomega:\BbR_{\ge0}\to\BbR_{\ge0}$, depending on
    $\TlL$, such that
    $\Tlomega(t)\Toz$ as $t\Toz$ and
    \begin{align}
        \sup_{x\in\Fix(L)}
        d\bigl(x,\Fix(\TlL)\bigr)
        \le \Tlomega(\Delta).
        \label{eq:one-sided-equilibrium-bound}
    \end{align}
    \item 
    \label{enum:equilibrium-bound-2}
    Suppose that
    $\Fix(\TlL)=\{\Tlx^1,\ldots,\Tlx^m\}$ is finite and
    every $\Tlx^r$ is regular. 
    Let 
    \begin{align}
        \gamma
        \Is
        \min_{r \in \{1,\hdots,m\}}
        \ 
        \inf_{h\in T, \|h\|=1}
        \|(I-\TlL_T'(\Tlx^r)) h\| > 0. 
    \end{align}
    Then, there is a constant $\overline\Delta>0$ such that, if $\Delta<\overline\Delta$,
    $L$ has at least $m$ distinct fixed points that can be labeled
    $x^1,\ldots,x^m$ so that
    \begin{align}
        \max_{r \in \{1,\hdots,m\}}\|x^r-\Tlx^r\|
        \le \frac{2\Delta}{\gamma}.
        \label{eq:equilibrium-shift-bound}
    \end{align}
    If furthermore $\Fix(L)$ also consists of exactly $m$ points, these are all its fixed points and
    \begin{align}
        d_{\RmH}\bigl(\Fix(L),\Fix(\TlL)\bigr)
        \le
        \max_{r \in \{1,\hdots,m\}}\|x^r-\Tlx^r\|
        \le \frac{2\Delta}{\gamma}. 
    \end{align}
\end{enumerate}
\end{theorem}

The two parts address opposite directions of approximation. 
Part~\ref{enum:equilibrium-bound-1} takes a true SLE $x$ satisfying $x=L(x)$ as the starting point. 
Evaluating the approximated rule at this equilibrium, we see $\|x -\TlL(x)\| = \|L(x)-\TlL(x)\|\le\Delta$ by definition. 
That is, every SLE $x$ has a fixed-point error of at most $\Delta$ for the approximated rule. 
Thus, every true SLE lies near an equilibrium of the approximated rule, or equivalently near a logit equilibrium of the virtual game. 
This conclusion is one-sided in the sense that $\TlL$ may have other fixed points with no nearby SLE.

Part~\ref{enum:equilibrium-bound-2} takes the regular equilibria of the approximated rule $\TlL$ as the references. 
Regularity implies that the fixed-point equation crosses zero rather than merely touching it, so a small perturbation only moves the equilibrium without eliminating it. 
Accordingly, when the approximation of the choice rule is sufficiently accurate (when $\Delta$ is small), a distinct SLE $x^r$ should exist near each regular equilibrium $\Tlx^r$.\footnote{The quantity $\gamma$ measures the robustness of the least robust approximated equilibrium. A larger $\gamma$ yields a tighter bound for the possible location of nearby SLE, whereas a small $\gamma$ indicates greater sensitivity to approximation error. 
For instance, near a bifurcation of the equilibrium, $\gamma$ can approach zero, and the bound \eqref{eq:equilibrium-shift-bound} correspondingly deteriorates.}  
When the two rules also have the same number of equilibria, the bound applies to the two fixed-point sets in both directions.

Summing up, \cref{thm:equilibrium-distortion-general} formalizes the natural idea that an accurate approximation of the choice rule yields an accurate approximation of its equilibria.

\subsection{Erroneous payoff inference and payoff distortions}
\label{sec:origins-of-premiums}

We now turn our attention to the origins of the correction terms $v$ and $q$. 
\cref{thm:approximation} shows that actions with relatively higher values of $v$ and/or $q$ become endogenously exaggerated in aggregate behavior. 
In the language of \cref{obs:TlL-interpretation}, they are ``virtually'' preferred by agents beyond what the true payoffs would imply. 

Naturally, $v$ and $q$ originate from finite-sample variability of inferred payoffs. 
The payoffs inferred from the empirical distribution $w = \frac{1}{k} z$ can be expanded as 
\begin{align}
    F_i(w) 
    \approx 
    F_i(x) 
    + 
    \underbrace{
        \la F_i'(x), 
        w - x
        \ra 
    }_{\mathclap{\text{\small First-order effect: $\zeta_{1i}$}}}
    \; + \;
    \underbrace{\tfrac{1}{2} (w - x)^\top F_i''(x) (w - x)}_{\text{\small Second-order effect: $\zeta_{2i}$}}.
    \label{eq:F-errors}
\end{align}
We have $\BbE[w - x] = \Vt0$ and $\Var[w - x] = \frac{1}{k}\Sigma(x)$.  
In expectation, the first-order effect vanishes because $\BbE[\zeta_{1i}] = 0$.
The second-order effect, however, may have a nonzero mean. 
In fact, 
\begin{align}
    \BbE[ 
    F_i(w)
    -  
    F_i(x) 
    ] 
    \approx 
    \BbE[\zeta_{2i}]
    =
    \frac{1}{2}\BbE\left[(w - x)^\top F_i''(x)(w - x)\right]
    = 
    \frac{1}{2k}\la F_i''(x), \Sigma(x) \ra, 
    \label{eq:q-E}
\end{align}
from which we see the second correction term $q$ in \cref{eq:def.delta} satisfies 
\begin{align}
    q_i(x) = \frac{1}{\eta} \BbE[\zeta_{2i}]. 
\end{align}
Intuitively, $q_i$ represents a Jensen-type ``lucky-draw'' effect. 
For instance, consider small symmetric perturbations for which $w - x = \pm \nu$ with equal probability $1/2$.
Then, 
\begin{align}
    \BbE[F_i(w) - F_i(x)]
    & = 
    \frac{F_i(x + \nu) + F_i(x - \nu)}{2}
    - F_i(x) 
    \\
    & 
    = 
    \frac{  
    \left(F_i(x + \nu) - F_i(x)\right)
    - 
    \left(F_i(x) - F_i(x - \nu)\right)
    }{2}
    \\
    &
    \approx
    \frac{1}{2}\nu^\top F_i''(x)\nu.
    \label{eq:q-origin}
\end{align}
Thus, positive curvature along a direction in which samples vary makes \cref{eq:q-origin} positive: symmetric sampling errors raise the inferred payoff on average if $F_i''$ is positive definite at the state. 
Negative curvature lowers it instead, and linearity eliminates it. 
This motivates calling $q(x)$ the \emph{curvature premium}.\footnote{The importance of curvature of choice rules in shaping evolutionary incentives also appears in \cite{Hofbauer-Weibull-JET1996,Viossat-ETB2015}, where convex or concave transformations of payoffs generate qualitatively different evolutionary behavior.} 
The symmetric two-point case is only illustrative, as the second-order effect is given by \cref{eq:q-E} for any zero-mean sampling error. 

A second lucky-draw effect arises due to logit choice, even when the inferred payoff is unbiased, $\BbE[\zeta] = \BbE[F_i(w) - F_i(x)] = 0$. 
To see this, let $p(\mu)=\exp(\eta^{-1}\mu)$. 
For a small error $\zeta$ with $\BbE[\zeta] = 0$, convexity of $p$ implies\footnote{Concretely, from $p(\mu + \zeta) \approx p(\mu) + p'(\mu) \zeta + \frac{1}{2} p''(\mu) \zeta^2$, up to the second order, $\BbE[p(\mu + \zeta)] \approx p(\mu) + p'(\mu) \BbE[\zeta] + \frac{1}{2}p''(\mu) \BbE[\zeta^2] \approx (1 + \frac{1}{\eta} \BbE[\zeta] + \frac{1}{2 \eta^2} \Var[\zeta]) \cdot p(\mu)$ where $p''(\mu) = \frac{1}{\eta^2} p(\mu) > 0$.}
\begin{align}
    \frac{\BbE[p(\mu + \zeta)] - p(\mu)}{p(\mu)}
    \approx
    \frac{1}{2\eta^2}\Var[\zeta]
    > 0.
    \label{eq:error-Var}
\end{align}
An upward error raises $p$ more than an equally sized downward error lowers it, and the magnitude of this effect is measured by the variance of errors. 
In fact, basic statistical identities imply
\begin{align}
    v_i(x)
    &=
    \frac{1}{2\eta^2}\Var\bigl[\widehat{\zeta}_{1i}\bigr], 
    \label{eq:v-Var}
\end{align}
where we define the \emph{relative} first-order effects by
\begin{align}
    \widehat{\zeta}_{1i}
    \Is
    \zeta_{1i} - \sum_{l\in S}P_l(x)\zeta_{1l}
    =
    \bigl\la \WH{F_i'}(x), w - x\bigr\ra.
\end{align}
The relative values matter because logit probabilities depend only on relative payoffs. 
Hence, $v_i(x)$ measures the amplification generated by the variability of inferred relative payoffs, which motivates calling $v(x)$ the \emph{variance premium}.

For small $\eta$, the logit rule approximates the best response and agents' choice behavior becomes more sensitive to small payoff differences, as reflected by the factors of $\eta$ in the denominators of $v$ and $q$. 
For fixed $\eta$, expanding the logarithm in \cref{obs:TlL-interpretation} gives
\begin{align*}
    \TlG_i(x)
    =
    \eta\bigl(\WH v_i(x)+\WH q_i(x)\bigr)
    +O(k^{-2})
\end{align*}
uniformly on $X$. 
From \cref{eq:def.delta}, the scales of the two leading terms $\eta\WH v_i$ and $\eta\WH q_i$ are $O((k\eta)^{-1})$ and $O(k^{-1})$, respectively. 
Thus, when $k\eta^2\to\infty$, the variance contribution can be larger by a factor of order $\eta^{-1}$ for small $\eta<1$.  
This difference reflects the source of each correction in \cref{eq:F-errors}. 
The variance premium arises from the leading effect of erroneous payoff evaluation ($\zeta_{1i}$), and the curvature premium is the second-order correction ($\zeta_{2i}$). 

The curvature premium arises when sampling variability is translated into inferred payoffs, whereas the variance premium arises when inferred relative-payoff variability is translated into exponential weights in logit choice. 
Their logit-relative values $\{\WH q_i\}$ and $\{\WH v_i\}$ determine the net correction to the final choice probability because the logit formula normalizes these weights across actions.

\subsection{Variance premium in linear population games}

As a concrete example of the variance premium, consider a linear population game $F(x) = Ax$ in which the curvature premium is absent because $F''(x) = O$. 
The variance premium is given by 
\begin{align}
    v_i(x) = \frac{1}{2k\eta^2} \sigma_i(x)
    \quad\text{with}\quad 
    \sigma_i(x)
    \Is 
    \WH{a}_i(x)^\top \Sigma(x) \WH{a}_i(x) 
    = 
    k \Var[ \la \WH{F_i'}(x), w - x\ra ], 
    \label{eq:v-vector-form}
\end{align}
where $a_i \Is F_i'(x)$ is the $i$th \emph{row} of $A$ as a column vector and $\WH{a}_i(x) = a_i - \OL{a}(x) = a_i - \sum_{l \in S} P_l(x) a_l$. 
We see $\sigma_i(x)$ is the variance of action $i$'s relative marginal payoffs $\{\WH{a}_{il}(x)\}_{l\in S}$ for a single draw of the opponent's play $l\in S$ at state $x$: 
\begin{align}
    \sigma_i(x)
    = 
    \sum_{l \in S} \left(\WH{a_i}_l(x) \right) ^2 x_l
    - 
    \Bigl(
    \sum_{l \in S} 
    \WH{a_i}_l(x) x_l 
    \Bigr)^2. 
\end{align}
Thus, in linear population games, the approximated choice rule $\TlL$ assigns higher probability to actions that have higher variance $\sigma_i(x)$ of relative marginal payoffs. 
The following corollary of \cref{thm:approximation} provides a concrete example. 
\begin{corollary}
\label{cor:2x2-delta} 
Consider a two-action linear population game generated by 
\begin{align}
    & A 
    = 
    \begin{bmatrix}
        a & b \\
        c & d
    \end{bmatrix}
    \quad \text{with} \quad 
    \lambda \equiv (a - c) - (b - d) \ne 0. 
    \label{eq:A-2x2}
\end{align} 
Let $i,j\in S = \{1,2\}$ with the convention $i \ne j$. 
Then, $\sigma_i(x) = P_j(x)^2 \cdot \sigma_A(x)$,  
where $\sigma_A(x) 
    = \lambda^2 x_1 x_2\ge 0$ with $\lambda = (a - c) - (b - d) \ne 0$. 
In turn, the variance premium satisfies 
\begin{align}
    & 
    \WH{v}_i(x)
    = 
    \frac{1}{2k\eta^2}
    \; 
    \WH{\sigma_i}(x) 
= 
    \frac{1}{2k\eta^2}
P_j(x)
    \left(1 - 2 P_i(x)\right)\sigma_A(x)
& 
\forall x \in X. 
    \label{eq:Ax.D-sigma}
\end{align} 
\end{corollary}

Here, $\sigma_A$ is the one-observation variance of the estimated relative payoff: 
\begin{align}
    \Var[F_1(w) - F_2(w)] 
    = 
    \Var[(b - d) + \lambda w_1] 
    = 
    \frac{1}{k} \sigma_A(x). 
\end{align}
This vanishes at $x_1 \in \{0,1\}$ since $\sigma_A(x) = 0$: any sample drawn at a pure population state allows agents to infer the state correctly.  
If the population state is more balanced, inference based on samples fluctuates more, and the resulting bias $\WH{v}$ tends to become large. 
This is reflected by $\sigma_A(x)$ attaining its maximum at $x_1 = 1/2$. 
Also, for $x_1 \in (0,1)$, $\WH{v}_1(x) = \WH{v}_2(x) = 0$ if and only if $x$ is the mixed Nash equilibrium.

\cref{cor:2x2-delta} shows that there is a \emph{virtual preference for the suboptimal}. 
Specifically, 
\begin{align}
    & 
    \WH{v}_i (x) 
    \gtreqless 0 
    \quad \Leftrightarrow \quad 
P_i(x)
    \lesseqgtr 
    \frac{1}{2} 
    \quad \Leftrightarrow \quad 
    F_i(x) \lesseqgtr F_j(x), 
    & 
    \forall i,j\in S, i \ne j
    \label{eq:2x2-delta-prop}
\end{align} 
for all $x \notin \{e_1,e_2\}$, 
with the same sign for the inequalities. 
Thus, the population behaves \emph{as if} agents prefer the suboptimal option.  
The $\eta$-logit choice rule already possesses this property because it assigns a positive choice probability to the suboptimal action. 
\Cref{eq:2x2-delta-prop} indicates that sampling errors amplify this bias.\footnote{\citet{Danenberg-Spiegler-JPEM2025} also discuss an inferior-alternative bias in two-action environments in a different stochastic choice environment.}

\begin{example}
\Cref{fig:sigma-eta} plots $\Htsigma = 2k\eta^2 \Htv$ for the case $A = \begin{bsmallmatrix} 2 & 0 \\ 0 & 1 \end{bsmallmatrix}$ with different $\eta$. 
We see $\WH{\sigma}_1(x) > 0$ if and only if $x_1 \in(0,1/3)$ where $\br(x) = \{2\}$ and $\WH{\sigma}_2(x) > 0$ if and only if $x_1 \in (1/3,1)$ where $\br(x) = \{1\}$, confirming the preference for the suboptimal. 
We also note that, for $x \ne x^*$, $\lim_{\eta \Toz} \;\WH{\sigma}_i(x) = 0$ if $i \in \operatorname{br}(x)$ and $\lim_{\eta \Toz} \;\WH{\sigma}_i(x) = \sigma_A(x)$ if $i \notin \operatorname{br}(x)$. 
\end{example}

\begin{figure}[t]
    \centering
	\begin{subfigure}[b]{.32\hsize}
		\centering
        \includegraphics[width=\hsize]{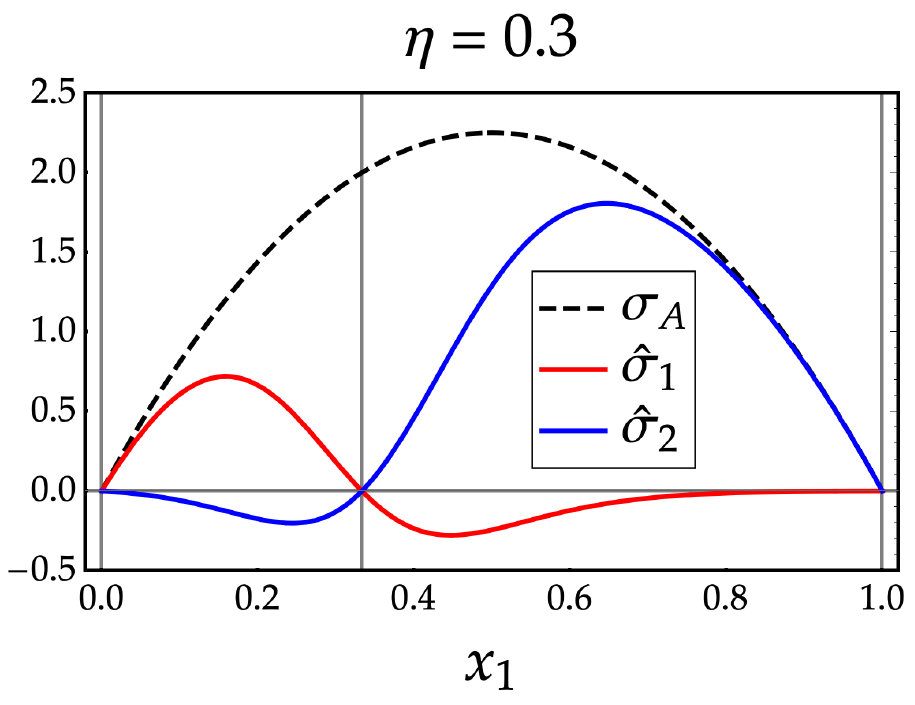}
    \end{subfigure}
	\begin{subfigure}[b]{.32\hsize}
		\centering
        \includegraphics[width=\hsize]{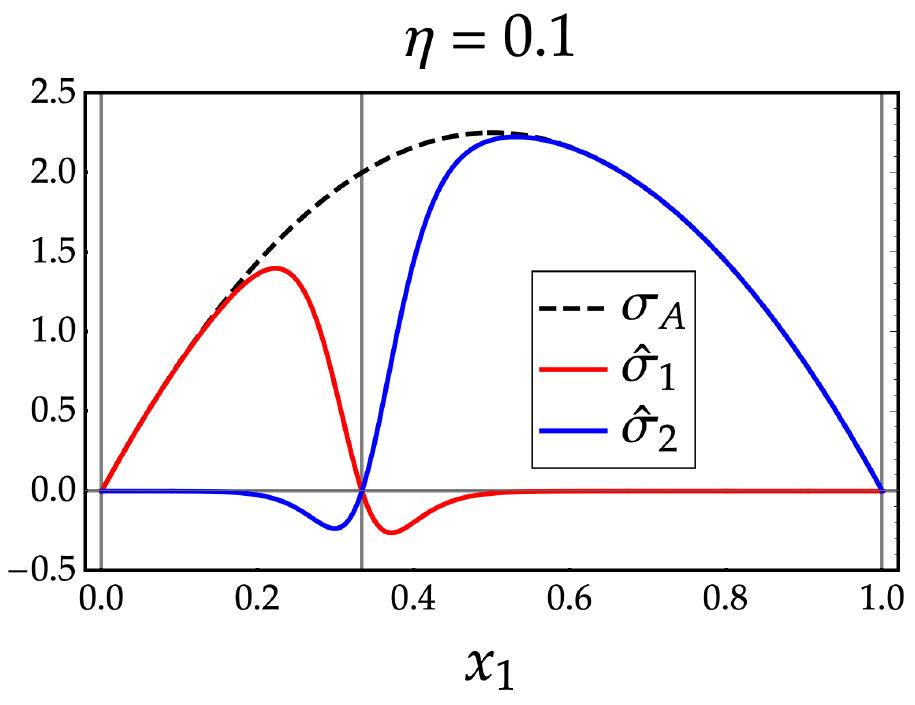}
    \end{subfigure}
	\begin{subfigure}[b]{.32\hsize}
		\centering
        \includegraphics[width=\hsize]{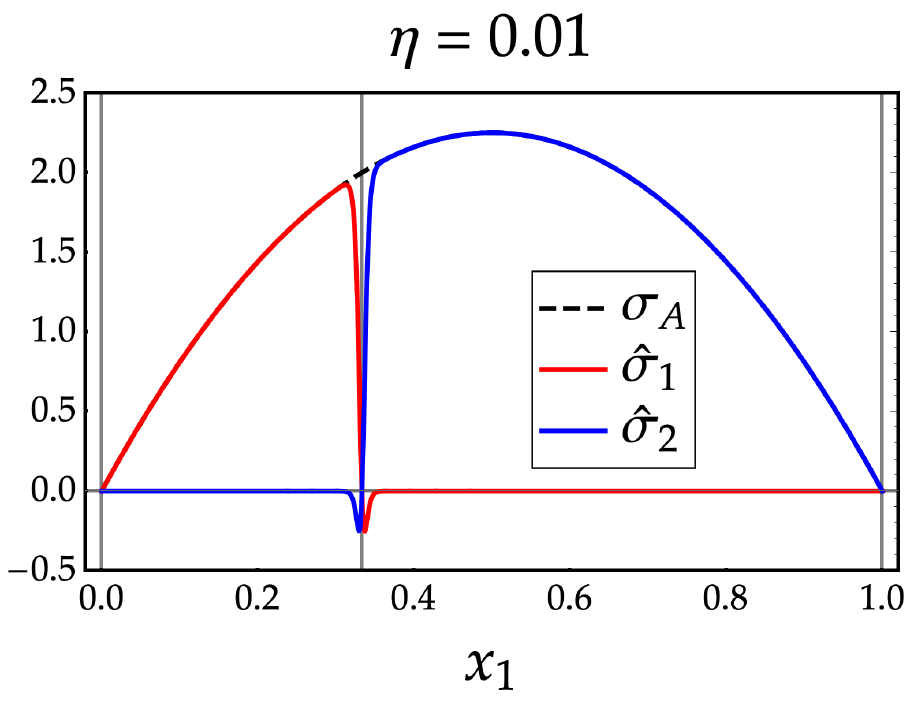}
    \end{subfigure}

    \caption{The vertical axes report $\WH{\sigma}_i(x) = 2 k \eta^2 \;\WH v_i(x)$ for the case $A = \begin{bsmallmatrix} 2 & 0 \\ 0 & 1 \end{bsmallmatrix}$ with different $\eta$. The interior Nash equilibrium of the game is $x_1^* \Is \frac{1}{3}$, and $\br(x) = \{2\}$ if $x_1 < x_1^*$ and $\br(x) = \{1\}$ if $x_1 > x_1^*$.} \label{fig:sigma-eta}
\end{figure}

We next characterize the interior fixed points for $\TlL$ and $L$ in the two-action game \eqref{eq:A-2x2}. 
Suppose there is a unique interior equilibrium 
\begin{align}
    & x_1 = x^* \Is \frac{d - b}{\lambda} \in (0,1) 
    \quad \text{with} \quad \lambda = (a - c) - (b - d) \ne 0. 
\end{align}
If $\lambda > 0$, the game is a coordination game whose two pure states are strict Nash equilibria. 
If $\lambda < 0$, it is a (strict) \emph{stable game} (or \emph{contractive game}) \citep{Hofbauer-Sandholm-JET2009}, and $x^*$ is its unique Nash equilibrium. 
We focus on the generic asymmetric case $x^* < 1/2$ without loss of generality. 
\begin{proposition}
\label{prop:approx-coord}
Consider a two-action linear population game $F(x) = Ax$ with $A$ in \cref{eq:A-2x2}. 
Suppose that $0 < x^* < 1/2$, and write $\theta \Is 1/(2k\eta^2)$. 
Treating $\eta$ and $\theta$ as independent perturbation parameters, for all sufficiently small $\eta>0$ and $\theta\ge0$, $\TlL$ has a locally unique interior fixed point $\Tlx(\eta,\theta)$ smoothly converging to $x^*$ as $(\eta,\theta)\to(0,0)$. Moreover,
\begin{align}
    \Tlx(\eta,\theta)-\Tlx(\eta,0)
    =
    - \frac{\eta\theta}{\lambda}\sigma_A(x^*)(1-2x^*)
    +o(\eta\theta) 
    \label{eq:approx-coord-displacement}
\end{align}
where $\sigma_A(x^*)>0$ is the relative-payoff variance from \cref{cor:2x2-delta}. 
Along positive pairs $(\eta,\theta)\to(0,0)$ for which $k=(2\theta\eta^2)^{-1}\in\mathbb N$, there is also a true SLE $\Brx(\eta,\theta)=L(\Brx(\eta,\theta))$ satisfying
\begin{align}
    |\Brx(\eta,\theta)-\Tlx(\eta,\theta)|
    =o(\eta\theta).
    \label{eq:approx-coord-true-distance}
\end{align}
For all sufficiently small admissible pairs, both $\Tlx(\eta,\theta)$ and $\Brx(\eta,\theta)$ lie below $x^*$ if $\lambda>0$ and above $x^*$ if $\lambda<0$.
\end{proposition}

Setting $\theta=0$ removes the sampling correction, so $\Tlx(\eta,0)$ is the corresponding full-information $\eta$-logit equilibrium. 
The difference in \cref{eq:approx-coord-displacement} therefore isolates the effect of sampling noise from the equilibrium displacement already caused by logit noise. 
Since $\sigma_A(x^*)>0$ and $1-2x^*>0$, sampling shifts the interior fixed point toward $x_1=0$ when $\lambda>0$ and toward $x_1=1$ when $\lambda<0$. 
This direction agrees with the preference for the suboptimal action described above. 
By \cref{eq:approx-coord-true-distance}, the true SLE inherits the same leading-order displacement, and hence the same direction. 

\subsection{Curvature premium in separable games}
\label{sec:curvature}

As a parsimonious example of the curvature premium, consider a \emph{separable game} where each $F_i$ depends only on $x_i$ and can be represented as $F_i(x) = f_i(x_i)$. 
This class of games yields 
\begin{align}
    & 
    q_i(x) = \frac{1}{2k\eta} \la F_i''(x), \Sigma(x) \ra = 
    \frac{1}{2k\eta} \cdot f_i''(x_i) \cdot \Sigma_{ii}(x) = 
    \frac{1}{2k\eta} \cdot f_i''(x_i) \cdot x_i (1 - x_i), 
    \label{eq:delta-with-curvature}
\end{align}
since $[F_i''(x)]_{ij} = 0$ unless $i = j$. 
Other things being equal, at each $x$, agents behave \emph{as if} they prefer actions with (i) locally ``more convex'' payoffs because of $f_i''(x_i)$ and/or (ii) relatively high but not too high popularity because of $x_i (1 - x_i)$. 
The second is a property shared with variance premium and reflects sample fluctuations. 
Due to the first property, however, the sign of $\WH{v}$ and $\WH{q}$ can be different, reflecting their distinct origins.

For separable two-action games, \cref{thm:approximation} implies the following. 
\begin{corollary}
    \label{cor:separable.2}
For a separable two-action game satisfying $F_i(x) = f_i(x_i)$,  
\begin{align}
    \WH{v}_1(x) 
    & 
    = 
    \frac{1}{2k\eta^2} P_2(x)
    \left(1 - 2P_1(x)\right) 
    \sigma_F(x)
\quad \text{and}
    \\ 
    \WH{q}_1(x) 
    & 
    = 
    \frac{1}{2k\eta} 
    P_2(x)
    (f_1''(x_1) - f_2''(x_2)) x_1 (1 - x_1),
    \label{eq:2x2-curvature-premium}
\end{align}
where $\sigma_F(x) = \left(f_1'(x_1) + f_2'(x_2)\right)^2 x_1 x_2$. 
\end{corollary}

\begin{figure}[t]
    \centering
	\begin{subfigure}[b]{.32\hsize}
		\centering
        \includegraphics[width=\hsize]{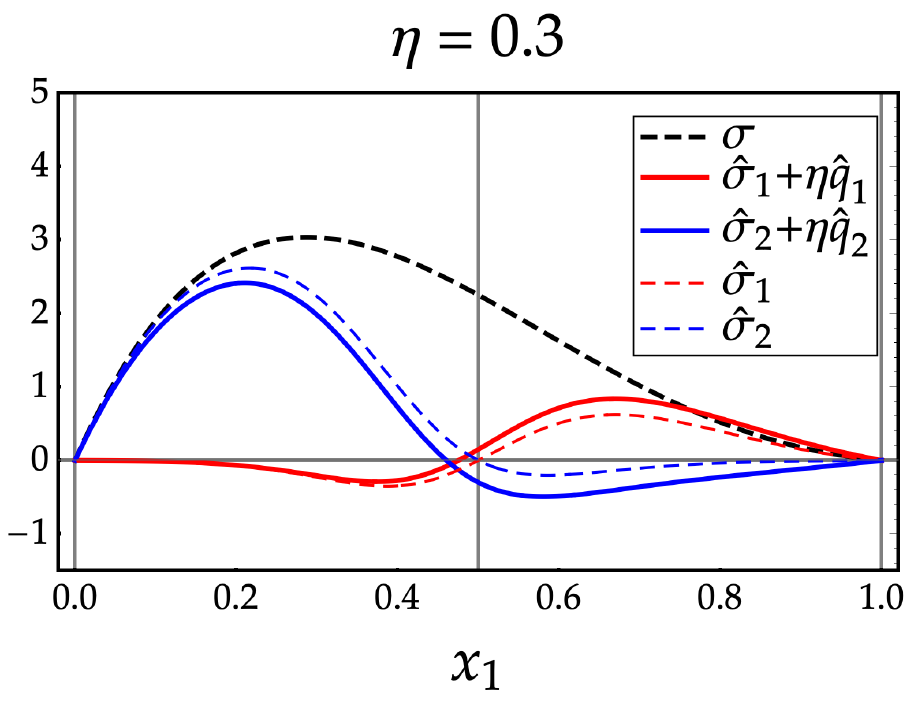}
    \end{subfigure}
	\begin{subfigure}[b]{.32\hsize}
		\centering
        \includegraphics[width=\hsize]{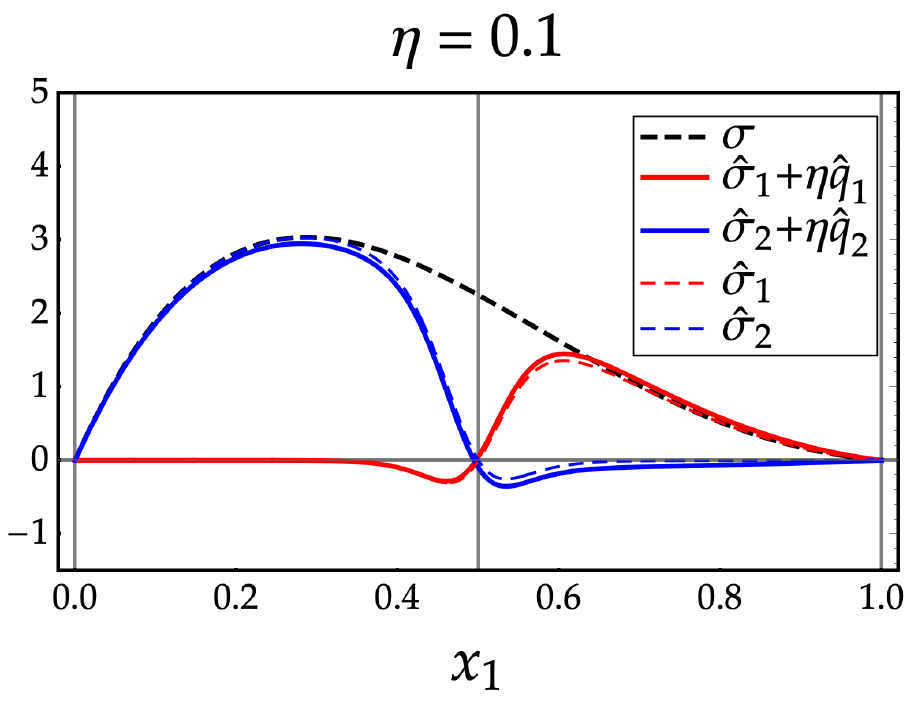}
    \end{subfigure}
    \hfill
	\begin{subfigure}[b]{.32\hsize}
		\centering
        \includegraphics[width=\hsize]{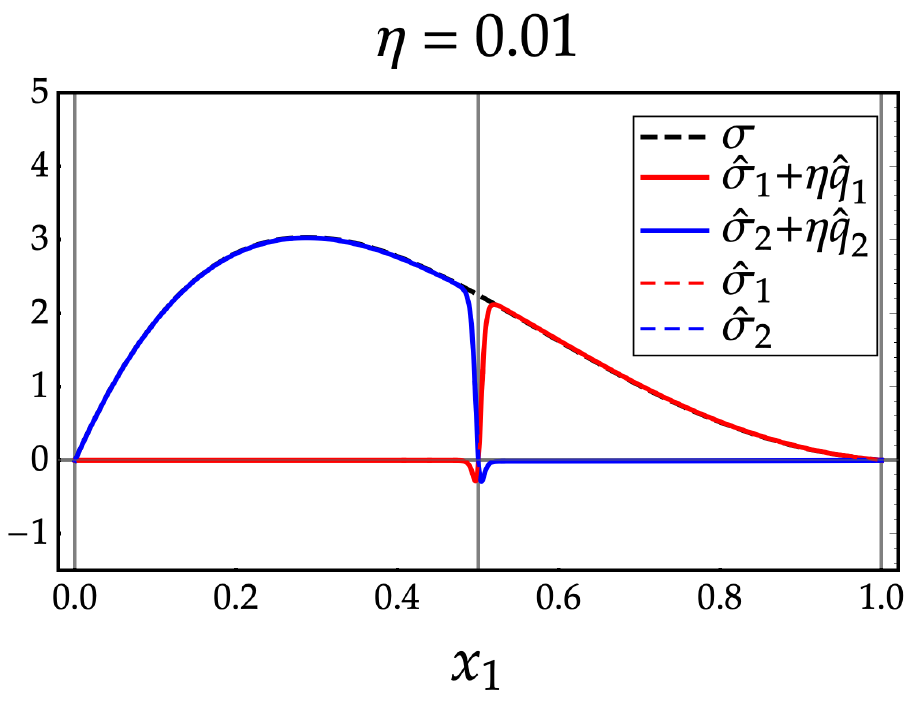}
    \end{subfigure}

    \caption{Illustration of \cref{cor:separable.2} based on graphs of $2k\eta^2(\WH{v} + \WH{q}) = \WH{\sigma} + \eta\WH{\kappa}$ under different values of $\eta$ in a separable two-action congestion game, where $\kappa_i(x) \Is \langle F_i''(x),\Sigma(x)\rangle$ is the unscaled curvature term. The vertical axes report this scaled correction. $\operatorname{br}(x) = \{1\}$ if $x_1 < 1/2$ and $\operatorname{br}(x) = \{2\}$ if $x_1 > 1/2$. For comparison, the dashed curves show only $\WH{\sigma}$.} 
    \label{fig:sep-congestion}
\end{figure}

\begin{example}
\Cref{fig:sep-congestion} considers a separable \emph{congestion game} $F(x) = (- x_1, - 2 (x_2)^2)$ in which the unique Nash equilibrium satisfies $x_1^* = 1/2$. 
The curvature premium biases the aggregate choice toward action $1$ because $f_1''(x_1) = 0 > - 4 = f_2''(x_2)$ for all $x \in X$. 
Notably, because marginal payoffs are state dependent, $\sigma_F(x) = (5 - 4x_1)^2 x_1 x_2$ is not symmetric about the midpoint $1/2$, in contrast to linear population games in \cref{cor:2x2-delta}. 
While $\WH v$ vanishes at the interior Nash equilibrium $x_1^*$ as $P_1(x^*) = 1/2$, nonzero distortion remains at $x_1^*$ due to $\WH q$. 
In this congestion game, the interior fixed point of the approximated rule $\TlL$ must lie to the right of the Nash equilibrium $x_1^* = 1/2$. 
At the Nash equilibrium where variance premiums vanish, the population behaves as if agents prefer the less risky option (i.e., action $1$), whose expected cost is lower under sampling noise due to curvature effects. 
We also confirm from \cref{fig:sep-congestion} that the magnitudes of curvature premiums are broadly smaller than those of variance premiums, consistent with our discussion in \cref{sec:origins-of-premiums}. 
\end{example}

Before concluding this section, it should be noted that the curvature premium is conditional on the assumed inference rule.
In our model, agents first adopt the empirical distribution $w$ as an estimate of $x$ and then evaluate the plug-in payoff $F(w)$.
For a linear game $F(x)=Ax$, this plug-in rule has a natural interpretation because $Aw$ is the sample average of the payoffs each action would earn against the observed opponents, and it is an unbiased estimator of $Ax$. 
With nonlinear payoffs, however, estimating $x$ by $w$ and estimating $F(x)$ by $F(w)$ are not equivalent, and agents do not have to infer $F(x)$ by $F(w)$. 
As discussed, $q_i$ captures the Jensen-type bias due to convexity or concavity of the payoff function. 
Alternative payoff estimators may therefore change or eliminate this term. 
A general analysis of this aspect is left for future research.

\section{The critical joint limit}
\label{sec:joint-limit}

\subsection{The critical scaling of \texorpdfstring{$k$}{k} and \texorpdfstring{$\eta$}{eta}}

The meaning of a ``large'' sample depends on how precisely agents respond to payoffs.
Sampling error in the empirical state $w$ has scale $k^{-1/2}$ because $\Var[w]=k^{-1}\Sigma(x)$, whereas $\eta$ is the payoff scale on which logit choice discriminates among actions.
The ratio of these scales is therefore $k^{-1/2}/\eta = 1/(\sqrt{k}\eta)$.
As discussed, for the approximation in \cref{thm:approximation} to become relatively exact, we need $k\eta^2 \to \infty$, in which case the noise ratio must vanish. 
Then, sampling effects appear as only a marginal perturbation of full-information logit choice, as the approximation formula in \cref{thm:approximation} confirms. 
Here, we instead study the boundary case in which the ratio remains positive and finite.
Throughout, arrows indicate convergence as $k \to \infty$ unless otherwise noted. 

Let the logit noise level $\eta > 0$ fall as the sample size $k$ grows, so that agents respond more and more precisely as information becomes more and more accurate. 
Sampling noise gradually vanishes, but agents become more and more sensitive at the same time, so the interaction between the two sources of noise remains. 
Specifically, for a sequence of logit noise levels $\{\eta_k\}_{k=1}^\infty$, we define the \emph{$\tau$-limit}, with $\tau\in(0,\infty)$, by
\begin{align}
    & \tau_k \Is \sqrt{k}\eta_k \longrightarrow \tau\in(0,\infty) 
    \quad \text{as} \quad 
    k \longrightarrow \infty. 
    \label{eq:tau-lim}
    \tag{$\tau$-\textsf{LIM}}
\end{align}
Along such a sequence, both sources of noise vanish, but the noise ratio converges to $1/\tau$. 
The scaled sampling error becomes asymptotically Gaussian without becoming behaviorally negligible.\footnote{An alternative interpretation is based on the virtual payoffs. 
The logit virtual payoff $H$ is of order $\eta_k$, and the leading sampling-induced virtual payoff $\TlG$ is of order $1/(k\eta_k)$ as discussed in \cref{sec:origins-of-premiums}. 
Both are of order $k^{-1/2}$ under condition \eqref{eq:tau-lim} and therefore have comparable effects on agents' choice behavior.} 
At fixed $\tau \in (0,\infty)$, increasing $k$ alone therefore does not make the approximation in \cref{thm:approximation} accurate, so we need a different characterization.

\subsection{Equilibrium distortion in the \texorpdfstring{$\tau$}{tau} limit} 

For simplicity, we consider the local limit of interior equilibria for linear population games $F(x)=Ax$. 
Linearity eliminates curvature effects and makes inferred payoffs exactly linear in the sampling error.

Let $x^*\in\Int{X}$ be a completely mixed Nash equilibrium for which $Ax^*= \Olu \Vt1$ for some $ \Olu \in\BbR$, and let $a_i$ denote the $i$th row of $A$ written as a column vector.
Local displacements of the population state lie in the tangent space $T \Is \{h\in\BbR^n: \sum_{i\in S} h_i = 0\}$.
A displacement $h\in T$ changes the payoff difference between actions $i$ and $j$ by $F_i(x^* + h) - F_j(x^* + h) = \la a_i-a_j,h\ra$.
We impose the local payoff-identification condition
\begin{align}
    h\in T
    \quad\text{and}\quad
    \la a_i-a_j, h \ra =0
    \quad \forall i,j\in S
    \quad\Longrightarrow\quad
    h=\Vt0.
    \label{eq:tau-lim.regularity}
    \tag{\textsf{R}}
\end{align}
Thus, every nonzero deviation $h\in T$ must change at least one payoff difference, which is a natural regularity condition for an interior Nash equilibrium.
For the two-action case, this is equivalent to the condition $\lambda \ne 0$ in \cref{cor:2x2-delta}, and $\lambda = 0$ is the degenerate case in which all $x \in X$ are Nash equilibria. 

Since $\eta \to \tau/\sqrt{k}$, both frictions operate at scale $k^{-1/2}$ in the $\tau$ limit. 
We thus represent states near $x^*$ in the magnified local coordinate $h = \sqrt{k}\, (x^k(h) - x^*) \in T$, or equivalently, 
\begin{align}
    x^k(h)
    \Is
    x^*+\frac{h}{\sqrt{k}}. 
    \label{eq:tau-lim.local-state}
\end{align}
Then, any fixed $h \in T$ gives a valid population state $x^k(h) \in X$ for all sufficiently large $k$.
Alternatively, for fixed $k$, define the set of $h\in T$ that yields a valid population state by $T_k\Is\{h\in T:x^k(h)\in\Int{X}\}$. 
The sampling logit rule, in terms of the local coordinates $h \in T$, can be written as  
\begin{align}
    & \ClL^k(h)
    \Is
    L^{k,\eta_k}(x^k(h))
    & \forall h\in T_k.
\end{align}
Then, $h_k\in T_k$ represents an SLE when $\ClL^k(h_k)=x^k(h_k)$, and its limit as $k \to \infty$ characterizes the equilibrium distortion.\footnote{We must restrict $\ClL^k$ to $T_k$ because otherwise the multinomial probabilities we use in $L^{k,\eta}$ are not well defined.}

The next lemma identifies the limit of $\ClL^k: T_k \to \Int{X}$ under condition \eqref{eq:tau-lim}.
\begin{lemma}[Critical local response limit]
\label{lem:tau-lim.local-limit}
For a linear population game $F(x)=Ax$ and an interior Nash equilibrium $x^*\in \Int{X}$, suppose that $\eta_k > 0$ satisfies the $\tau$-limit condition \eqref{eq:tau-lim}.
Set $\Sigma^*\Is\Sigma(x^*)$ and let $Z^*\sim N(\Vt0,\Sigma^*)$. 
Define $Q^\tau:T\to\Int{X}$ by
\begin{align}
    Q^\tau(h)
    \Is
    \BbE\left[\Lambda^\tau\bigl(A(h+Z^*)\bigr)\right].
    \label{eq:tau-lim.local-response}
\end{align}
Then, $\mathcal L^k\to Q^\tau$ in $C^1_{\mathrm{loc}}(T)$. 
That is, every compact $K\subset T$ is contained in $T_k$ eventually, and the functions and their derivatives converge uniformly on $K$.
\end{lemma}

Here, $Z^*$ is the Gaussian limit of the scaled sampling error $\sqrt{k}\,(w-x^k(h))$ as $k \to\infty$, and the expectation in \cref{eq:tau-lim.local-response} is taken over $Z^*$. 
As sampling errors preserve total population mass, 
$Z^*$ lies in $T$ because $w-x^k(h) \in T$ for any $w$, $k$, and $h$. 
Thus, $h+Z^*$ is the sample's limiting displacement from $x^*$, and $Q^\tau(h)$ averages the $\tau$-logit response to the resulting payoff variation $A(h+Z^*)$ over $Z^*$.

If $h_k\to h$ along a sequence of SLE for each $k$, then \cref{lem:tau-lim.local-limit} implies $Q^\tau(h)=x^*$.
The following theorem establishes that this limiting equation has a unique solution $h^\tau$, and that there is a unique SLE $x^{k,\eta_k}$ in a neighborhood converging to it. 
\begin{theorem}[SLE in the $\tau$-limit]
\label{thm:tau-lim.joint-limit}
Consider a linear population game $F(x)=Ax$ with an interior Nash equilibrium $x^*\in \Int{X}$ satisfying the local payoff-identification condition \eqref{eq:tau-lim.regularity}.
Let $\eta_k>0$ satisfy the $\tau$-limit condition \eqref{eq:tau-lim}.
Then the following statements hold.
\begin{enumerate}
    \item
    \label{enum:tau-lim.root}
    There is a unique local state displacement $h^\tau\in T$ satisfying $x^*=Q^\tau(h^\tau)$.
    \item
    \label{enum:tau-lim.equilibrium}
    There are $\epsilon>0$ and $k^* <\infty$ such that, for every $k\ge k^* $, the sampling logit choice rule $L^{k,\eta_k}$ has a unique fixed point $x^{k,\eta_k}$ in
    $x^*+\frac{1}{\sqrt{k}} \{h\in T:\|h-h^\tau\|<\epsilon\}$.
    This local SLE satisfies $\sqrt{k}\bigl(x^{k,\eta_k}-x^*\bigr) \to h^\tau$.
\end{enumerate}
\end{theorem}

The convergence in part \cref{enum:tau-lim.equilibrium} is equivalently written as
\begin{align}
    x^{k,\eta_k}
    =
    x^*
    +\frac{h^\tau}{\sqrt{k}}
    +o(k^{-1/2}).
    \label{eq:tau-lim.sle-local-expansion}
\end{align}
Thus, the local SLE converges to $x^*$, but $h_i^\tau/\sqrt{k}$ represents the leading displacement of action $i$'s equilibrium share as $k$ grows. 
The sign of $h_i^\tau$ indicates whether that action is over- or underrepresented relative to $x^*$. 

The critical local response has a random utility interpretation.
Let $\varepsilon=(\varepsilon_i)_{i\in S}$ be a vector of independent standard Type-I extreme-value (i.e., Gumbel) random variables, independent of $Z^*$.
The Gumbel representation of logit choice \citep{McFadden-Chap1974} gives
\begin{align}
    Q_i^\tau(h)
    =
    \BbP\left(
        \bigl[A(h+Z^*)\bigr]_i+\tau\varepsilon_i
        \ge
        \max_{j\in S}
        \left\{\bigl[A(h+Z^*)\bigr]_j+\tau\varepsilon_j\right\}
    \right).
    \label{eq:tau-lim.random-utility}
\end{align}
Here, $AZ^*$ is the correlated Gaussian payoff disturbance generated by sampling, whereas the shocks $\tau\varepsilon_i$ are independent Gumbel choice shocks.
Both remain of first order for every fixed $\tau\in(0,\infty)$.\footnote{In the terminology of discrete choice modeling, $Q^\tau$ has the form of a \emph{logit-kernel probit} \citep[Sec.~5.6.2]{Train-Book2009}, as it averages a logit kernel over the Gaussian terms $AZ^*$, or equivalently a \emph{Gaussian error-components mixed logit} \citep[Sec.~6.3]{Train-Book2009}.}
The caveat is that this random utility interpretation applies only to the limiting map $Q^\tau$, for which the distribution of $AZ^*+\tau\varepsilon$ is fixed as $h$ varies.
At finite $k$, by contrast, the multinomial distribution of the sampled state varies with the population state $x$.
Thus, the sampling mechanism generally does not represent $L^{k,\eta}$ as an additive random utility model with material payoff $F(x)$ as its systematic component and a state-independent disturbance distribution.\footnote{Relatedly, \cref{obs:TlL-interpretation} should be read as an asymptotic interpretation of the exact sampling rule. 
Its virtual-payoff representation is exact for the large-$k$ approximation $\TlL$ but not for $L^{k,\eta}$.
In particular, it does not by itself characterize the critical regime considered here, in which $k\eta_k^2\to\tau^2$ rather than $k\eta_k^2\to\infty$.}

The local limit also yields natural sufficient conditions for stability under the sampling logit dynamic \eqref{eq:sLD}, based on the underlying payoff matrix $A$. 
\begin{proposition}[Stability of the $\tau$-limit SLE]
\label{prop:tau-lim.stability}
Under the hypotheses of \cref{thm:tau-lim.joint-limit}, let $x^{k,\eta_k}$ be the local SLE of part \cref{enum:tau-lim.equilibrium}.
\begin{enumerate}
    \item
    \label{enum:tau-lim.stability-stable}
    If $\la h, Ah \ra < 0$ for all $h \in T\setminus\{\Vt0\}$, there is $k^{**}<\infty$ such that $x^{k,\eta_k}$ is locally asymptotically stable under \eqref{eq:sLD} for every $k \ge k^{**}$.
    \item
    \label{enum:tau-lim.stability-unstable}
    If $\la h, Ah \ra > 0$ for all $h \in T\setminus\{\Vt0\}$, there is $k^{**}<\infty$ such that $x^{k,\eta_k}$ is unstable under \eqref{eq:sLD} for every $k \ge k^{**}$.
\end{enumerate}
\end{proposition}
\noindent Part \ref{enum:tau-lim.stability-stable} applies to strictly stable (contractive) games \citep{Hofbauer-Sandholm-JET2009}, in which $x^*$ is the unique Nash equilibrium.
Part \ref{enum:tau-lim.stability-unstable} includes completely mixed equilibria of two-action coordination games, where $T$ is one-dimensional and $\la h,Ah\ra=\lambda h_1^2>0$ for $h\ne\Vt0$ in the notation of \eqref{eq:A-2x2}.
While the interior equilibrium associated with the mixed equilibrium of a coordination game remains unstable, this suggests that, in strictly stable games, the critical SLE is a robust prediction under \eqref{eq:sLD}. 

\subsection{Boundary behavior} 

As $\tau \in (0,\infty)$ varies, the critical local response connects sampling-dominant and choice-dominant limits.
To relate the latter to the variance premium in \cref{sec:approximation}, set
\begin{align}
    \sigma_i^*
    \Is
    (a_i-\overline a^*)^\top
    \Sigma^*
    (a_i-\overline a^*)
    \quad \text{with} \quad
    \overline a^*
    \Is
    \sum_{l\in S}x_l^*a_l.
    \label{eq:tau-lim.sigma-star}
\end{align}
The scalar $\sigma_i^*$ is the one-draw variance of action $i$'s relative payoff at $x^*$, using the limiting equilibrium weights $x^*$ in defining $\overline{a}$.\footnote{This is the limit of $\sigma_i$ in \cref{eq:v-vector-form} along equilibrium curves converging to $x^*$. 
It generally differs from the zero-noise limit of $\sigma_i(x^*)$ at a fixed positive $\eta$ because the logit weights at $x^*$ are uniform for every $\eta > 0$.}

\begin{corollary}[Boundary behaviors of the $\tau$-limit]
\label{cor:tau-lim.endpoints}
Under the hypotheses of \cref{thm:tau-lim.joint-limit}, the unique solution $h^\tau$ of $x^*=Q^\tau(h^\tau)$ has the following limiting behavior.
\begin{enumerate}
    \item
    \label{enum:tau-lim.small-tau}
    As $\tau\Toz$, $Q^\tau(h)\to Q^0(h)$ for every $h\in T$, where $Q^0:T\to\Int{X}$ is defined by
    \begin{align}
        & Q_i^0(h)
        \Is
        \BbP\left(
            \bigl[A(h+Z^*)\bigr]_i>
            \max_{j\ne i}\bigl[A(h+Z^*)\bigr]_j
        \right)
        & \forall i\in S,\ h\in T.
        \label{eq:tau-lim.gaussian-br}
    \end{align}
    Moreover, $h^\tau\to h^0$, where $h^0$ is the unique solution of $x^*=Q^0(h^0)$.
    \item
    \label{enum:tau-lim.large-tau}
    For $\tau \to \infty$, define the \emph{local virtual payoff} around $x^*$ as 
    \begin{align}
        & \ClF_i^\tau(h)
        \Is
        a_i^\top h
        +\frac{1}{2\tau}\sigma_i^*
        -\tau\log x_i^*
        & \forall i\in S,\ h\in T.
        \label{eq:tau-lim.local-virtual-payoff}
    \end{align}
    Then, as $\tau\to\infty$, there is a scalar $c^\tau$, common to all actions, such that 
    \begin{align}
        & \ClF_i^\tau(h^\tau)
        =
        c^\tau+O(\tau^{-3})
        & \forall i\in S
        \label{eq:tau-lim.virtual-indifference}
    \end{align}
    for some $h^\tau \in T$. 
\end{enumerate}
\end{corollary}

The small-$\tau$ limit in part \cref{enum:tau-lim.small-tau} is a Gaussian sampling best response (i.e., a \emph{probit choice}) in which agents best respond to $A(h+Z^*)$. 
In the large-$\tau$ limit, part \cref{enum:tau-lim.large-tau} is equivalent to the following pairwise expansion: 
\begin{align}
    & \la a_i-a_j , h^\tau \ra
    =
    \tau\log\frac{x_i^*}{x_j^*}
    -
    \frac{1}{2\tau}(\sigma_i^*-\sigma_j^*)
    +O(\tau^{-3})
    & \forall i,j\in S,\ i\ne j. 
    \label{eq:tau-lim.large-tau-expansion}
\end{align}
Condition \eqref{eq:tau-lim.regularity} then ensures that such payoff differences identify a unique $h^\tau \in T$. 

The three terms in \cref{eq:tau-lim.local-virtual-payoff} are the local material payoff, the sampling-induced variance premium, and the virtual payoff associated with logit noise.
Thus, \Cref{eq:tau-lim.virtual-indifference} approximately equalizes total local virtual payoffs across actions, in parallel with \cref{obs:TlL-interpretation}~\ref{enum:virtual-game-2}.
In \cref{eq:tau-lim.local-virtual-payoff}, $\tau \log x_i^*$ gives the logit effects, whereas $\sigma_i^*/(2\tau)$ gives the variance effects.\footnote{For the two-action game \eqref{eq:A-2x2}, substituting $\eta_k=\tau/\sqrt{k}$, and hence $\theta=1/(2\tau^2)$, into the local expansion in the proof of \cref{prop:approx-coord} yields the same two leading terms as \cref{eq:tau-lim.large-tau-expansion}.}

\begin{example}
Consider the two-action game \eqref{eq:A-2x2} with an interior Nash equilibrium $x_1^* = x^*\in(0,1)$, for which the payoff-identification condition reduces to $\lambda\ne0$.
Write $h=(h_1,-h_1)\in T$ and $Z^*=(Z_1^*,-Z_1^*)$ with $Z_1^*\sim N(0,\rho^2)$ and $\rho^2 \Is x^*(1-x^*)$.
The limiting payoff difference is $[A(h+Z^*)]_1-[A(h+Z^*)]_2=\lambda(h_1+Z_1^*)$, so the equilibrium condition $Q^\tau(h)=x^*$ reduces to 
\begin{align}
    \BbE\left[
        \frac{1}{1+\exp\bigl(-\lambda(h_1+Z_1^*)/\tau\bigr)}
    \right]
    = x^*,
    \label{eq:tau-lim.2x2-equation}
\end{align}
whose unique solution is denoted by $h_1^\tau$.
As $\tau \Toz$, solving $Q^0(h)=x^*$ gives
\begin{align}
    h_1^0
    = \sgn(\lambda) \rho \Phi^{-1}(x^*),
    \label{eq:tau-lim.2x2-h0}
\end{align}
where $\Phi$ is the standard normal distribution function.
For $\lambda>0$, $h_1^0$ is the $x^*$-quantile of $Z_1^*$ that makes action $1$ optimal with probability exactly $x^*$. 
In a coordination game with $x^*<1/2$, this displacement is negative, so the critical branch lies below the mixed equilibrium, consistent with \cref{prop:approx-coord}. 
In a strictly stable game ($\lambda<0$), its sign is reversed.
As $\tau\to\infty$, the large-$\tau$ expansion \eqref{eq:tau-lim.large-tau-expansion} is
$\lambda h_1^\tau = \tau\log\frac{x^*}{1-x^*} - \frac{1}{2\tau}(1-2x^*)\sigma_A(x^*)$, where $\sigma_A(x^*)=\lambda^2\rho^2$ as in \cref{cor:2x2-delta}. 

\Cref{fig:tau-lim.2x2-probs} plots the choice probabilities for the coordination game \eqref{eq:A.2x2.coord-ex}, approximated by taking $k = 1600$. 
\Cref{fig:tau-lim.2x2} illustrates both boundary cases and the finite-$k$ convergence of \cref{thm:tau-lim.joint-limit}. 
The vertical axis in \cref{fig:tau-lim.2x2-total} measures the local displacement of the (unstable) interior Nash equilibrium. 
A state $x_1^*+h/\sqrt{k}$ lies above the critical SLE approximately when $h>h_1^\tau$. 
\Cref{fig:tau-lim.2x2-detrend} isolates the role of sampling by removing the pure-logit component from $h_1^\tau$.\footnote{To represent the pure-logit component, we define $h^\infty=(h_1^\infty,-h_1^\infty)\in T$ by
$h_1^\infty\Is\lambda^{-1}\log\bigl(x^*/(1-x^*)\bigr)$.
If sampling noise is removed by setting $Z_1^*=0$ in \cref{eq:tau-lim.2x2-equation}, its unique solution is $h_1=\tau h_1^\infty$.
Thus, $h^\tau-\tau h^\infty$ measures the critical displacement relative to the pure-logit case.
For each finite-$k$ marker in \cref{fig:tau-lim.2x2-detrend}, we analogously subtract the displacement of the corresponding full-information $\eta_k$-logit equilibrium.}
As $\tau\to\infty$, the sampling-induced displacement vanishes at rate $O(\tau^{-1})$ because the Gumbel choice shocks grow relative to the normalized Gaussian sampling disturbance $Z^*$ (cf. \cref{eq:tau-lim.random-utility}). 
\end{example}

\begin{figure}[t]
    \centering
    \includegraphics[width=.392\hsize]{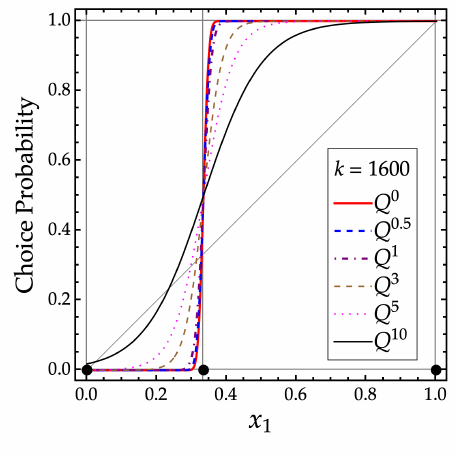}
    \caption{Choice probabilities generated by $Q^\tau$. We plot $Q_1^\tau(\sqrt{k}\,(x - x^*))$ using $k = 1600$ and setting $\eta_k = \tau/\sqrt{k}$ for each $\tau$. For $Q^0$, we use the explicit quantile representation based on the Gaussian distribution.}
    \label{fig:tau-lim.2x2-probs}
\end{figure}

\begin{figure}[t]
    \centering
    \begin{subfigure}[c]{.48\textwidth}
        \includegraphics[width=.95\hsize]{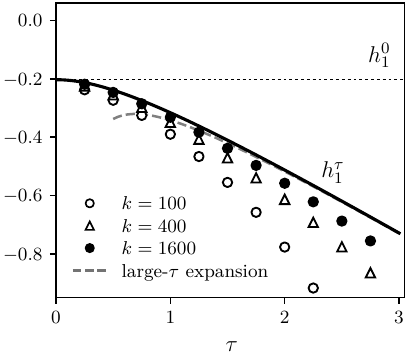}
        \caption{Total displacement $h^\tau$}
        \label{fig:tau-lim.2x2-total}
    \end{subfigure}
    \hfill 
    \begin{subfigure}[c]{.48\textwidth}
        \includegraphics[width=.95\hsize]{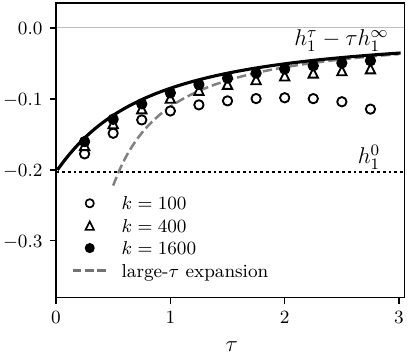}
        \caption{Sampling displacement $h^\tau - \tau h^\infty$}
        \label{fig:tau-lim.2x2-detrend}
    \end{subfigure}
    \caption{Critical local displacement in the coordination game \eqref{eq:A.2x2.coord-ex} for which $\lambda=3$ and $x^*=1/3$. In panel (A), the solid curve is the solution $h_1^\tau$ of \eqref{eq:tau-lim.2x2-equation}, the dotted line is its small-$\tau$ limit \eqref{eq:tau-lim.2x2-h0}, and the dashed curve is its large-$\tau$ expansion \eqref{eq:tau-lim.large-tau-expansion}. The markers show exact SLE displacements on the magnified scale $\sqrt{k}\,(x_1^{k,\eta_k}-x_1^*)$ under $\eta_k=\tau/\sqrt{k}$. Panel (B) plots the sampling effects on the local displacements.}
    \label{fig:tau-lim.2x2}
\end{figure}

\cref{thm:approximation} in \cref{sec:approximation} is global on $X$, permits nonlinear payoffs, and characterizes sampling effects as marginal perturbations of logit choice.
\Cref{thm:tau-lim.joint-limit} is local to an interior Nash equilibrium of a linear game but covers the regime in which sampling noise and logit noise interact at first order. The family $\{Q^\tau\}_{\tau\in(0,\infty)}$ connects sampling-dominant and logit-dominant behaviors: it yields Gaussian sampling best response as $\tau\Toz$ and logit choice with a variance premium as $\tau\to\infty$. 
In the context of the $\tau$ limit, the variance premium formula in \cref{sec:approximation} corresponds to a large-$\tau$ expansion of $Q^\tau$. 

\section{Small samples}
\label{sec:examples}

\Cref{sec:approximation,sec:joint-limit} concern regimes in which $k$ grows, with sampling noise respectively smaller than and comparable to logit noise. 
We now consider the complementary regime in which $k$ is fixed and small to discuss connections between the sampling logit framework and previous results in the literature. 

\subsection{Exact virtual payoffs}
\label{sec:small-k-examples}

In \cref{sec:approximation}, the large-$k$ analysis approximates the sampling logit rule $L$ by $\TlL$, and it decomposes the associated virtual payoff into variance and curvature premiums. 
For small-$k$ settings, however, the approximation generally becomes less accurate. 
For instance, the large-$k$ virtual payoff representation $\TlG$ in \cref{obs:TlL-interpretation} may no longer be well defined because $1 + \Htv + \Htq$ can become negative on $X$ when $k$ is small. 

\begin{figure}[t!]
    \centering
    \begin{subfigure}[c]{.49\textwidth}
        \includegraphics[width=.85\hsize]{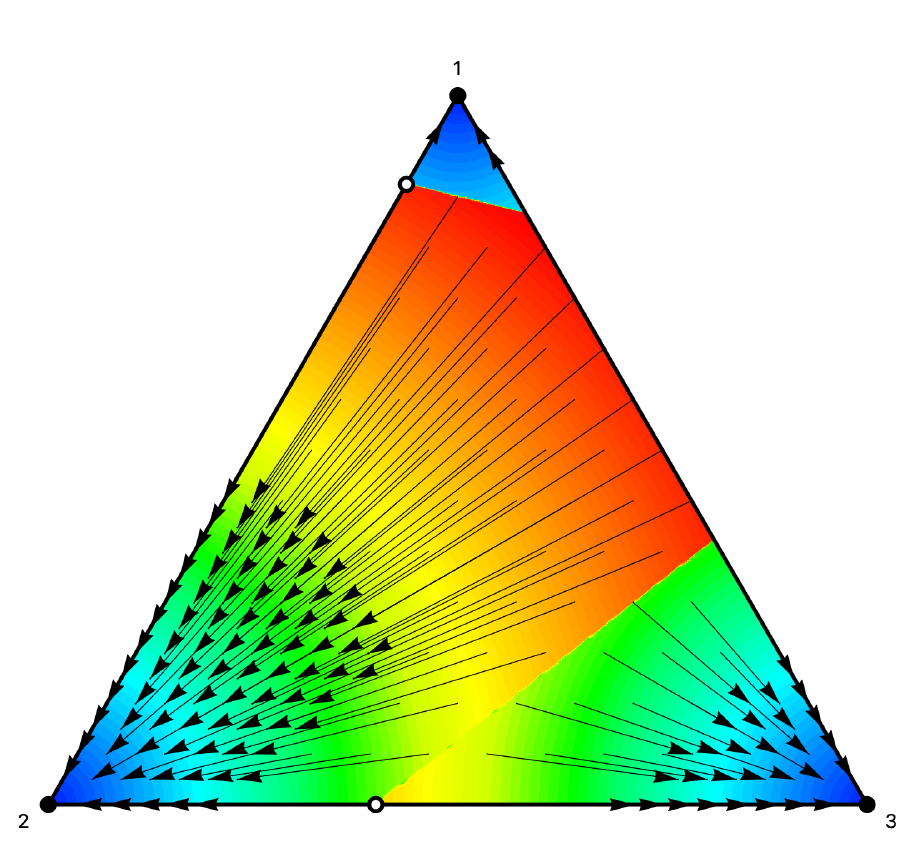}
        \caption{Best response}
        \label{fig:Young-BR}
    \end{subfigure}
    \begin{subfigure}[c]{.49\textwidth}
        \includegraphics[width=.85\hsize]{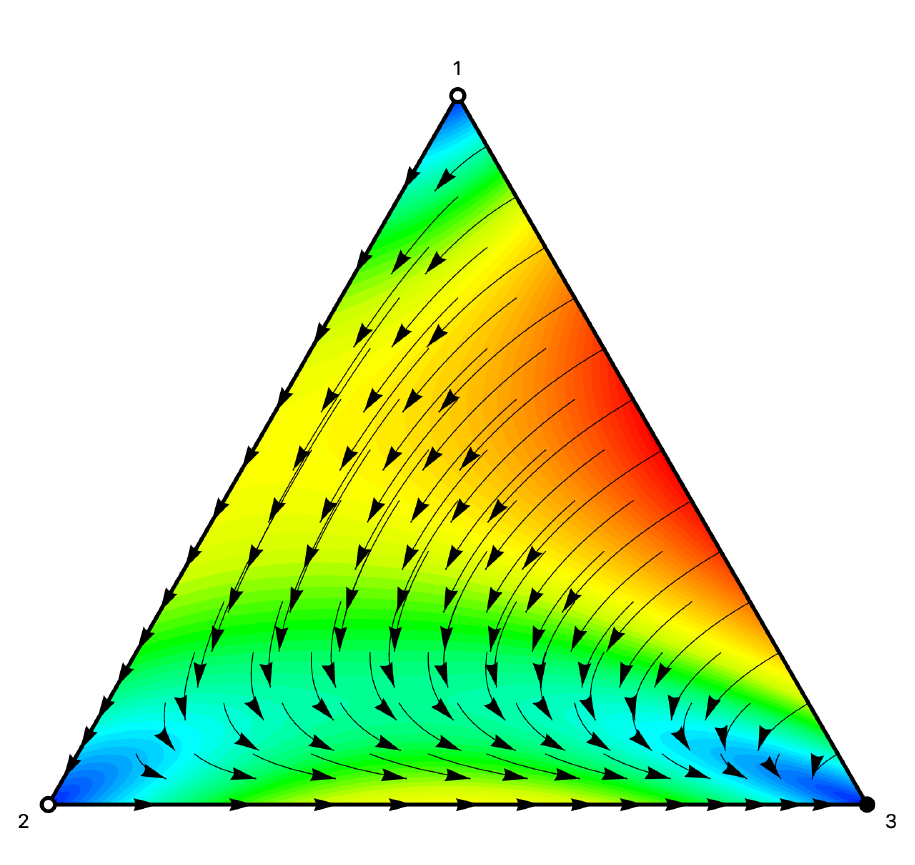}
        \caption{Sampling best response ($k = 2$)}
        \label{fig:Young-SBR}
    \end{subfigure}

    \begin{subfigure}[c]{.49\textwidth}
        \includegraphics[width=.85\hsize]{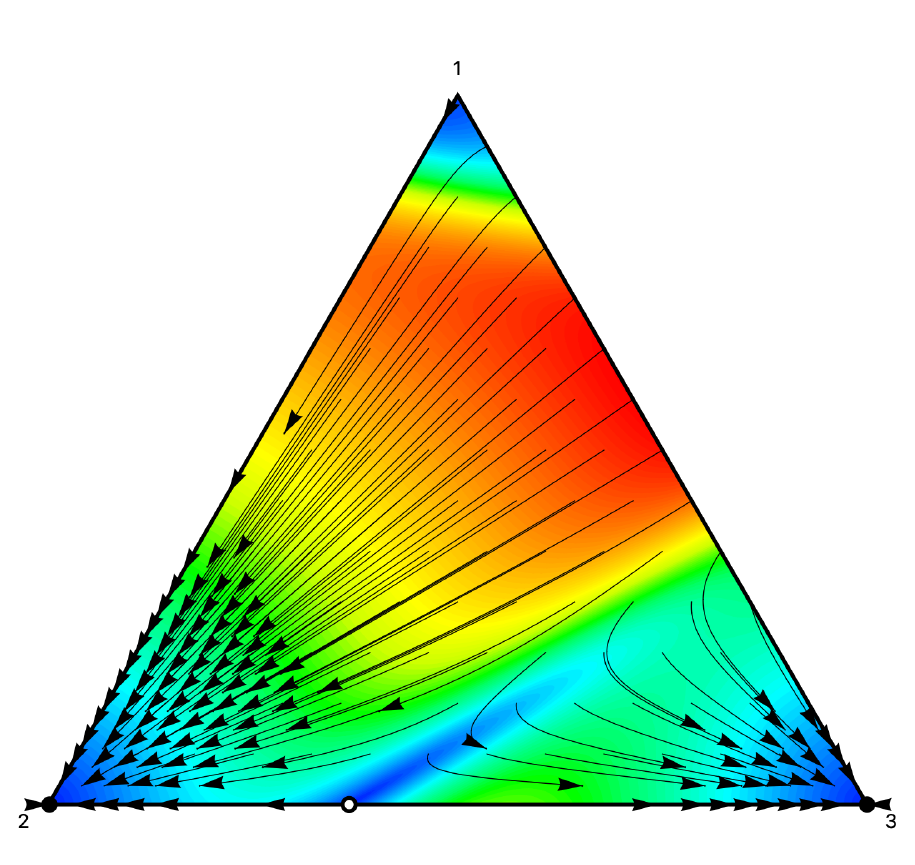}
        \caption{Logit ($\eta = 0.3$)}
        \label{fig:Young-LD}
    \end{subfigure}
    \begin{subfigure}[c]{.49\textwidth}
        \includegraphics[width=.85\hsize]{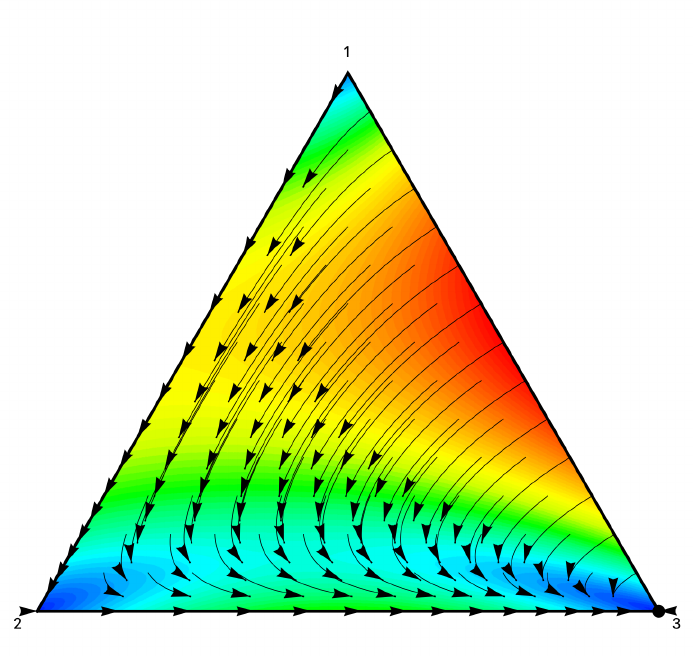}
        \caption{Sampling logit ($k = 2$, $\eta = 0.3$)}
        \label{fig:Young-SLD}
    \end{subfigure}

    \caption{Phase diagrams of the four dynamics in Young's game \eqref{eq:A.Young}. Arrows show sample trajectories, and background contours depict the speed of adjustment: warmer colors indicate faster adjustment, whereas cooler colors indicate slower adjustment. This figure is generated with the \texttt{Dynamo} software \citep{Franchetti-Sandholm-BT2013}.} 
    \label{fig:Young}
\end{figure}
In principle, the virtual-payoff representation does not depend on an asymptotic argument that relies on large $k$. 
To see this, define the \emph{exact payoff distortion}
\begin{align}
    G_i(x)=\eta\log\frac{L_i(x)}{P_i(x)}, 
    \label{eq:exact-premium}
\end{align}
which is well defined for every $k\geq 1$ and $\eta>0$ because $L_i(x)$ and $P_i(x)$ are strictly positive. 
Direct substitution shows that the sampling logit choice is the logit choice for the virtual game $F + G$, and an SLE is an $\eta$-logit equilibrium of this game. 
In the large-$k$ regime, $\TlG$ in \cref{obs:TlL-interpretation} provides a second-order approximation to $G$ and yields a clear decomposition into two premiums as discussed in \cref{sec:approximation}. 
For small $k$, $G$ does not admit such an intuitive decomposition, and it is more of an accounting identity.

\begin{figure}[t!]
    \centering
    \begin{subfigure}[c]{.49\textwidth}
        \includegraphics[width=.88\hsize]{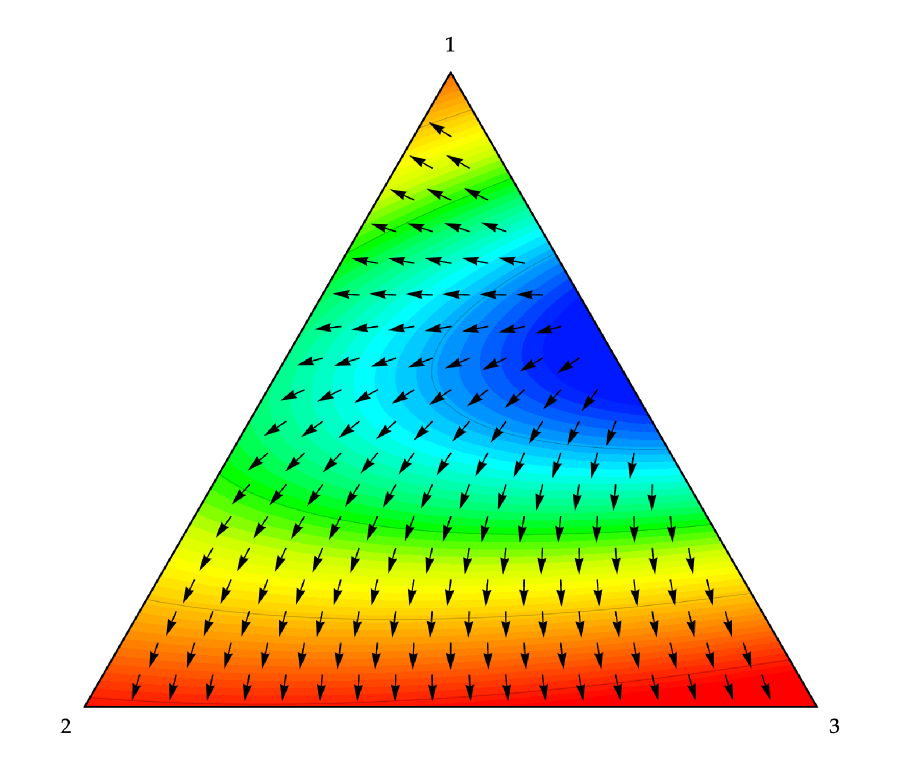}
        \caption{Payoff $F$}
        \label{fig:Young-F}
    \end{subfigure}
	\begin{subfigure}[c]{.49\textwidth}
        \includegraphics[width=.88\hsize]{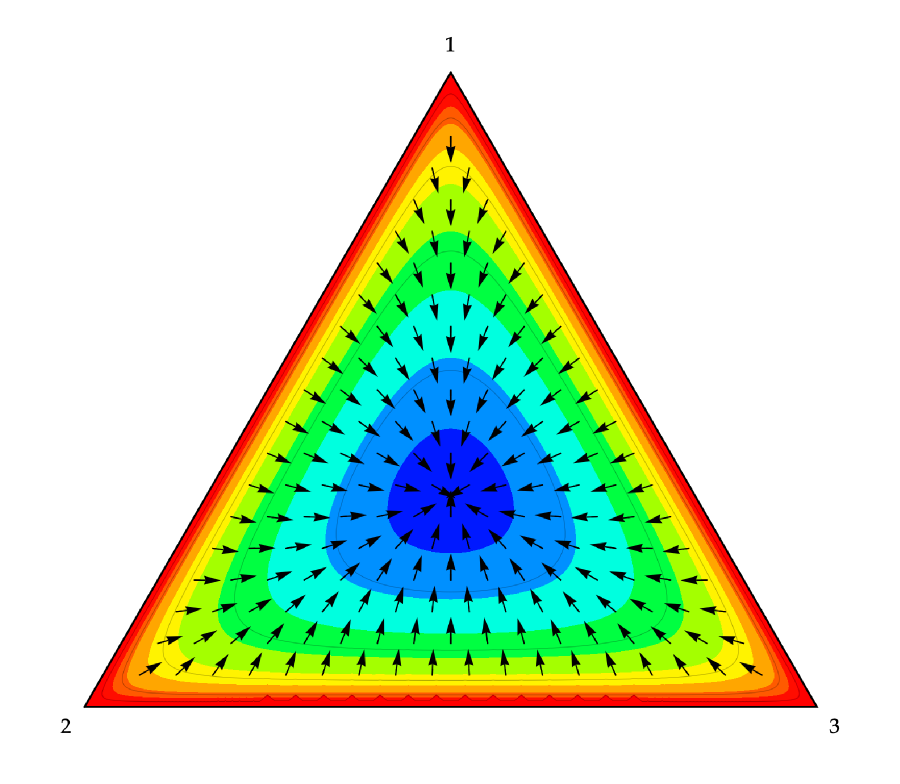}
        \caption{Logit virtual payoff $H$ ($\eta = 0.3$)}
        \label{fig:Young-H}
    \end{subfigure}

    \begin{subfigure}[c]{.49\textwidth}
        \includegraphics[width=.88\hsize]{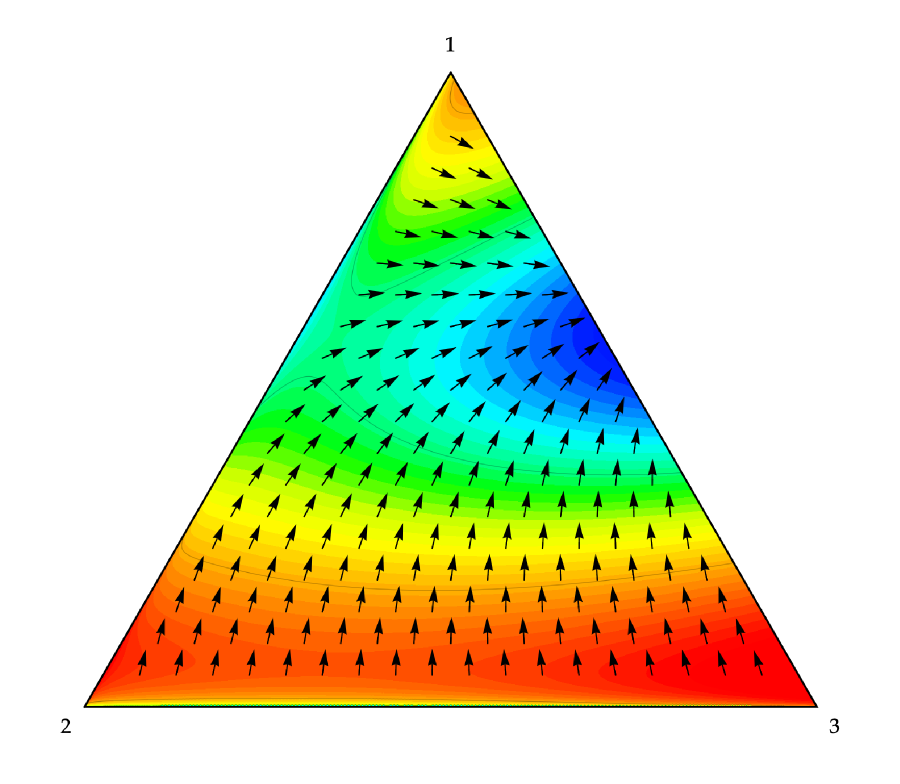}
        \caption{Exact payoff distortion $G$ ($k = 2, \eta = 0.3$)}
        \label{fig:Young-SL-G}
    \end{subfigure}
	\begin{subfigure}[c]{.49\textwidth}
        \includegraphics[width=.88\hsize]{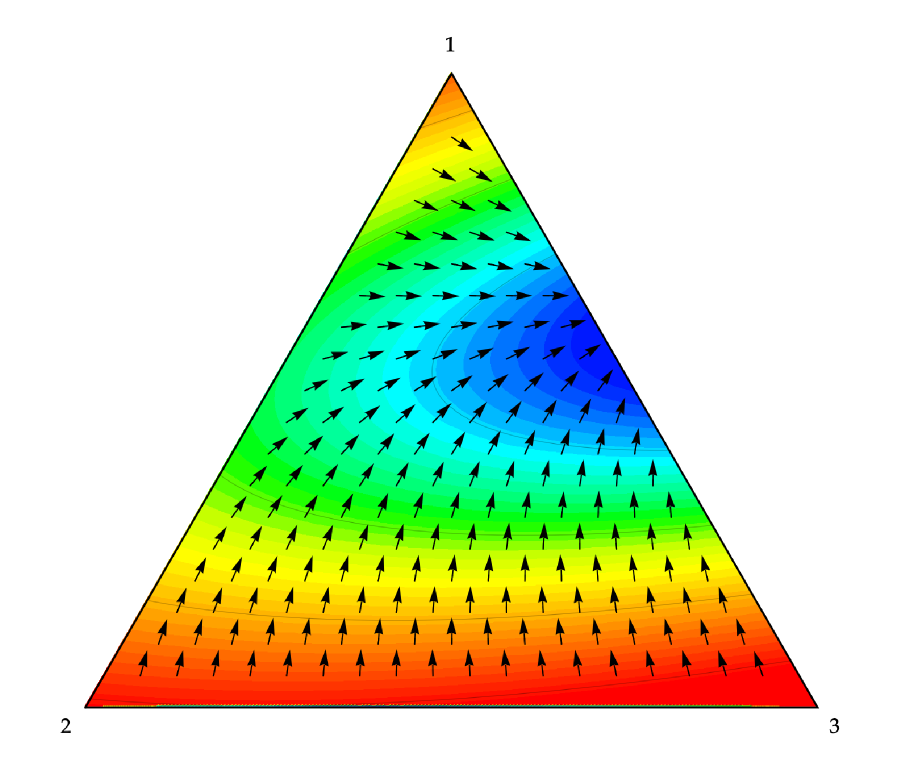}
        \caption{Exact payoff distortion $G$ ($k = 2, \eta = 0.001$)}
        \label{fig:Young-SBR-G}
    \end{subfigure}
    \caption{Projected payoff fields in Young's game \eqref{eq:A.Young}. Arrows show the normalized direction and background contours show its magnitude. Panel (D) approximates the virtual payoff fields for the sampling best response dynamic by taking small $\eta$.}
    \label{fig:exact-G-Young}
\end{figure}
Nonetheless, $G$ provides a useful illustration for the small-$k$ regime. 
For example, consider the linear population game based on the game of \citet{Young-ECTA1993}: 
\begin{align}
    A = 
    \begin{bmatrix}
        6 & 0 & 0 \\ 5 & 7 & 5 \\ 0 & 5 & 8
    \end{bmatrix}. 
    \label{eq:A.Young}
\end{align}
The game possesses three strict equilibria $\{e_1,e_2,e_3\}$ and two mixed Nash equilibria on the boundary of $X$. 
\Cref{fig:Young} compares the phase portraits of the four dynamics introduced in \cref{sec:model} for the case $k = 2$ and, for logit-based models, $\eta = 0.3$. 
Because $k=2$, their behavior should not be inferred from $\TlL$, and the exact distortion $G$ provides the appropriate virtual payoff representation. 

The best response dynamic \eqref{eq:BRD} provides a basic refinement under which the two boundary equilibria are deemed unstable (\cref{fig:Young-BR}). 
However, it does not yield a unique prediction, as all three corners are locally stable. 
\Cref{fig:Young-SBR} demonstrates that \eqref{eq:sBRD} with $k = 2$ selects $x = e_3$ because $e_3$ attracts all trajectories starting from $X \setminus \{e_1,e_2\}$ \citep[][Example 2]{Oyama-Sandholm-Tercieux-TE2015}. 
\Cref{fig:Young-LD} considers the logit dynamic \eqref{eq:LD}. 
As the dynamic approximates \eqref{eq:BRD} for small $\eta$, multiple stationary states are also locally stable under this dynamic. 
\Cref{fig:Young-SLD} illustrates that finite sampling can similarly sharpen the prediction under positive logit noise, as \eqref{eq:sLD} has a unique SLE near $e_3$, to which all trajectories converge.\footnote{While the comparison of \crefrange{fig:Young-BR}{fig:Young-SLD} may give the impression that the sampling dynamics \eqref{eq:sBRD} and \eqref{eq:sLD} exhibit broadly similar qualitative behavior, this conclusion depends on the underlying game and the noise level. See \cref{app:bilingual} for an example with the \emph{bilingual game} \citep{Galesloot-Goyal-JME1997,Goyal-Janssen-JET1997}.}

The exact payoff distortion $G$ clarifies how sampling produces this change in behavior. 
\Cref{fig:exact-G-Young} plots the projection of $G$ onto the tangent space of the simplex along with those for the payoff function $F$ and the logit virtual payoff $H$ (cf. \cref{obs:TlL-interpretation}). 
The map $G$ is state dependent in a nontrivial manner because finite sampling changes relative incentives in different directions across the simplex. 
Unlike $H$, which is independent of $F$ and favors movement toward the center of the simplex (\cref{fig:Young-H}), $G$ depends on the payoff structure of the game. 
In this example, $G$ points broadly in the opposite direction to $F$ (\cref{fig:Young-F}). 
Sampling raises the virtual payoff of relatively unattractive actions, echoing the preference for suboptimal actions in \cref{cor:2x2-delta}. 
Moreover, the directions in the positive-noise case (\cref{fig:Young-SL-G}) remain broadly similar in the near-zero-noise case (\cref{fig:Young-SBR-G}), connecting the equilibrium selection under \eqref{eq:sLD} and that under \eqref{eq:sBRD}.  
The exact payoff distortion illustrates that sampling noise can be interpreted through a virtual payoff representation for both large and small $k$.

\subsection{On equilibrium selection}
\label{sec:small-k-selection}

The (approximate) equilibrium selection in Young's game leads us to examine selection in our model more systematically. 
Several exact results follow. 

For $k = 1$, the SLE is unique and stable for \emph{any} population game. 
\begin{proposition}
    \label{prop:k=1}
    For any $\eta>0$, if $k=1$, there exists a unique SLE, and it is globally asymptotically stable under \eqref{eq:sLD}. 
\end{proposition}
\noindent For $k=1$, the sampling logit choice rule is linear in the population state. Specifically, we have $L(x) = \Pi x$ where $\Pi_{ij} = P_i^\eta(e_j)$, and \eqref{eq:sLD} reduces to the linear system $\dot x=(\Pi-I)x$. 
The unique $(1,\eta)$-SLE can be seen as the stationary distribution of the associated logit-perturbed Markov chain. 
This equivalence provides a direct analogy with the \emph{stochastic stability} approach under log-linear learning rules \citep{Young-ECTA1993,Blume-GEB1993,Blume-GEB1995,Blume-GEB2003}. 
That is, depending on the game, as we gradually decrease $\eta$, the population state may approach a specific pure population state, as we will discuss later.

A similar uniqueness result is available for $k = 2$ in any two-action game.\begin{proposition}
\label{prop:k=2}
For any $\eta>0$, if $k=2$ and $n = 2$, there exists a unique SLE, and it is globally asymptotically stable under \eqref{eq:sLD}. 
\end{proposition}
\noindent The global uniqueness result, however, does not generalize to environments with $k = 2$ and $n \ge 3$, or those with $k \ge 3$. 

Instead, linear two-action games allow for a more detailed characterization for general $k$. 
We focus on a normalized form as explained below. 
Logit choice and best responses depend only on payoff differences. 
Accordingly, any two-action linear population game generated by $\begin{bsmallmatrix}a&b\\c&d\end{bsmallmatrix}$ is behaviorally equivalent to the diagonal form $\begin{bsmallmatrix}s&0\\0&t\end{bsmallmatrix}$ with $s=a-c$ and $t=d-b$, because both games have the same relative payoff $F_1(x)-F_2(x)=s x_1-t x_2 = (s + t) x_1 - t$. 
Thus, when $s,t>0$ and $s\ne t$, with appropriate relabelling of the actions, the following formulation covers all such two-action linear coordination games: 
\begin{align}
    &  
    A = 
    \begin{bmatrix}
        s & 0 \\ 
        0 & t 
    \end{bmatrix}
    \quad  (s > t > 0).  
    \label{eq:A.2x2.coord}
\end{align}

The Nash equilibria satisfy $x_1 \in \{0,t/(s+t),1\}$, where $t/(s+t)<1/2$ is the population share at which the two actions yield equal payoffs. 
Thus, among the two strict equilibria, $e_1=(1,0)$ is \emph{risk-dominant} because $\BR(1/2,1/2)=\{e_1\}$. 
More generally, following \citet{Morris-etal-ECTA1995}, $e_1$ is $p$-dominant with $p\in(0,1)$ if $x_1\ge p$ implies $e_1 \in \BR(x_1,1-x_1)$. 
In the present game, $e_1$ is $t/(s + t)$-dominant. 
Since $p$-dominance implies $p'$-dominance for all $p' > p$ and risk dominance is equivalent to $1/2$-dominance, we also confirm from $t/(s + t) < 1/2$ that $e_1$ is the risk-dominant equilibrium of the game. 

We have the following general characterization including a selection result. 
\begin{proposition}[Global uniqueness in two-action coordination games]
    \label{prop:global-uniqueness-in-s-t-coordination}
Fix $k \ge 1$ and consider the coordination game \eqref{eq:A.2x2.coord}. 
The game has a unique $(k,\eta)$-SLE, which is globally asymptotically stable under \eqref{eq:sLD} for every $\eta>0$ if and only if the indifference threshold is weakly below $1/k$, that is, 
\begin{align}
    s \ge (k-1)t
    \quad \Longleftrightarrow \quad 
    \frac{t}{s+t} \le \frac{1}{k}.
    \label{eq:global-uniqueness-s-t}
\end{align}
When these conditions hold, the unique SLE converges to the risk-dominant equilibrium $e_1$ as $\eta \Toz$. 
If $s<(k-1)t$, then there are exactly three SLE for all sufficiently small $\eta>0$. 
\end{proposition}

\begin{figure}[t]
    \centering
    \includegraphics[width=.42\textwidth]{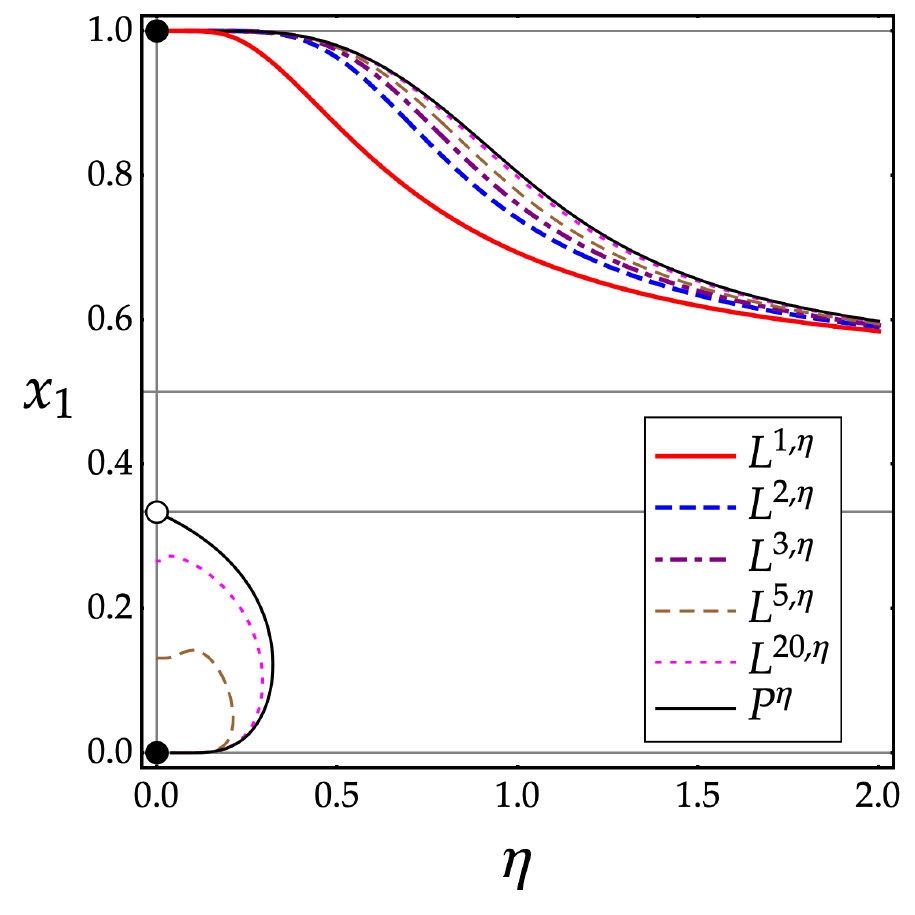}
    \caption{Sampling logit equilibria for $k \in \{1,2,3,5,20\}$ in the game $A = \begin{bsmallmatrix}
    2 & 0 \\ 0 & 1
    \end{bsmallmatrix}$.} 
    \label{fig:eta-cont}
\end{figure}

\Cref{fig:eta-cont} illustrates \cref{prop:global-uniqueness-in-s-t-coordination} for the case $(s,t)=(2,1)$. 
The uniqueness condition \eqref{eq:global-uniqueness-s-t} holds for $1 \le k \le 3$, whereas multiple SLE emerge for small $\eta$ when $k > 3$. 
The logit equilibrium curves are shown for reference. 
There are multiple logit equilibria for small $\eta$, and each converges to one of the Nash equilibria as $\eta$ approaches zero. 

The convergence to the risk-dominant equilibrium in \cref{prop:global-uniqueness-in-s-t-coordination} is a sampling counterpart of a result by \citet{Turocy-GEB2005}. 
In $2 \times 2$ coordination games with two strict equilibria, his Theorem 7 shows that the principal branch of quantal response functions converges to the risk-dominant equilibrium. 
The logit equilibrium curves in \cref{fig:eta-cont} illustrate this result, where the upper curve is the principal branch.\footnote{\citet{Leonardos-EC2023} provide a direct population-game application of this principal-branch selection result. In their full-information technology-adoption game, equivalent to \eqref{eq:A.2x2.coord} with $(s,t)=(1+\gamma,1-\gamma)$, temporarily increasing logit noise $\eta$ past the critical point for a saddle-node bifurcation (around $\eta \approx 0.321$ in \cref{fig:eta-cont}) and then lowering it selects the risk-dominant technology.} 

The threshold \eqref{eq:global-uniqueness-s-t} has an intuitive interpretation in terms of sampling. 
The condition places $t/(s + t)$ weakly below $1/k$, the first positive point on the support of empirical population states $\{0,1/k,2/k,\hdots,1\}$. 
Under condition \eqref{eq:global-uniqueness-s-t}, observing one action-$1$ player makes action $1$ at least as attractive as action $2$. 
This condition makes the sampled logit probabilities discretely concave on $\{0,1/k,2/k,\hdots,1\}$ for every $\eta>0$. 
Taking expectations preserves this concavity and prevents the equilibrium curve from folding back to create additional equilibria (cf. \cref{fig:choice-probs}). 
If the condition fails, by contrast, samples containing zero or one action-$1$ observation both favor action $2$ when noise is small. 
Both strict equilibria then generate nearby SLE, separated by an unstable interior SLE as \cref{fig:eta-cont} illustrates. 

The selection in \cref{prop:global-uniqueness-in-s-t-coordination} is tightly related to almost global convergence under sampling best response processes discussed in the literature. 
The strict condition $t/(s+t)<1/k$ implies that $e_1$ is $1/k$-dominant. 
In a finite-population model where revising agents best respond to $k$ sampled opponents, \citet{Sandholm-IJGT2001} shows that, for $k\ge2$, from any fixed positive initial share, the probability of convergence to the $1/k$-dominant strategy approaches one as the population grows. 
Also, because the degenerate sample size distribution concentrated on $k$ is ``$k$-good'' in the terminology of \citet{Oyama-Sandholm-Tercieux-TE2015}, their Theorem~1 implies that every trajectory of \eqref{eq:sBRD} starting from $x_1\in(0,1]$ converges to $e_1$ if condition \eqref{eq:global-uniqueness-s-t} holds strictly. 
For $k\in\{1,2\}$, the $1/k$-dominance condition is automatically satisfied, so their selection result describes the zero-noise limit of the unique SLE branch. 
\Cref{prop:global-uniqueness-in-s-t-coordination} extends this conclusion to the sampling logit dynamic: it precludes multiplicity for every $\eta>0$ and also covers the boundary $t/(s+t)=1/k$. 
At the boundary, the sampled best response is tied after exactly one action-$1$ observation, so the strict $1/k$-dominance result does not apply. 
The logit rule resolves the tie symmetrically and yields a unique prediction. In \cref{fig:eta-cont}, $k=3$ with $(s,t)=(2,1)$ illustrates this case as $1/(2 + 1) = 1/3$. 

The finite-population model of \citet{Kreindler-Young-GEB2013} is also worth mentioning in relation to \cref{prop:global-uniqueness-in-s-t-coordination}. 
For comparison, introduce the normalized coordinates 
\begin{align}
    \alpha = \frac{s}{t} - 1 > 0
    \quad \text{and} \quad 
    \beta = \frac{t}{\eta} > 0. 
    \label{eq:KY-parameter-correspondence}
\end{align}
With $\alpha$ and $\beta$, the sampling logit equilibrium problem in game \eqref{eq:A.2x2.coord} is equivalent to the problem with base payoff matrix $\begin{bsmallmatrix}1 + \alpha & 0 \\ 0 & 1 \end{bsmallmatrix}$ and the $1/\beta$-logit choice, which is precisely the problem considered in Kreindler and Young. 
Also, the condition \eqref{eq:global-uniqueness-s-t} reads $\alpha \ge k - 2$. 

Kreindler and Young consider a finite-population stochastic evolutionary model in which they call selection \emph{fast} when, starting from universal use of action~$2$, the expected number of revisions per capita required for action~$1$ to reach a share of $1/2$ (i.e., the hitting time) is bounded uniformly in the finite population size. 
Under the parameter correspondence \eqref{eq:KY-parameter-correspondence}, the mean-field dynamic for the finite-population model with partial information in their Section 6, in which agents observe a finite sample of opponents' actions, is nothing but the sampling logit dynamic \eqref{eq:sLD} for game \eqref{eq:A.2x2.coord}. 
For every fixed $\beta = t/\eta > 0$, their Theorem~3 proves fast selection for $k\geq3$ whenever $\alpha>\min\{h^*(\beta),k-2\}$, where $h^*(\beta) \ge 0$ is a simple analytical formula in $\beta$. 
We can see that the $\alpha > k - 2$ part corresponds to our condition \eqref{eq:global-uniqueness-s-t}. 
The other condition $\alpha > h^*(\beta)$ states that, for the SLE to be unique for a fixed pair of $\beta = t/\eta > 0$ and $k < \alpha + 2$, it is sufficient to require a large payoff gain $\alpha = (s/t) - 1$ for the innovative technology~1.

The mechanism in Kreindler and Young's results parallels the variance premium interpretation in \cref{sec:origins-of-premiums}. 
At low adoption rates, samples that overrepresent adopters raise the logit response more than samples that underrepresent them lower it. 
In our large-$k$ approximation for game \eqref{eq:A.2x2.coord}, action $1$ is suboptimal below the indifference point $t/(s + t)$ and therefore receives a positive relative variance premium by \cref{eq:2x2-delta-prop}. 
Their exact two-action analysis connects the mean field to finite-population waiting times, while our premium decomposition extends the intuition to general population games.\footnote{As stressed in \cite{Oyama-Sandholm-Tercieux-TE2015}, selection under mean-field sampling dynamics such as \eqref{eq:sBRD} or \eqref{eq:sLD} is ``fast'' in the sense that convergence is described by a deterministic ordinary differential equation, rather than requiring the long-run evolution of a stochastic process. This notion of ``fast'' selection, however, differs from that used in \cite{Kreindler-Young-GEB2013}.} 

\section{Concluding remarks}

This study introduced a static SLE framework and a virtual-payoff representation of finite-sampling distortions. 
For linear games, the critical joint limit $\sqrt{k}\eta\to\tau\in(0,\infty)$ characterizes local convergence and stability near a completely mixed Nash equilibrium satisfying payoff identification. 
Exact finite-sample results add sharp uniqueness, stability, and selection predictions. 

The scope of our large-$k$ results should be clarified as follows. 
For $0 < \eta \le 1$, \cref{thm:approximation}\cref{enum:error-bound} makes the large-$k$ approximation relatively accurate in the joint regime $k\eta^2\to\infty$, in which the sampling distortion, while systematic, is second order relative to logit noise. 
\Cref{thm:tau-lim.joint-limit} covers the critical scale $k\eta^2\to\tau^2\in(0,\infty)$ without relying on the approximation in \cref{thm:approximation}, but its equilibrium conclusion is local to a completely mixed Nash equilibrium and relies on linear payoffs. 
Neither result should therefore be extrapolated to the global selection effects for small $k$ discussed in \cref{sec:examples}.

One of the key features of \cite{Oyama-Sandholm-Tercieux-TE2015}, the immediate basis for this study, is the heterogeneity of sample sizes whereby the possible number of observations $k \ge 1$ is also a random variable. 
This possibility is abstracted away in this study because our emphasis is on the connection between the number of observations $k$ and the precision of decision making represented by $\eta$, as well as the resulting aggregate biases under positive $\eta$. 
Since this heterogeneity is key to their equilibrium selection results, understanding the role of randomness in $k$ in our context is an interesting extension. 
The examples in \cref{sec:examples} suggest that similar conclusions about equilibrium selection may be drawn because small-sample effects are key to Oyama et al.'s selection result. 
It is also important to explore whether other predictions based on the sampling best response dynamic \citep[e.g.,][]{Sawa-Wu-GEB2023,Arigapudi-etal-AEJMicro2024,Arigapudi-Heller-GEB2025} are robust to logit noise.

Finally, several apparent limitations of the present study should be noted. 
First, the general premium decomposition relies on the large-sample approximation, whereas the exact joint-limit analysis is restricted to linear games and local equilibrium behavior. 
Second, both the sample size $k$ and the noise level $\eta$ are treated as exogenous. 
A natural extension as a model of learning would be to endogenize these parameters by incorporating costs of information acquisition or cognitive effort. 
Finally, while our framework provides sharp predictions in two-action games, its analytical tractability in general $n$-action games, possibly in \emph{potential games} \citep{Sandholm-JET2001,Sandholm-JET2009}, \emph{stable games} \citep{Hofbauer-Sandholm-JET2009}, or \emph{supermodular games} \citep[Sec. 3.4]{Sandholm-Book2010}, warrants further investigation. 
The combination of logit choice with \emph{best experienced payoff dynamics} \citep{Osborne-Rubinstein-AER1998,Sethi-GEB2000,Sandholm-etal-JET2020} is also of interest.

\clearpage 
{\small \singlespacing 
\bibliographystyle{aer}
\bibliography{refs}
}

\clearpage 
\appendix

\crefalias{section}{appendix}

\renewcommand{\theproposition}{\Alph{proposition}}
\renewcommand{\theclaim}{\Alph{claim}}

\setstretch{1.2}
\section{Proofs}
\label{app:derivations}

\UlHeading{Preliminaries} 
We use Landau notation: 
$f(x)=O(g(x))$ as $x\to x_0$ means that there exist constants $C>0$ and $\delta>0$ such that 
$|f(x)|\le C|g(x)|$ whenever $0<|x-x_0|<\delta$. 
Also, we write $f(x)=o(g(x))$ as $x\to x_0$ if $f(x)/g(x)\to 0$ as $x\to x_0$.

\begin{proof}[Proof of \cref{thm:approximation}]
We first derive the approximation formula. 
The $\ell$th-order delta method approximates $L(x) = \BbE[P(w)]$ by plugging the $\ell$th-order Taylor approximation of $P(w)$ into the expectation. 
Different values of $\ell$ yield different formulas. 
Throughout, we assume the second-order method. 
The second-order approximation of $P_i(w)$ is
\begin{align}
    P_i^\star(w) 
    \Is   
    P_i(x) + \la P_i' (x), w - x \ra + \frac{1}{2} (w - x)^\top  P_i'' (x) (w - x), 
    \label{eq:P-expansion}
\end{align}
where $P_i'$ and $P_i''$ denote the gradient and the Hessian matrix of $P_i$, respectively. 

Then, the second-order delta method approximates $L_i(x) = \BbE[P_i(w)]$ by  
\begin{align}
    \TlL_i (x) 
    \Is 
    \BbE\left[P_i^\star(w)\right] 
    = 
    P_i(x) + \frac{1}{2k } \la P_i''(x) , \Sigma(x) \ra. 
    \label{eq:L_i}  
\end{align}
Here, for square matrices $A = [a_{kl}]$ and $B = [b_{kl}]$, $\la A, B \ra \Is \operatorname{trace}[A^\top B] = \sum_{k,l} a_{kl} b_{kl}$. 
This approximation is impossible for the sampling best response correspondence $\BR^k(x) = \BbE[\BR(w)]$ because $\BR(\cdot)$ is generally not differentiable. 

The approximation formula and part \ref{enum:validity} require only that $F$ admits a $C^2$ extension to an open neighborhood of $X$. 
Part \ref{enum:error-bound} uses the stated $C^4$ extension. 
Direct computation yields the approximation formulae \eqref{eq:TlL} as follows. 
Let $\vartheta \Is \eta^{-1} > 0$. 
For brevity, we write $p_i = P_i(x)$ and suppress the dependence of functions on $x$ from the notation below. 
The gradient of $P_i$ is given by  
\begin{align}
    P_i'
    & 
    = \vartheta p_i \cdot 
    \bigl(F_i' - F'^\top P\bigr) 
    = \vartheta p_i \WH{F_i'}
    = \vartheta p_i R_i
\end{align}
where we set $R_i \Is \WH{F_i'}$. 
The Hessian matrix of $P_i$ is 
\begin{align}
    P_i'' 
    = 
    \vartheta
    R_i P_i^{' \top} 
    + 
    \vartheta
    p_i R_i' 
    = 
    \vartheta^2 p_i R_i R_i^\top 
    + 
    \vartheta
    p_i R_i'. 
    \label{eq:app.P''}
\end{align}
We recall that $\OL{A}$ and $\WH{A}_i$ for a collection of matrices $\{A_i\}_{i\in S}$ are defined as $\OL{A} \Is \sum_{l\in S} p_l A_l$ and $\WH{A}_i = A_i - \OL{A}$. 
Then, the Jacobian matrix of $R_i$ is computed as 
\begin{align}
    R_i' 
    & 
    = \bigl( F_i' - \bigl(\sum_l p_l F_l'\bigr) \bigr)' 
    \\
    & = 
    - 
    \sum_l 
    \bigl(F_l'(P_l')^\top + p_l F_l''\bigr)
    + 
    F_i'' 
    \\
    & = 
    - 
    \vartheta
    \sum_l p_l F_l' R_l^\top
    + 
    F_i'' 
    - 
    \sum_l p_l  F_l''
    \\
    & =
    - 
    \vartheta
    \sum_l p_l (F_l' - \OL{F'}) R_l^\top 
    - 
    \vartheta
    \sum_l p_l \OL{F'} R_l^\top 
    +
    \WH{F_i''}
    \label{eq:app.Ri'-comp}
    \\
    & =
    - 
    \vartheta
    \sum_l p_l R_l R_l^\top 
    +
    \WH{F_i''}, 
\end{align}
where $\sum_l p_l \OL{F'} R_l^\top =  \OL{F'} \sum_l p_l (F'_l - \OL{F'})^\top = O$ in \cref{eq:app.Ri'-comp}. 
Together with \cref{eq:app.P''}, 
\begin{align}
    P_i''
    & = 
    \vartheta^2 p_i 
    \Bigl(R_i R_i^\top 
    - 
    \sum_l p_l R_l R_l^\top 
    \Bigr) 
    + 
    \vartheta p_i 
    \WH{F_i''}.
    \label{eq:P''} 
\end{align}
We now compute $\la P_i'', \Sigma \ra$ in \cref{eq:L_i}. 
In doing so, we note that, from the identity $\la b b^\top,  A  \ra = \sum_{i,j} b_i b_j a_{ij} = b^\top A b$ for any $b = (b_i)\in\BbR^n$ and $A = [a_{ij}]\in \BbR^{n\times n}$, 
\begin{align}
    & 
    \Bigl\la 
    \Bigl(R_i R_i^\top 
    - 
    \sum_l p_l  R_l R_l^\top 
    \Bigr) 
    , 
    \Sigma 
    \Bigr\ra
    = 
    R_i^\top \Sigma R_i  
    - 
    \sum_l p_l  R_l^\top \Sigma R_l 
    = 
    \WH{ R_i^\top \Sigma R_i }, 
    \\ 
    & 
    \bigl\la 
    \WH{F_i''} 
    , 
    \Sigma 
    \bigr\ra
    = 
    \bigl\la 
    F_i''
    , 
    \Sigma 
    \bigr\ra
    -
    \bigl\la 
    \sum_l 
    p_l  
    F_l''
    , 
    \Sigma 
    \bigr\ra
    = 
    \bigl\la 
    F_i''
    , 
    \Sigma 
    \bigr\ra
    -
    \sum_l 
    p_l  
    \bigl\la 
    F_l''
    , 
    \Sigma 
    \bigr\ra
    = 
    \WH{\la F_i'',\Sigma \ra}. 
\end{align}
Plugging these into \cref{eq:P''}, we obtain 
\begin{align}
    \la P_i'', \Sigma \ra
    & =  
    \vartheta^2 p_i 
    \, 
    \WH{R_i^\top \Sigma R_i}
    + 
    \vartheta p_i 
    \, 
    \WH{\la F_i'',\Sigma \ra}. 
\end{align}
Then, from \cref{eq:L_i} and $\vartheta = \eta^{-1}$, we obtain 
\begin{align}
    \TlL_i(x) 
    = 
    \left(
        1 
        + 
        \frac{1}{2k\eta^2}
        \WH{R_i^\top \Sigma R_i}
        + 
        \frac{1}{2k\eta}
        \WH{\la F_i'',\Sigma \ra}
    \right)
    P_i(x), 
\end{align}
which implies \cref{eq:TlL} as we define $v_i \Is \frac{1}{2k\eta^2} R_i^\top \Sigma R_i$ and $q_i \Is \frac{1}{2k\eta} \la F_i'', \Sigma \ra$. 

\UlHeading{Part \ref{enum:validity}}   
Since $X$ is compact and $F$ admits a $C^2$ extension to an open neighborhood of $X$, there is $C_0 > 0$, depending only on uniform bounds for $F'$ and $F''$ over $X$ and independent of $k$ and $\eta$, satisfying 
\begin{align}
    \sup_{x\in X} \max_{i\in S} 
    \left| \WH{v}_i(x) + \WH{q}_i(x) 
    \right| 
    \le \frac{C_0}{k} \left(\eta^{-2} + \eta^{-1}\right). 
\end{align}
Then, for each fixed $\eta > 0$, we can choose $k^*(\eta) < \infty$ such that 
\begin{align}
    & \frac{C_0}{k} \left(\eta^{-2} + \eta^{-1}\right) < 1 
    & \forall k \ge k^*(\eta)
\end{align}
For $k \ge k^*(\eta)$, $1 + \WH{v}_i(x) + {\WH{q}}_i(x) > 0$ and hence $\TlL_i(x) = (1 + \WH{v}_i(x) + {\WH{q}}_i(x)) P_i(x) > 0$ for all $i\in S$ and $x \in X$. 
From the identity $\sum_{i\in S} \TlL_i(x) = \sum_{i\in S} (1 + \WH{v}_i(x) + {\WH{q}}_i(x)) P_i(x) = \sum_{i\in S} P_i(x)= 1$, positivity of $\TlL$ implies $\TlL_i(x) < 1$ for all $i\in S$. 
Thus, $\TlL(x) \in X$ for all $k > k^*(\eta)$. That is, for any fixed $\eta > 0$, $\TlL$ is a valid choice rule for sufficiently large $k$. 

\UlHeading{Part \ref{enum:error-bound}} 
Below, $\|\cdot\|$ denotes any fixed norm on the relevant finite-dimensional spaces, and multilinear maps are equipped with the corresponding induced operator norms. 
The constants below may depend on this choice of norm, but only up to constant factors, so that the asymptotic orders in $k$ and $\eta$ are unaffected. 
We now show the following claim, which is a slightly more detailed version of part \ref{enum:error-bound}. 
\begin{claim}[Uniform error bounds]
\label{prop:error-estimate}
Assume that $F$ admits a $C^4$ extension to an open neighborhood of $X$. 
Then, there exists a constant $C>0$ such that
\begin{align}
& 
\Delta \Is 
\sup_{x\in X}
\left\|L(x)-\TlL(x)\right\|
\le
\frac{C}{k^2}\sum_{m=1}^4\eta^{-m}
\le
\begin{cases}
    4Ck^{-2}\eta^{-4} & \text{if } 0 < \eta \le 1 \\ 
    4Ck^{-2}\eta^{-1} & \text{if } \eta > 1
\end{cases}
& 
\forall k\ge1. 
\label{eq:raw-error-estimate}
\end{align}
Hence, the approximation error is negligible relative to $k^{-1}\eta^{-2}$ in the regime
$k\eta^2\to\infty$ for $\eta\le1$, and relative to $k^{-1}\eta^{-1}$ 
when $k \to \infty$ 
for $\eta>1$. 
\end{claim}

\noindent\textbf{Proof.}\xspace 
Let $\xi\Is w - x$. By construction, $\mathbb E[\xi]=0$ and $\mathbb E[\xi\xi^\top]=\frac{1}{k}\Sigma(x)$.
Classical formulas for multinomial cumulants through order four \citep{Wishart-Biometrika1949} imply that every third centered joint moment of $kw$ is $O(k)$ and every fourth centered joint moment is $O(k^2)$. 
For the latter, fourth moments are sums of fourth cumulants and products of second cumulants. 
Since the number of actions is fixed and $\xi=(kw-kx)/k$, these bounds hold uniformly over $x\in X$ and yield
\begin{align}
\sup_{x\in X}\bigl\|\mathbb E[\xi^{\otimes 3}]\bigr\|=O(k^{-2}) 
\quad\text{and}\quad 
\sup_{x\in X}\mathbb E[\|\xi\|^4]=O(k^{-2}).
\label{eq:moment-bounds}
\end{align}
Next, we recall $P=\Lambda^\eta\circ F$, where $\Lambda^\eta$ is the $\eta$-logit map. Since $F$ has a $C^4$ extension to a neighborhood of $X$ and the $m$th derivatives of $\Lambda^\eta$ are of order $\eta^{-m}$, there is a constant $C'>0$ satisfying 
\begin{align}
& \sup_{x\in X}\|D^mP(x)\|
\le
C' \sum_{r=1}^m \eta^{-r},
& m=3,4. 
\label{eq:derivative-bounds}
\end{align}
Because $x,w\in X$ and $X$ is convex, the line segment joining $x$ and $w=x+\xi$ is contained in $X$. Applying Taylor's theorem of order three to each coordinate of $P(x+\xi)$ around $x$ and taking expectations, we obtain
\begin{align}
L(x)-\TlL(x)
=
\frac{1}{6} \la D^3 P(x), \mathbb E[\xi^{\otimes 3}] \ra 
+
R(x),
\end{align}
where the remainder $R$ satisfies
\begin{align}
\|R(x)\|
\le
C'' \sup_{y\in X}\|D^4 P(y)\|\,\mathbb E[\|\xi\|^4]
\end{align}
for some constant $C''>0$.
Combining this with \cref{eq:moment-bounds,eq:derivative-bounds} gives the first inequality of \cref{eq:raw-error-estimate} with some constant $C > 0$. 
We have $\sum_{m=1}^4\eta^{-m}\le 4\eta^{-4}$ 
if $0 < \eta\le 1$ and $\sum_{m=1}^4\eta^{-m}\le 4\eta^{-1}$ if $\eta>1$, showing the second inequality of \cref{eq:raw-error-estimate}. \hfill 

\Cref{fig:errors} illustrates how the accuracy of the approximation deteriorates when $k$ is small, especially for relatively small $\eta$. 
\end{proof}
\begin{figure}[tb]
    \centering 
    \includegraphics[width=.325\hsize]{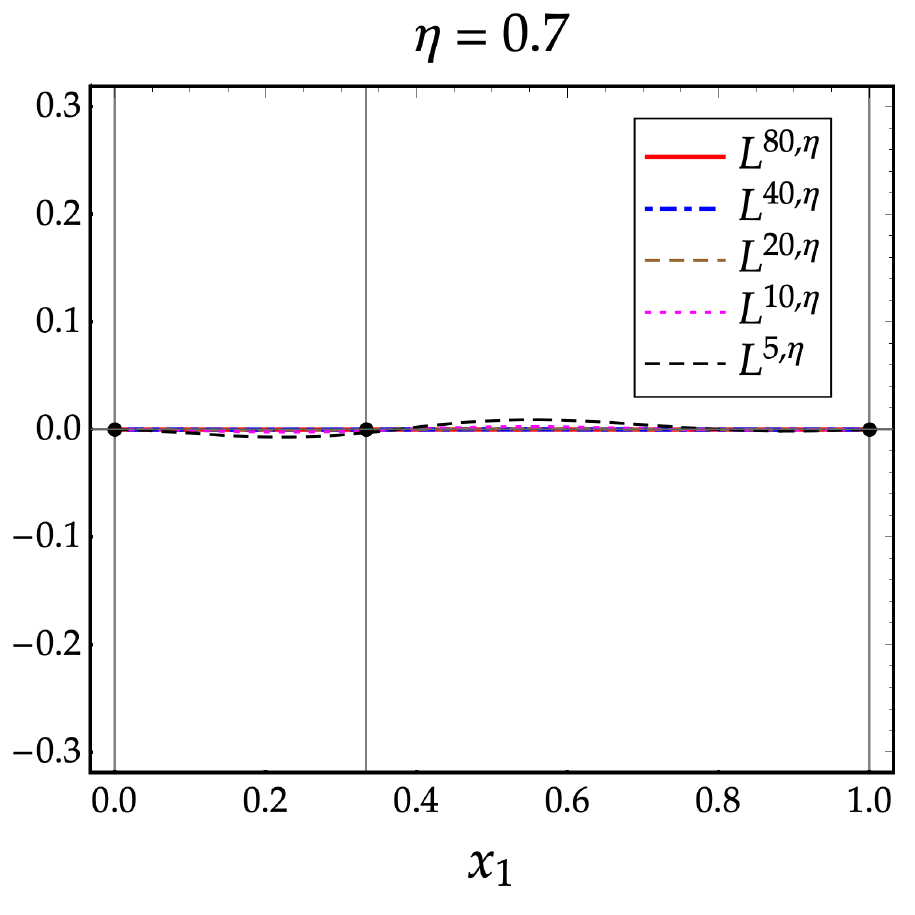}
    \hfill
    \includegraphics[width=.325\hsize]{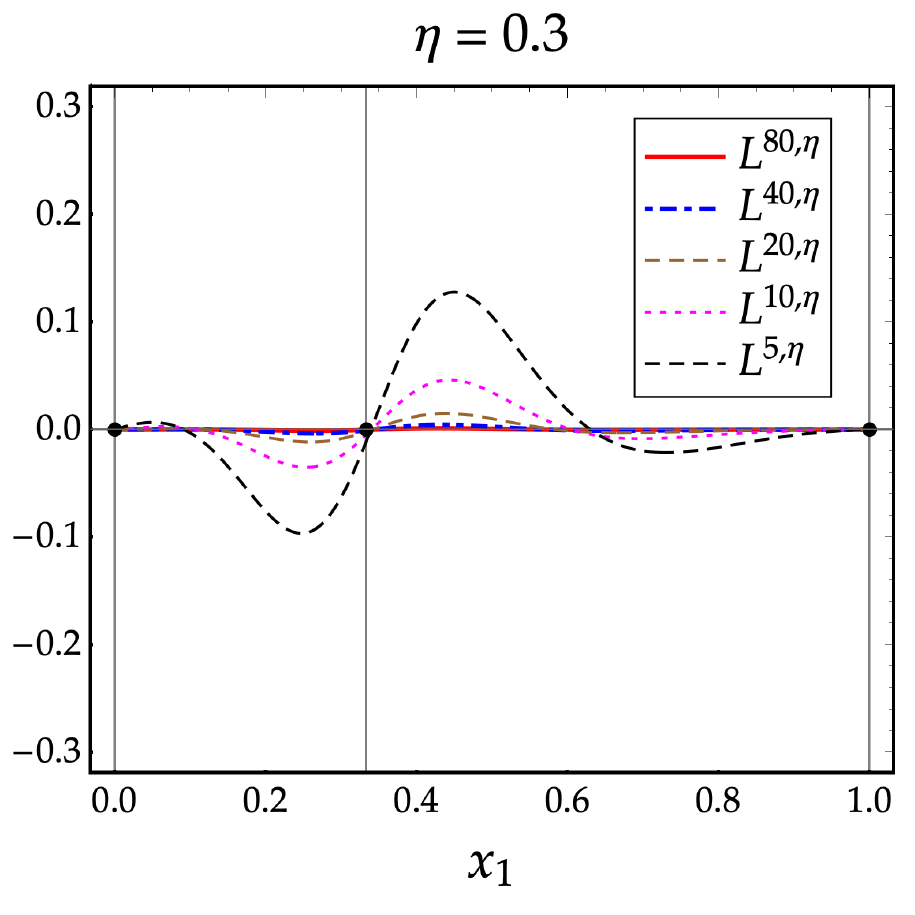}
    \hfill
    \includegraphics[width=.325\hsize]{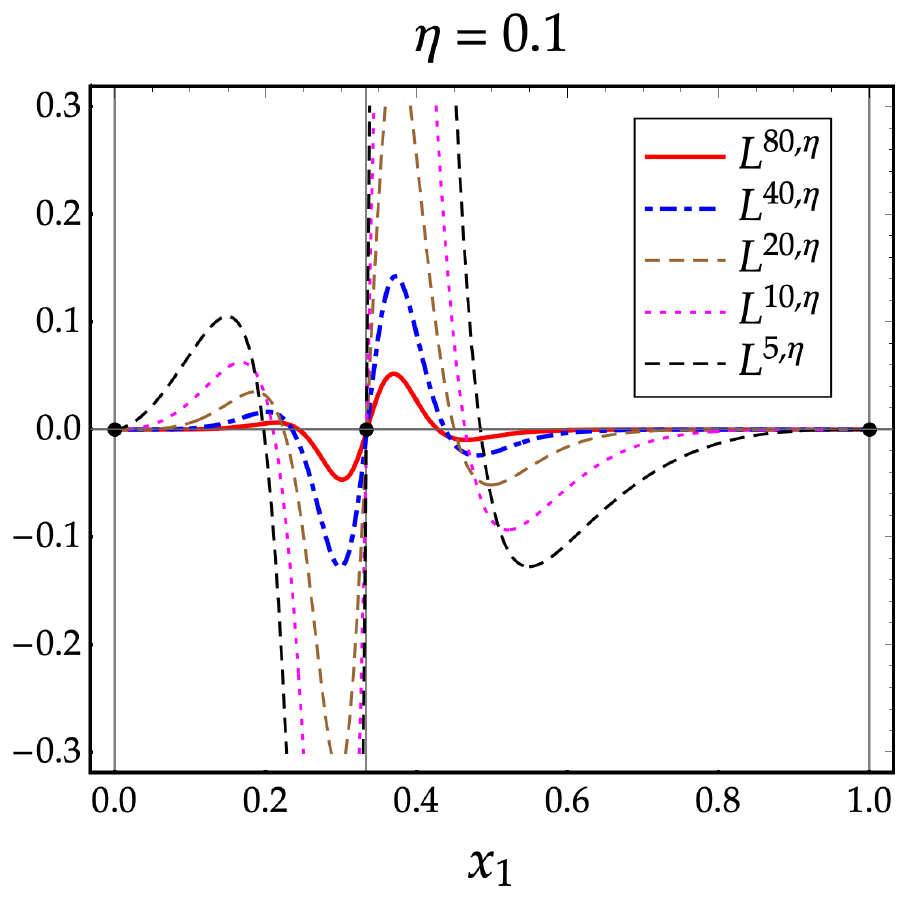}
    \caption{Approximation errors for the choice probability of action $1$ ($L_1 - \TlL_1$) under different $\eta$ and $k$. 
    The coordination game \eqref{eq:A.2x2.coord} with $(s,t) = (2,1)$ as considered in \cref{fig:choice-probs,fig:eta-cont} is assumed. 
    The black markers on the horizontal axis show the Nash equilibria.} 
    \label{fig:errors}
\end{figure}

\begin{proof}[Proof of \Cref{obs:TlL-interpretation}]
Set $Z(x) \Is \sum_{j \in S} \exp(\eta^{-1} F_j(x))$ and $\delta \Is \WH{v} + \WH{q}$. 
Then, 
\begin{align}
    \sum_{j \in S} (1 + \delta_j(x)) \exp \left(\eta^{-1} F_j(x)\right)
    & = 
    Z(x)
    \sum_{j \in S} (1 + \delta_j(x)) \frac{\exp \left(\eta^{-1} F_j(x)\right)}{Z(x)}
    \\ 
    & = 
    Z(x) 
    \sum_{j \in S} (1 + \delta_j(x)) P_j(x)
    = Z(x)
\end{align}
since $\sum_{j \in S} \delta_j(x) P_j(x) = 0$. 
Define $\TlF_i(x) \Is F_i(x) + \eta \log (1 + \delta_i(x))$ assuming large $k$ such that $\min_{i\in S} \min_{x\in X} \{1 + \delta_i(x)\} > 0$. 
Then, 
\begin{align}
    \frac
        {\exp(\eta^{-1} \TlF_i(x))}
        {\sum_{j \in S}\exp(\eta^{-1} \TlF_j(x)) }
    & 
    = 
    \frac
        {(1 + \delta_i(x)) \exp\left(\eta^{-1} F_i(x)\right)}
        {\sum_{j \in S} (1 + \delta_j(x)) \exp\left(\eta^{-1} F_j(x)\right) }\\
    & 
    = 
    \frac
        {(1 + \delta_i(x)) \exp\left(\eta^{-1} F_i(x)\right)}
        {Z(x)} 
    \\
    & 
    = 
    (1 + \delta_i(x)) P_i(x) 
    = \TlL_i(x). 
\end{align}
Thus, the condition $x = \TlL(x)$ is nothing but the $\eta$-logit equilibrium condition for the virtual payoff $\TlF$, showing part \cref{enum:virtual-game-1}. 
A direct application of known results for logit choice \citep[e.g.,][Appendix]{Hofbauer-Sandholm-JET2007} yields part \cref{enum:virtual-game-2}. 
\end{proof}

\begin{proof}[Proof of \cref{thm:equilibrium-distortion-general}]
For fixed $\eta$, let
$p_*\Is\min_{i\in S,\,w\in X}P_i(w)>0$.
Then, every $x\in\Fix(L)$ satisfies
$x_i=L_i(x)=\BbE[P_i(w)]\ge p_*$. Moreover,
$\TlL=P+O(k^{-1})$ uniformly on $X$, so
$\TlL_i(x)\ge p_*/2$ for every $i$ and $x$ when $k$ is sufficiently
large. Hence, all fixed points and the neighborhoods used below lie in the relative interior of $X$.

\UlHeading{Part \cref{enum:equilibrium-bound-1}} 
For a fixed norm, define the fixed-point error modulus
\begin{align}
    \Tlomega(t)
    \Is
    \max\left\{
    d\bigl(x,\Fix(\TlL)\bigr):
    x\in X,\ 
    \|x-\TlL(x)\|\le t
    \right\}.
    \label{eq:residual-modulus}
\end{align}
For any $x\in\Fix(L)$, $\|x-\TlL(x)\| =\|L(x)-\TlL(x)\| \le\Delta$.
The definition of $\Tlomega$ therefore gives
\cref{eq:one-sided-equilibrium-bound}. 
More importantly, $\lim_{t \Toz}\Tlomega(t) = 0$. 
If $\Tlomega(t)$ did not converge to zero as $t \Toz$, there would be an $a>0$, a sequence $t_j\Toz$, and corresponding $y^j\in X$ such that
\begin{align}
    \|y^j-\TlL(y^j)\|\le t_j
    \quad\text{and}\quad
    d\bigl(y^j,\Fix(\TlL)\bigr)\ge a.
\end{align}
Along a convergent subsequence, $y^j\to y\in X$.
From continuity, $y=\TlL(y)$ and thus $y \in \Fix(\TlL)$, which is a contradiction because 
$0 = d(y,\Fix(\TlL))\ge a > 0$.

\UlHeading{Part \cref{enum:equilibrium-bound-2}}  
Define the residual map $R(x)\Is x-\TlL(x)$ for the approximated fixed-point problem. 
Let $T\Is\{h\in\BbR^n:\sum_{i\in S}h_i=0\}$ be the tangent space of the affine hull of $X$.
Also, let 
\begin{align}
    R_T'(x) 
    \Is
    I-\TlL_T'(x), 
\end{align}
where $R_T'$ and $\TlL_T'$ denote the derivatives restricted to $T$. 
By regularity of every $\Tlx^r$, every $R_T'(\Tlx^r)$ is nonsingular, and 
\begin{align}
    \gamma
    \Is
    \min_{r\in\{1,2,\hdots, m\}}
    \ 
    \inf_{h\in T, \|h\|=1}
    \|R_T'(\Tlx^r) h\|
    >0. 
    \label{eq:equilibrium-conditioning}
\end{align}
Around each $\Tlx^r$, we have $R(\Tlx^r) = \Vt0$ by definition, and up to the first order, we have 
\begin{align}
    & 
    R(\Tlx^r+h)
    =
    R(\Tlx^r) + R_T'(\Tlx^r) h + o(\|h\|) 
    =
    R_T'(\Tlx^r) h + o(\|h\|)
    & 
    \forall h \in T.
    \label{eq:residual-linearization}
\end{align}
We can therefore choose (relative) open balls $\{B_r\}$ in $X$ that are pairwise disjoint, each $B_r$ centered at $\Tlx^r$, with closures in the relative interior of $X$, satisfying 
\begin{align}
    & \|R(x)\|\ge\frac{\gamma}{2}\|x-\Tlx^r\|
    & \forall x\in B_r. 
    \label{eq:local-residual-bound}
\end{align}

For completeness, we briefly recall the degree-theoretic facts used below. 
See, e.g., \citet[Chap.~1]{Deimling-Book1985} for a concise reference. 
After fixing coordinates on $T$, the \emph{Brouwer degree}
$\deg(f,U,\Vt0)$ is defined for any continuous map $f: \operatorname{cl}(U) \to T$ defined on a bounded open set $U\subset T$ such that $\Vt0\notin f(h)$ for any $h \in \partial U$, where $\operatorname{cl}(U)$ denotes the closure and $\partial U$ denotes the boundary.
If $f$ is continuously differentiable and has a unique zero $x^*$ in $U$ with nonsingular derivative, then
$\deg(f,U,\Vt0)=\sgn\det f'(x^*)$, which is the \emph{local Brouwer degree}, or index, of $x^*$. 
The degree is invariant under any continuous homotopy
that remains nonzero on $\partial U$ and a nonzero degree implies the existence of a zero in $U$.

In our problem, the local Brouwer degree for $R$ and $B_r$ satisfies
\begin{align}
    \deg(R,B_r,\Vt0) = \sgn\det R_T'(\Tlx^r) \ne0
\end{align}
for each $r\in\{1,2,\hdots, m\}$ because $R_T'(\Tlx^r)$ is nonsingular by assumption. 
Let
\begin{align}
    \overline\Delta\Is
    \min_{r\in\{1,2,\hdots, m\}}\min_{x\in\partial B_r}\|R(x)\|>0
\end{align}
where $\partial B_r$ denotes the boundary of $B_r$. 
If $\sup_{x \in X} \| \TlL(x) - L(x) \| = \Delta<\overline\Delta$, the homotopy
\begin{align}
    & R_s(x)\Is R(x)+s\left(\TlL(x)-L(x)\right) 
    \quad \text{for} \quad 
    s\in[0,1],
\end{align}
does not vanish on any boundary $\partial B_r$ because
$\|R_s(x)\|\ge \overline\Delta-\Delta>0$. 
By homotopy invariance of local Brouwer degree, this implies  
\begin{align}
    \deg(R_1,B_r,\Vt0)=\deg(R,B_r,\Vt0)\ne0 
\end{align}
by setting $s = 1$. 
Because $R_1(x) = x - L(x)$, $\deg(R_1,B_r,\Vt0)\ne 0$ shows that $L$ has a fixed point in every $B_r$. 
Choose one such point $x^r$ in each $B_r$. 
Then, these $m$ points are distinct because the balls are disjoint. 
Moreover, for every $r$, from \cref{eq:local-residual-bound}, 
\begin{align}
    \frac{\gamma}{2}\|x^r-\Tlx^r\|
    &\le \|R(x^r)\|
    =\|x^r-\TlL(x^r)\|
    =\|L(x^r)-\TlL(x^r)\|
    \le\Delta. 
\end{align}
This proves the first bound in part~\ref{enum:equilibrium-bound-2}. 
If furthermore $\Fix(L)$ has exactly $m$ elements, there is exactly one in each ball and none outside, so these points exhaust $\Fix(L)$ and the claimed Hausdorff bound follows.
\end{proof}

\begin{proof}[Proof of \cref{cor:2x2-delta}]
Let $A_1 = (a, b)^\top$ and $A_2 = (c,d)^\top$ and define $A_0 \Is A_1 - A_2 = (a - c, b - d)^\top$. 
For brevity, write $p_1 = P_1(x)$ and $p_2 = P_2(x)$. 
With $R_i \Is \WH{F_i'}(x)$, 
\begin{align}
    & R_1 
    = A_1 - (p_1 A_1 + p_2 A_2) 
= p_2 (A_1 - A_2) 
    = p_2 A_0, 
    \\
    & R_2 
    = A_2 - (p_1 A_1 + p_2 A_2) 
= - p_1 (A_1 - A_2) 
    = - p_1 A_0. 
\end{align} 
Let 
$\sigma_A(x) 
\Is A_0^\top \Sigma A_0 
= 
((a - c) - (b - d))^2 \cdot x_1 x_2 \ge 0$, 
with equality iff either $a - c = b - d$ or $x_1x_2 = 0$. 
Then, we confirm 
\begin{align}
    & \sigma_1 = R_1^\top \Sigma R_1 = p_2^2 A_0^\top \Sigma A_0 = p_2^2 \sigma_A,\\
    & \sigma_2 = R_2^\top \Sigma R_2 = p_1^2 A_0^\top \Sigma A_0 = p_1^2 \sigma_A,
    \\ 
    & 
    p_1 \sigma_1 + p_2 \sigma_2 
    = 
    p_1 p_2^2 \sigma_A + p_2 p_1^2 \sigma_A 
    = (p_1 + p_2) p_1 p_2 \sigma_A 
    = p_1 p_2 \sigma_A. 
\end{align}
Thus, 
$\WH{\sigma}_1(x) = p_2(p_2 - p_1) \sigma_A$ and $\;\WH{\sigma}_2(x) = p_1(p_1 - p_2) \sigma_A$. 
\end{proof}

\begin{proof}[Proof of \cref{prop:approx-coord}]
Write the state as $x=(y,1-y)$ with $y = x_1$, and use $F_i(y)$, $L_i(y)$, $\TlL_i(y)$, and $\sigma_A(y)$ as shorthand for the corresponding functions evaluated at $(y,1-y)$. 
We first express the fixed-point condition for the approximated rule $\TlL$ in a normalized local coordinate around $x^*$. 
Since $F_1(y)-F_2(y)=\lambda(y-x^*)$, set
\begin{align}
    u \Is \frac{\lambda(y-x^*)}{\eta},
\end{align}
so that $y(u)=x^*+\eta u/\lambda$ and we can use the normalized logit map for $u$: 
\begin{align}
    P_1(y(u)) = \pi(u) \Is \frac{\exp(u)}{1+\exp(u)}. 
\end{align}
As $F$ is linear, the curvature term $q$ in \cref{eq:def.delta}
vanishes. Moreover, \cref{cor:2x2-delta} gives $\sigma_A(y)=\lambda^2y(1-y)$ and
\begin{align}
    \widehat v_1(y(u))
    =\theta\sigma_A(y(u))
    (1-\pi(u))(1-2\pi(u)).
\end{align}
Thus, the fixed-point condition $y-\TlL_1(y)=0$ becomes
\begin{align}
    G(u,\eta,\theta)
    \Is
    y(u)
    -\pi(u)
    -\theta\sigma_A(y(u))
    \pi(u)(1-\pi(u))(1-2\pi(u))
    =0.
    \label{eq:proof-scaled-two-action-fp}
\end{align}
Take
$u_0\Is\pi^{-1}(x^*)=\log\frac{x^*}{1-x^*}$, so that
$\pi(u_0)=x^*$. At $(u,\eta,\theta)=(u_0,0,0)$,
\begin{align}
    G(u_0,0,0)
    & 
    = x^* - \pi(u_0) = 0
    \quad \text{and}
    \\
    G_u(u_0,0,0)
    & 
    =-\pi'(u_0)
    =-x^*(1-x^*)\ne0. 
\end{align}
The implicit function theorem therefore yields a unique smooth function
$u(\eta,\theta)$ in a neighborhood of $(0,0)$ such that
$G(u(\eta,\theta),\eta,\theta)=0$ and $u(0,0)=u_0$.
Then, in the original coordinate,  
\begin{align}
    \Tlx(\eta,\theta)
    &=y\bigl(u(\eta,\theta)\bigr)
    =x^*+\frac{\eta}{\lambda}u(\eta,\theta)
    \label{eq:fixed-point-scaled-coordinate}
\end{align}
for $\eta>0$ and $\theta\ge0$. 
Because $x^*\in(0,1)$, this fixed point is interior for all sufficiently small $(\eta,\theta)$. 
It is locally unique because $u(\eta,\theta)$ is the unique solution near $u_0$. 
Formula \eqref{eq:fixed-point-scaled-coordinate} also defines a smooth extension to $\eta=0$ where $\Tlx(0,\theta)=x^*$.

Differentiating \cref{eq:fixed-point-scaled-coordinate} at the origin gives
\begin{align}
    \partial_\eta\Tlx(0,0)
    &=\frac{u(0,0)}{\lambda} = \frac{u_0}{\lambda} 
    =\frac{1}{\lambda}\log\frac{x^*}{1-x^*},
    \\
    \partial_\theta\partial_\eta\Tlx(0,0)
    &=\frac{u_\theta(0,0)}{\lambda}.
\end{align}
Implicit differentiation of $G(u(\eta,\theta),\eta,\theta)=0$ yields
$u_\theta(0,0) = -  G_\theta(u_0,0,0) / G_u(u_0,0,0)$. 
Since the required derivatives are
\begin{align}
    G_\theta(u_0,0,0)
    &=-\sigma_A(x^*)x^*(1-x^*)(1-2x^*),
    \\
    G_u(u_0,0,0)
    &=-x^*(1-x^*), 
\end{align}
it follows that
\begin{align}
    \partial_\theta\partial_\eta\Tlx(0,0)
    &= \frac{u_\theta(0,0)}{\lambda} = -\frac{\sigma_A(x^*)}{\lambda}(1-2x^*).
\end{align}
For the sampling effect, the smoothness of $u$ and \cref{eq:fixed-point-scaled-coordinate} imply
\begin{align}
    \frac{\Tlx(\eta,\theta)-\Tlx(\eta,0)}{\eta\theta}
    &=
    \frac{1}{\lambda}
    \int_0^1 u_\theta(\eta,t\theta)\,\mathrm{d}t
    \longrightarrow
    \frac{u_\theta(0,0)}{\lambda}
\end{align}
as $(\eta,\theta)\to(0,0)$ through positive pairs. 
This establishes \cref{eq:approx-coord-displacement}. 
Because $x^*<1/2$, $u_0<0$, and hence $u(\eta,\theta)<0$ for
all sufficiently small $(\eta,\theta)$. 
Thus, $\Tlx-x^*=\eta u(\eta,\theta)/\lambda$ has the sign opposite to $\lambda$. 

The discussion above characterizes the approximated fixed point $\Tlx$. 
To compare the approximated fixed point with a fixed point of $L$, we now restrict attention to admissible positive pairs for which
$k=(2\theta\eta^2)^{-1}\in\mathbb N$.
Define the scalar residuals
\begin{align}
    \TlR(y)\Is y-\TlL_1(y)
    \quad\text{and}\quad
    R(y)\Is y-L_1(y),
\end{align}
and write $\Tlu\Is u(\eta,\theta)$, so that
$\Tlx=y(\Tlu)$ and $\TlR(y(u))=G(u,\eta,\theta)$. At the approximated
fixed point 
$G_u(\Tlu,\eta,\theta)
\to G_u(u_0,0,0)
=-x^*(1-x^*)<0$ as $(\eta,\theta)\to(0,0)$. 
By continuity of $G_u$ and $u(\eta,\theta)$, there are
constants $c_0>0$ and $\epsilon>0$ such that 
\begin{align}
    G_u(s,\eta,\theta)\le -c_0
    \quad\text{whenever}\quad
    |s-\Tlu|\le\epsilon 
    \label{eq:proof-Gu-uniform-bound}
\end{align}
for all sufficiently small $(\eta,\theta)$. 

As we focus on small $\eta$, we may assume that $0<\eta\le1$. 
Applying \cref{eq:theorem-error-bound} with the maximum norm and substituting
$k=(2\theta\eta^2)^{-1}$ gives
\begin{align}
    \sup_{y\in[0,1]}|R(y)-\TlR(y)|
    &\le
    \sup_{x\in X}\|L(x)-\TlL(x)\|_\infty
    \le Ck^{-2}\eta^{-4}
    =4C\theta^2,
    \label{eq:proof-two-action-rule-transfer}
\end{align}
where $C$ is independent of $\eta$ and $\theta$. 
We use this uniform bound to bracket a zero of $R$, and hence a true SLE, around $\Tlx$.

Fix $c_1>4C/(c_0|\lambda|)$ and set
$y_\pm\Is\Tlx\pm c_1\eta\theta^2 = x^* + (\eta/\lambda) \Tlu \pm c_1\eta\theta^2$. 
Because $c_1$ is fixed, $\Tlx\to x^*$, $\eta\theta^2\to0$, and $\theta^2\to0$, we may restrict the parameter neighborhood for $(\eta,\theta)$ further so that
\begin{align}
    c_1|\lambda|\theta^2<\epsilon
    \quad\text{and}\quad
    |\Tlx-x^*|+c_1\eta\theta^2
    <\min\{x^*,1-x^*\}, 
\end{align}
where $\epsilon > 0$ is the same as in \cref{eq:proof-Gu-uniform-bound}. 
The second restriction implies $0<y_-<y_+<1$ and ensures that the implied states are in fact valid population states. 
In terms of $u$, we have 
\begin{align}
    u_\pm
    \Is\frac{\lambda(y_\pm-x^*)}{\eta}
    = \frac{\lambda(x^* + (\eta/\lambda) \Tlu \pm c_1 \eta \theta^2 - x^*)}{\eta}
    = \Tlu\pm c_1\lambda\theta^2.
\end{align}
The first restriction ensures that the entire segment from $\Tlu$ to either
$u_-$ or $u_+$ lies in the neighborhood from
\cref{eq:proof-Gu-uniform-bound}. Because $\TlR(\Tlx)=G(\Tlu,\eta,\theta)=0$,
the fundamental theorem of calculus gives
\begin{align}
    \TlR(y_\pm)-\TlR(\Tlx)
    &=G(u_\pm,\eta,\theta)-G(\Tlu,\eta,\theta)
    =\int_{\Tlu}^{u_\pm}G_u(s,\eta,\theta)\,\mathrm{d}s.
\end{align}
By \cref{eq:proof-Gu-uniform-bound}, the integrand has a constant sign, so
\begin{align}
    |\TlR(y_\pm)|
    &=\int_{\min\{\Tlu,u_\pm\}}^{\max\{\Tlu,u_\pm\}}
      |G_u(s,\eta,\theta)|\,\mathrm{d}s 
    \ge c_0|u_\pm-\Tlu|
    =c_0c_1|\lambda|\theta^2
    >4C\theta^2,  
    \label{eq:TlR-bound}
\end{align}
where the last inequality follows from the assumption that $c_1>4C/(c_0|\lambda|)$. 
Moreover, $u_--\Tlu=-(u_+-\Tlu)$, while $G_u$ has the same sign on both
integration segments. Hence $\TlR(y_-)$ and $\TlR(y_+)$ have opposite
signs. By \cref{eq:proof-two-action-rule-transfer},
\begin{align}
    |R(y_\pm)-\TlR(y_\pm)|
    \le
    \sup_{y\in[0,1]}|R(y)-\TlR(y)|
    \le 
    4C\theta^2
    <|\TlR(y_\pm)|,
\end{align}
so $R(y_\pm)$ must have the same sign as $\TlR(y_\pm)$. Thus, $R(y_-)$ and $R(y_+)$ have opposite signs. 
Continuity of $R$ then yields a zero $\Brx$ between $y_-$ and $y_+$. 
Equivalently, $\Brx$ is an interior fixed point of $L$. 
For each admissible parameter pair, choose any such zero. The construction gives $|\Brx-\Tlx| \le c_1\eta\theta^2$, and hence $|\Brx-\Tlx| / (\eta\theta) \le c_1\theta\Toz$.
This proves \cref{eq:approx-coord-true-distance}. 
Combining this estimate with \cref{eq:approx-coord-displacement} also shows that $\Brx$ has the same leading-order sampling displacement from $\Tlx(\eta,0)$.

It remains to locate the true fixed point relative to the mixed Nash equilibrium. 
For the approximate equilibrium $\Tlx$, since $\Tlu\to u_0 = \log \frac{x^*}{1 - x^*} \ne0$, for sufficiently small parameters, we have $|\Tlu| = |u(\eta,\theta)| \ge |u_0| - |\Tlu - u_0| > |u_0|/2$ and hence 
\begin{align}
    |\Tlx-x^*|
    =\frac{\eta}{|\lambda|}|\Tlu|
    \ge\frac{|u_0|}{2|\lambda|}\eta,
\end{align}
whereas $|\Brx-\Tlx|\le c_1\eta\theta^2=o(\eta)$ as discussed. 
Thus, $\Brx$ is eventually too close to $\Tlx$ to cross $x^*$ along the joint limit, so $\Brx-x^*$ has the same sign as
$\Tlx-x^*$. 
\end{proof}

\begin{proof}[Proof of \cref{cor:separable.2}] 
A general formula is available for a separable game of the form $F_i(x)=f_i(x_i)$ for every $i\in S$. 
$F'(x) = F'(x)^\top = \diag{(f'_i(x_i))_{i = 1}^n}$ for all $x \in X$. 
Also let $a_i(x)\Is d_i(x)e_i$ so that $F_i'(x)=a_i(x)$. 
We have 
$\OL{F'}(x)
=\sum_{l\in S}P_l(x)a_l(x)
=F'(x)P(x)$. 
Consequently, 
    $\WH{F_i'}(x)
    =a_i(x)-F'(x)P(x)
    =F'(x)(e_i-P(x))$. 
From this and \cref{eq:v-vector-form}, we show 
    $v_i(x)
    =\frac{1}{2k\eta^2}
    (e_i-P(x))^\top K(x)(e_i-P(x))
    $ where $K(x) \Is F'(x) \Sigma(x) F'(x)$. 

For the two-action case, let $a_i=F_i'(x_i)e_i$. 
We have $\WH{F_1'}(x)=P_2(x)(a_1-a_2)$ and 
$\WH{F_2'}(x)=-P_1(x)(a_1-a_2)$. 
We note 
\begin{align}
    (a_1-a_2)^\top\Sigma(x)(a_1-a_2)
    =\left(F_1'(x_1)+F_2'(x_2)\right)^2x_1x_2
    =\sigma_F(x).
\end{align}
It follows that
$v_1=(2k\eta^2)^{-1}P_2^2\sigma_F$ and
$v_2=(2k\eta^2)^{-1}P_1^2\sigma_F$. Hence,
\begin{align}
    \WH v_1
    &=P_2(v_1-v_2)
    =\frac{1}{2k\eta^2}P_2(P_2-P_1)\sigma
    =\frac{1}{2k\eta^2}P_2(1-2P_1)\sigma.
\end{align}
Also, \Cref{eq:delta-with-curvature} yields \cref{eq:2x2-curvature-premium} as we note $x_1(1-x_1)=x_2(1-x_2)=x_1x_2$. 
\end{proof}

\begin{proof}[Proof of \cref{lem:tau-lim.local-limit}]
Set $\tau_k\Is\sqrt{k}\eta_k$. Fix a compact set $K\subset T$.
As $x^*\in\Int{X}$, $x^k(h)=x^*+h/\sqrt{k}$ lies in $\Int{X}$ for every $h\in K$ when $k$ is sufficiently large.
For such $k$ and $h$, the sample $z$ follows $\Mult(k\!\mid\!x^k(h))$. 
Define the normalized random vector $U_{k,h} \in T$ by 
\begin{align}
    U_{k,h}
    \Is\frac{z-kx^k(h)}{\sqrt{k}}
    \quad\text{where}\quad 
    z \sim \Mult(k\!\mid\!x^k(h)). 
\end{align}
Then, with $w = z/k$, by the definitions of $\ClL^k$ and $L^{k,\eta_k}$, 
\begin{align}
    \ClL^k(h)
    &=
    \BbE\left[
        P^{\eta_k}\left(w\right)
    \right] 
    =
    \BbE\left[
        \Lambda^{\eta_k}\left(
            Aw
        \right)
    \right] 
    =
    \BbE\left[
        \Lambda^{\eta_k}
        \bigl(
            A\left(w-x^*\right)
        \bigr)
    \right]
    \\
    &=
    \BbE\left[
        \Lambda^{\tau_k}\bigl(
            \sqrt{k}\,A\left(w-x^*\right)
        \bigr)
    \right]
    =
    \BbE\left[
        \Lambda^{\tau_k}\bigl(A(h+U_{k,h})\bigr)
    \right].
    \label{eq:proof-critical-rescaled-rule}
\end{align}
The third equality uses $Ax^*=\Olu \Vt1$ and invariance of logit choice to a common payoff addition, the fourth uses $\tau_k=\sqrt{k}\eta_k$, and the last uses $\sqrt{k}(w - x^*)=h+U_{k,h}$.

For every sequence $h_k\to h$ in $K$, set $p_k=x^k(h_k)$ and write
$z=\sum_{r=1}^k\xi_{k,r}$, where the $\xi_{k,r}$ are independent
categorical vectors satisfying $\BbP(\xi_{k,r}=e_i)=p_{k,i}$. Then
\begin{align}
    U_{k,h_k}
    =\frac{1}{\sqrt{k}}
    \sum_{r=1}^k\bigl(\xi_{k,r}-p_k\bigr).
\end{align}
Because $p_k\to x^*$, the covariance of this sum satisfies
$\Sigma(p_k)\to\Sigma^*$. Moreover,
\begin{align}
    \left\|
        \frac{\xi_{k,r}-p_k}{\sqrt{k}}
    \right\|
    \le 
        \frac{\|\xi_{k,r}\| + \|p_k\|}{\sqrt{k}}
    \le\frac{2}{\sqrt{k}}. 
\end{align}
Thus, for every $\epsilon>0$, the Lindeberg indicators vanish for all sufficiently large $k$. 
The multivariate Lindeberg--Feller theorem \citep[e.g.,][Prop.~2.27]{vanderVarrt-Book2000} thus yields 
\begin{align}
    U_{k,h_k}\xrightarrow{d}Z^*,
    \label{eq:proof-critical-clt}
\end{align}
where $\xrightarrow{d}$ denotes convergence in distribution.
As $\tau_k\to\tau>0$, and $(t,y)\mapsto\Lambda^t(Ay) \in (0,1)^n$ is bounded and continuous for $t>0$, \Cref{eq:proof-critical-rescaled-rule,eq:proof-critical-clt} yield $\ClL^k(h_k)\to Q^\tau(h)$.
Continuity of $Q^\tau$ and sequential compactness of $K$ then imply $\ClL^k\to Q^\tau$ uniformly on $K$. 

It remains to show convergence of the derivatives. 
For any $u\in T$, differentiating the multinomial probabilities in the definition of $\ClL^k(h) = L^{k,\eta}(x^{k}(h))$ gives the score identity
\begin{align}
    D\ClL^k(h)u
    =
    \BbE\left[
        \Lambda^{\tau_k}\bigl(A(h+U_{k,h})\bigr)
        u^\top\diag{x^k(h)}^{-1}U_{k,h}
    \right].
    \label{eq:proof-critical-score}
\end{align}
Indeed, for a fixed sample $z$, the sum $h+U_{k,h}=z/\sqrt{k}-\sqrt{k}x^*$ is independent of $h$, so only the multinomial probability is differentiated in this calculation.
Since the multinomial coefficient is independent of $h$ and
$D_hx^k(h)u=u/\sqrt{k}$, we have
\begin{align}
    D_h\log M^k(z\!\mid\!x^k(h))u
    &=
    \frac{1}{\sqrt{k}}
    \sum_{i\in S}\frac{z_i u_i}{x_i^k(h)} 
    =
    \sum_{i\in S}
    \frac{(z_i-k x_i^k(h))u_i}
    {\sqrt{k}\, x_i^k(h)}
    +\sqrt{k}\sum_{i\in S}u_i
    \\
    &=
    \sum_{i\in S}
    \frac{(z_i-k x_i^k(h))u_i}
    {\sqrt{k}\, x_i^k(h)}=
    \sum_{i\in S}
    \frac{[U_{k,h}]_i u_i}
    { x_i^k(h)},
\end{align}
where the last equality uses $u\in T$, and hence $\sum_i u_i=0$.
The final expression equals
$u^\top\diag{x^k(h)}^{-1}U_{k,h}$, which gives the score factor in
\cref{eq:proof-critical-score}.
Uniformly over $h\in K$ and unit vectors $u\in T$, the score variables in \cref{eq:proof-critical-score} are bounded in $L^2$, and hence uniformly integrable, because $x^k(h)$ stays bounded away from the boundary of $X$ and $\BbE[\|U_{k,h}\|^2] = \trace[\Sigma(x^k(h))] \le1$. 
Consequently, for arbitrary sequences $h_k\to h$ in $K$ and $u_k\to u$ in $T$, the central limit theorem and uniform integrability imply that \cref{eq:proof-critical-score} converges to
\begin{align}
    \BbE\left[
        \Lambda^\tau\bigl(A(h+Z^*)\bigr)
        u^\top\diag{x^*}^{-1}Z^*
    \right].
    \label{eq:proof-critical-score-limit}
\end{align}
Gaussian integration by parts converts \cref{eq:proof-critical-score-limit} to $DQ^\tau(h)u$.
Let $b\Is\diag{x^*}^{-1}u$.
For every continuously differentiable scalar function $g$ with bounded derivative, Stein's Lemma gives 
$\BbE\bigl[g(Z^*)b^\top Z^*\bigr]
=
\BbE\bigl[Dg(Z^*)\Sigma^*b\bigr]$. 
In the present case,
\begin{align}
    \Sigma^*b
    =
    \Sigma^*\diag{x^*}^{-1}u
    = 
    (\diag{x^*} - x^* x^{*\top}) \diag{x^*}^{-1} u 
    =
    u-x^*(\Vt1^\top u)
    =
    u,
\end{align}
since $\Vt1^\top u = 0 $ for $u \in T$.
For each $i\in S$, apply Stein's Lemma to
$g_i(z)\Is\Lambda_i^\tau(A(h+z))$, whose derivative is
$Dg_i(z)=D\Lambda_i^\tau(A(h+z))A$, to obtain
\begin{align}
    \BbE\left[
        \Lambda_i^\tau\bigl(A(h+Z^*)\bigr)b^\top Z^*
    \right]
    & =
    \BbE\left[
        D\Lambda_i^\tau\bigl(A(h+Z^*)\bigr)A\Sigma^*b
    \right] 
    \\
    & =
    \BbE\left[
        D\Lambda_i^\tau\bigl(A(h+Z^*)\bigr)Au
    \right].
\end{align}
Stacking these componentwise identities gives
\begin{align}
    \eqref{eq:proof-critical-score-limit}
    =
    \BbE\left[
        D\Lambda^\tau\bigl(A(h+Z^*)\bigr)Au
    \right]
    =
    DQ^\tau(h)u.
\end{align}
Continuity of $DQ^\tau$, together with sequential compactness of $K$ and of the unit sphere in $T$, now gives uniform convergence of the derivatives on $K$.
Since $K$ was arbitrary, the lemma follows.
\end{proof}

\begin{proof}[Proof of \cref{thm:tau-lim.joint-limit}]
For the proof, let $\SfC \Is I-\frac{1}{n}\Vt1\Vt1^\top$ and define $J \Is \SfC A$. 
Here, $\SfC$ is a projection matrix onto $T$, and $J$ defines a linear map from $T$ to $T$. 
We note $\SfC F(x^* + h) = \SfC A (x^* + h) = \Olu \SfC \Vt1 + \SfC Ah = J h$.  
Condition \eqref{eq:tau-lim.regularity} ensures that $h \mapsto J h$ is one-to-one and hence invertible on $T$. 
We can thus interchangeably use the displacement of population state $h\in T$ and the (centered) payoff variation $u \in T$ induced by $h$ because $u = \SfC (F(x + h) - F(x)) = \SfC A h = Jh$. 
Let $V\Is JZ^*$ be the payoff variation induced by $Z^*$.
Common-payoff invariance of logit choice gives
\begin{align}
    Q^\tau(h)
    =
    \BbE\left[\Lambda^\tau(Jh+V)\right]
    =
    \BbE\left[\Lambda^\tau(u+V)\right].
\end{align}
We also define the expected surplus in terms of $u \in T$ as follows:
\begin{align}
    W^\tau(u)
    \Is
    \BbE\left[
        \tau\log\sum_{i\in S}
        \exp\left(\frac{u_i+V_i}{\tau}\right)
    \right].  
    \label{eq:proof-critical-surplus}
\end{align} 

\begin{remark}
The vector $\varepsilon$ in \cref{eq:tau-lim.random-utility} is independent of $V$.
Up to the additive constant $\tau\gamma$, where $\gamma$ is Euler's constant, $W^\tau$ is the surplus function of the additive random utility model with disturbance $V+\tau\varepsilon$:
\begin{align}
    W^\tau(u)+\tau\gamma
    =
    \BbE\left[\max_{i\in S}\{u_i+V_i+\tau\varepsilon_i\}\right].
\end{align}
The Williams--Daly--Zachary theorem therefore identifies its gradient with the choice-probability map $u\mapsto\BbE[\Lambda^\tau(u+V)]$.
Moreover, because the composite disturbance has a strictly positive smooth density, Theorem 2.1 of \citet{Hofbauer-Sandholm-ECTA2002} implies that this map restricts to a diffeomorphism from $T$ onto $\Int{X}$ and has a positive-definite derivative on $T$.
These conclusions concern the centered payoff coordinate $u$.
In the state coordinate $h$, the chain rule gives $\nabla_h W^\tau(Jh)=J^\top\BbE[\Lambda^\tau(Jh+V)]$, so $Q^\tau(h)$ is not itself a gradient of $W^\tau$ with respect to $h$. 
Below, we rely on the invertibility of $J$ rather than convexity in $h$.
The argument below verifies the required properties directly for completeness and derives intermediate formulas to be used later.
\end{remark}

\UlHeading{Part \cref{enum:tau-lim.root}}
We show existence and uniqueness of $h^\tau$. 
Differentiation gives
\begin{align}
    \nabla W^\tau(u)
    =\BbE[\ell]
    \quad \text{and} \quad 
    D^2W^\tau(u)
    =
    \frac{1}{\tau}
    \BbE\bigl[
        \diag{\ell}
        -\ell\ell^\top
    \bigr]
    =
    \frac{1}{\tau}
    \BbE\bigl[
        \Sigma(\ell)
    \bigr]
    \label{eq:proof-critical-surplus-derivatives}
\end{align}
with $\ell \Is \Lambda^\tau(u + V)$. 
The Hessian is positive definite on $T$. 
Indeed, for every $u\in T\setminus\{\Vt0\}$ and every positive probability vector $\ell$, 
\begin{align}
    u^\top\bigl(\diag{\ell}-\ell\ell^\top\bigr)u =\sum_{i\in S}\ell_i \bigl(u_i-\sum_{j\in S}\ell_j u_j\bigr)^2 >0.
    \label{eq:proof-sigma-PD}
\end{align}
For $u \in T$, define 
\begin{align}
    H^\tau(u)\Is W^\tau(u)-\la x^*,u\ra. 
\end{align}
Since the log-sum-exp function dominates the maximum and $\BbE[V_i]=0$ for every $i\in S$,
\begin{align}
    W^\tau(u)
    \ge
    \BbE\left[\max_{i\in S}(u_i+V_i)\right]
    \ge
    \max_{i\in S}u_i.
\end{align}
It follows that
\begin{align}
    H^\tau(u)
    \ge
    \max_{i\in S}u_i-\la x^*,u\ra.
    \label{eq:proof-critical-coercive-bound}
\end{align}
The right-hand side is coercive on $T$. 
To see this, its minimum over $\{u\in T:\|u\|=1\}$ is strictly positive. 
In fact, equality would require every $i \in S$, which has positive weight $x_i^*$, to attain the maximum, making $u$ a constant vector, which is impossible for a unit vector in $T$. 
Thus, $H^\tau$ is coercive and strictly convex, and it has a unique minimizer $u^\tau\in T$ for each $\tau\in (0,\infty)$. 
The first-order condition on $T$ is
\begin{align}
    \nabla W^\tau(u^\tau) = \BbE[\Lambda^\tau(u^\tau+V)] = x^*.
    \label{eq:proof-critical-surplus-foc}
\end{align}
The vectors on both sides of \cref{eq:proof-critical-surplus-foc} sum to one, so the first-order condition on $T$ indeed implies the vector equality.
By condition \eqref{eq:tau-lim.regularity}, $J$ has a well-defined inverse on $T$, so $h^\tau\Is J^{-1}u^\tau$ is the unique solution of $x^*=Q^\tau(h^\tau)$. 
This proves part \ref{enum:tau-lim.root}.

\UlHeading{Part \cref{enum:tau-lim.equilibrium}}
By \cref{lem:tau-lim.local-limit}, $\ClL^k\to Q^\tau$ in $C^1_{\mathrm{loc}}(T)$.
Define
\begin{align}
    R_k(h)
    \Is
    \ClL^k(h)-x^*-\frac{h}{\sqrt{k}},
    \quad\text{and}\quad 
    R(h)
    \Is
    Q^\tau(h)-x^*.
\end{align}
Thus, $R_k\to R$ in $C^1_{\mathrm{loc}}(T)$.
Moreover,
\begin{align}
    DQ^\tau(h)
    =
    \frac{1}{\tau}
    \BbE\left[
        \diag{\ell(h)}-\ell(h)\ell(h)^\top
    \right]J,
    \label{eq:proof-critical-choice-derivative}
\end{align}
where $\ell(h)\Is\Lambda^\tau(J(h+Z^*))$. 
The expected matrix preceding $J$ is positive definite on $T$ by \cref{eq:proof-sigma-PD}, and $J$ is invertible when restricted to $T$ by condition \eqref{eq:tau-lim.regularity}. 
Thus, $DR(h^\tau)$ is nonsingular.To show the existence of a unique zero of $R_k$ near $h^\tau$, we use the persistence of a nondegenerate zero under local perturbations.

For completeness, let $A_0\Is DR(h^\tau)$. 
Choose a closed ball $\overline B_r(h^\tau)\subset T$ small enough that $I-A_0^{-1}DR(h)$ has norm below $1/4$ throughout the ball. 
By $C^1$ convergence of $\ClL^k$ to $Q^\tau$, $R_k(h^\tau)\to\Vt0$ and, for all sufficiently large $k$, the map
\begin{align}
    \Phi_k(h)\Is h-A_0^{-1}R_k(h)
\end{align}
maps $\overline B_r(h^\tau)$ into itself and is a contraction there. 
It therefore has a unique fixed point $h_k$ in this ball. 
Equivalently, $R_k(h_k)=\Vt0$, and this is the only zero of $R_k$ in $\overline B_r(h^\tau)$.
The uniform convergence of $R_k$ and uniqueness of the zero of $R$ in the ball imply $h_k\to h^\tau$.
Thus, after increasing $k^*$ if necessary, $x^{k,\eta_k}=x^*+h_k/\sqrt{k}$ is the unique SLE in the neighborhood $x^*+\frac{1}{\sqrt{k}} \{h\in T:\|h-h^\tau\|<r\}$.
The global invertibility of $Q^\tau$ does not imply global uniqueness of finite-$k$ SLE, because the $C^1_{\mathrm{loc}}$ convergence controls $R_k$ only on fixed compact subsets of $T$.
This proves the limit and part \ref{enum:tau-lim.equilibrium}.
\end{proof}

\begin{proof}[Proof of \cref{prop:tau-lim.stability}]
As in the preceding proof, let $\SfC=I-\frac{1}{n}\Vt1\Vt1^\top$ and $J=\SfC A$ on $T$, and write the local SLE as $x^{k,\eta_k}=x^k(h_k)$ with $h_k\Is\sqrt{k}\bigl(x^{k,\eta_k}-x^*\bigr)\to h^\tau$.
Since $\sum_{i\in S}\bigl(L_i^{k,\eta_k}(x)-x_i\bigr)=0$, the dynamic \eqref{eq:sLD} leaves the affine hull of $X$ invariant, and local stability of the interior rest point $x^{k,\eta_k}$ is governed by the eigenvalues of the Jacobian of $x\mapsto L^{k,\eta_k}(x)-x$ restricted to $T$.

From $\ClL^k(h)=L^{k,\eta_k}(x^*+h/\sqrt{k})$, the chain rule gives $DL^{k,\eta_k}(x^k(h_k))u=\sqrt{k}\,D\ClL^k(h_k)u$ for $u\in T$.
The $C^1_{\mathrm{loc}}$ convergence in \cref{lem:tau-lim.local-limit}, together with $h_k\to h^\tau$ and continuity of $DQ^\tau$, yields $D\ClL^k(h_k)\to DQ^\tau(h^\tau)$.
By \cref{eq:proof-critical-choice-derivative},
\begin{align}
    DQ^\tau(h^\tau)=MJ,
    \qquad
    M\Is\frac{1}{\tau}\BbE\bigl[\diag{\ell}-\ell\ell^\top\bigr]
    \quad \text{with} \quad
    \ell = \Lambda^\tau\bigl(J(h^\tau+Z^*)\bigr),
\end{align}
and $M$ is symmetric and positive definite on $T$ by \cref{eq:proof-sigma-PD}.
The random utility property used below is that $M$ is symmetric and positive definite on $T$.
The relevant spectrum, however, is that of $MJ$, and the game-dependent product $MJ$ need be neither symmetric nor positive definite.
The sign of the real parts must therefore be established from the properties of the game $A$.
Define the approximation error for the derivative by
\begin{align}
    E_k\Is D\ClL^k(h_k)-MJ.
\end{align}
The convergence $D\ClL^k(h_k)\to MJ$ implies $\|E_k\|\to0$.
Thus, the Jacobian of the dynamic on $T$ may be written as 
\begin{align}
    B_k
    &=DL^{k,\eta_k}(x^k(h_k))-I
    =\sqrt{k}\bigl(MJ+E_k\bigr)-I.
\end{align}

We locate the spectrum of $MJ$ on $T$.
Let $\mu\in\BbC$ be an eigenvalue with eigenvector $v=a+\Rmi b\ne\Vt0$, where $a,b\in T$.
Positive definiteness makes $M$ invertible on $T$, so $MJv=\mu v$ implies $Jv=\mu M^{-1}v$, and multiplying by the conjugate transpose $v^*$ gives
\begin{align}
    v^*Jv=\mu\,v^*M^{-1}v
    \qquad \text{with} \quad v^*M^{-1}v>0.
\end{align}
For real $h\in T$, we have $\SfC h=h$ and hence $h^\top Jh=h^\top A h=\la h,Ah\ra$, so that
\begin{align}
    \operatorname{Re}(v^*Jv)
    =a^\top Ja+b^\top Jb
    =\la a,Aa\ra+\la b,Ab\ra.
\end{align}
Under the condition in part \ref{enum:tau-lim.stability-stable}, this quantity is negative, so every eigenvalue of $MJ$ on $T$ has a negative real part.
Under the condition in part \ref{enum:tau-lim.stability-unstable}, we instead use the trace.
On $T$, $MJ$ is similar to $M^{1/2}JM^{1/2}$, and the trace of $M^{1/2}KM^{1/2}$ vanishes for the skew-symmetric part $\frac{1}{2}(J - J^\top)$ of $J$.
Hence,
\begin{align}
    \trace[MJ]
    =\trace\Bigl[M^{1/2}\,\tfrac{J+J^\top}{2}\,M^{1/2}\Bigr]>0
\end{align}
because the symmetric part of $J$ is positive definite on $T$ by condition \ref{enum:tau-lim.stability-unstable}.
Since the eigenvalues of the real matrix $MJ$ occur in conjugate pairs, their real parts sum to $\trace[MJ]>0$, so at least one eigenvalue has a positive real part.

Finally, eigenvalues depend continuously on matrix entries, so there is $\delta>0$ such that, for all sufficiently large $k$, the eigenvalues $\mu_k$ of $MJ+E_k$ on $T$ satisfy $\max\operatorname{Re}\mu_k\le-\delta$ under condition \ref{enum:tau-lim.stability-stable} and $\max\operatorname{Re}\mu_k\ge\delta$ under condition \ref{enum:tau-lim.stability-unstable}.
The eigenvalues of $B_k$ are $\sqrt{k}\mu_k-1$.
In the first case, every real part is at most $-\sqrt{k}\delta-1<0$, and $x^{k,\eta_k}$ is locally asymptotically stable by the principle of linearized stability.
In the second case, some real part is at least $\sqrt{k}\delta-1>0$ for all large $k$, and the rest point is unstable.
\end{proof}

\begin{proof}[Proof of \cref{cor:tau-lim.endpoints}]
As in the preceding proof, let $\SfC=I-(1/n)\Vt1\Vt1^\top$, define $Jh=\SfC Ah$ for $h\in T$, and write $V=JZ^*$ and $u^\tau=Jh^\tau$. 

\UlHeading{Part \cref{enum:tau-lim.small-tau}}
We first identify the unique solution of the zero-noise limiting problem,
then show that $h^\tau$ converges to this solution, and finally establish
the pointwise convergence of $Q^\tau$ to $Q^0$.

Define, for $u\in T$,
\begin{align}
    W^0(u) 
    \Is
    \BbE\Bigl[\max_{i\in S}\{u_i+V_i\}\Bigr]
    \quad\text{and}\quad
    H^0(u)
    \Is
    W^0(u)-\la x^*,u\ra.
\end{align}
The function $W^0$ is the zero-noise counterpart of the smoothed surplus
$W^\tau$. 
As $\tau\Toz$, the independent Gumbel shocks disappear and the log-sum-exp becomes a maximum. Because the remaining Gaussian disturbance $V$ is supported on $T$, we verify the properties of $H^0$ directly.

The covariance matrix $J\Sigma^*J^\top$ is positive definite on $T$.
Hence, $V$ has full support on $T$, and pairwise payoff ties occur with
probability zero. The function $W^0$ is convex because it is the expectation
of a maximum of affine functions. It is in fact strictly convex on $T$.
To see this, take distinct $u,u'\in T$. Since $u'-u$ is not constant, there
are actions $i,j$ for which
$(u_i'-u_i)>(u_j'-u_j)$. There is then a nonempty open set of realizations
of $V$ on which $j$ uniquely maximizes $u+V$, whereas $i$ uniquely maximizes
$u'+V$. This set has positive Gaussian probability, making the convexity
inequality for the expected maximum strict. The coercivity argument in
\cref{eq:proof-critical-coercive-bound} also applies to $H^0$. Therefore,
$H^0$ has a unique minimizer $u^0\in T$.

Since ties occur with probability zero, differentiation under the expectation gives
\begin{align}
    \frac{\partial W^0(u)}{\partial u_i}
    =
    \BbP\left(u_i+V_i>\max_{j\ne i}\{u_j+V_j\}\right).
\end{align}
The first-order condition for minimizing $H^0$ on $T$ requires $\nabla W^0(u^0)-x^*$ to be orthogonal to $T$. This difference must therefore be a multiple of $\Vt1$. 
Since both $\nabla W^0(u^0)$ and $x^*$ sum to one, the multiple is zero, and hence $\nabla W^0(u^0)=x^*$. 
Setting $h^0\Is J^{-1}u^0$, \Cref{eq:tau-lim.gaussian-br} shows that this condition reduces to $Q^0(h^0)=x^*$. The uniqueness of $u^0$ and the invertibility of $J$ on $T$ give the uniqueness of $h^0$.

We next show that $h^\tau\to h^0$. For every $y\in\BbR^n$, the standard
log-sum-exp bound is
\begin{align}
    0
    \le
    \tau\log\sum_{i\in S}\exp(y_i/\tau)
    -\max_{i\in S}y_i
    \le
    \tau\log n.
    \label{eq:proof-critical-logsum-bound}
\end{align}
Indeed, after factoring out $\max_i y_i$, the remaining sum inside the
logarithm lies between $1$ and $n$. Applying the bound to $y=u+V$ and taking
expectations gives
\begin{align}
    0\le H^\tau(u)-H^0(u)\le\tau\log n
    \qquad \forall u\in T.
    \label{eq:proof-critical-logsum-bound-2}
\end{align}
Thus, $H^\tau\to H^0$ uniformly on all of $T$ as $\tau \Toz$.

Fix an arbitrary sequence $\{\tau_m\}$ in $(0,1]$ such that $\tau_m\Toz$.
To keep the corresponding minimizers in a common compact set, define the
coercive lower bound from \cref{eq:proof-critical-coercive-bound} by
\begin{align}
    g(u)\Is\max_{i\in S}u_i-\la x^*,u\ra.
\end{align}
For each $m$, the minimizing property of $u^{\tau_m}$ and
\cref{eq:proof-critical-logsum-bound-2} imply
\begin{align}
    g(u^{\tau_m})
    \le H^{\tau_m}(u^{\tau_m})
    \le H^{\tau_m}(\Vt0)
    \le H^0(\Vt0)+\log n.
\end{align}
Because $g$ is coercive on $T$, the sequence $\{u^{\tau_m}\}$ lies in a
compact subset of $T$.

Let $\{u^{\tau_{m_\ell}}\}$ be any convergent subsequence and suppose that it converges to $\bar u$. 
For every $u\in T$, optimality gives
\begin{align}
    H^{\tau_{m_\ell}}(u^{\tau_{m_\ell}})\le H^{\tau_{m_\ell}}(u).
\end{align}
Passing to the limit along the convergent subsequence, using uniform convergence of $H^{\tau_{m_\ell}}$ to $H^0$, yields
$H^0(\bar u)\le H^0(u)$ for every $u\in T$. Thus, $\bar u=u^0$ by the
uniqueness of the minimizer of $H^0$. 
This shows that every convergent subsequence has the
same limit, so compactness implies $u^{\tau_m}\to u^0$. Since the sequence $\{\tau_m\}$ was arbitrary, $u^\tau\to u^0$. Finally, continuity of $J^{-1}$ on $T$ gives $h^\tau=J^{-1}u^\tau\to h^0=J^{-1}u^0$.

It remains to show the convergence of the response maps. For each fixed
$h\in T$, pairwise ties among the components of
$J(h+Z^*)$ occur with probability zero. For every realization with a
unique maximizer, $\Lambda^\tau(J(h+Z^*))$ converges to the corresponding
unit vector as $\tau\Toz$. Since every component of the logit vector lies
in $[0,1]$, dominated convergence gives $Q^\tau(h)\to Q^0(h)$.
This proves part \ref{enum:tau-lim.small-tau}.

\UlHeading{Part \cref{enum:tau-lim.large-tau}}
Let $s\Is\tau^{-1}$, let $\Lambda\Is\Lambda^1$ denote the unit-noise
logit map, and set $\Omega\Is J\Sigma^*J^\top$. The scaling identity
$\Lambda^\tau(y)=\Lambda(y/\tau)$ converts the first-order condition
\cref{eq:proof-critical-surplus-foc} into
\begin{align}
    \BbE[\Lambda(b^\tau+sV)]
    =x^*,
    \qquad
    b^\tau\Is\frac{u^\tau}{\tau}\in T.
    \label{eq:proof-critical-large-tau-equation}
\end{align}
This rescaling separates the large deterministic displacement, contained
in $b^\tau$, from the Gaussian sampling term $sV$, which becomes small as
$\tau\to\infty$.

We first identify the solution when the sampling term is removed. The
identity $\Lambda(b)=x^*$ holds exactly when
$b=\log x^*+c\Vt1$ for some scalar $c$. Restricting to $T$ therefore
selects the unique solution
\begin{align}
    b^*\Is\SfC\log x^*.
\end{align}
Moreover,
    $D\Lambda(b^*)
    =\diag{x^*}-x^*x^{*\top}
    =\Sigma^*$, 
which is nonsingular when restricted to $T$. Applying the implicit
function theorem to
\begin{align}
    \Psi(b,s)\Is\BbE[\Lambda(b+sV)]-x^*
\end{align}
gives a unique smooth solution $b(s)$ near $(b^*,0)$. For $s>0$, this local solution equals $b^\tau$ because the critical equation was shown above to have a unique solution. Since $V$ and $-V$ have the same distribution, $\Psi(b,s)=\Psi(b,-s)$. 
Local uniqueness then implies $b(s)=b(-s)$. 
Thus, the expansion of $b(s)$ contains only even powers of $s$, and in particular $b(s)-b^*=O(s^2)$.

Expanding $\Psi(b(s),s)=\Vt0$ around $(b^*,0)$ now gives
\begin{align}
    \Vt0
    =
    \Sigma^*(b(s)-b^*)
    +
    \frac{s^2}{2}v^*
    +O(s^4) 
    \quad \text{where} \quad 
    v_i^*\Is\la D^2\Lambda_i(b^*),\Omega\ra.
    \label{eq:proof-critical-large-tau-taylor}
\end{align}
Here, the symmetry of the Gaussian distribution eliminates the odd powers of $s$. 
The bounded derivatives of the logit map and the finite Gaussian moments control the remainder. 
The Hessian term $v^*$ comes from the same calculation that
produces the variance premium for the approximated rule in
\cref{eq:v-vector-form}. 
In particular, specializing the general logit Hessian formula \cref{eq:P''} to the identity payoff map, unit noise, and $\Lambda(b^*)=x^*$ gives
\begin{align}
    D^2\Lambda_i(b^*)
    =x_i^*
    \Bigl(
        (e_i-x^*)(e_i-x^*)^\top
        -\sum_{l\in S}x_l^*(e_l-x^*)(e_l-x^*)^\top
    \Bigr).
\end{align}
Taking the inner product with $\Omega$ therefore yields
\begin{align}
    v_i^*
    =x_i^*
    \bigl(\,
        \rho_i^*-
        \sum_{l\in S}x_l^*\rho_l^*
    \bigr),
    \quad\text{with}\quad
    \rho_i^*\Is(e_i-x^*)^\top\Omega(e_i-x^*).
    \label{eq:proof-critical-hessian-contraction}
\end{align}
To identify $\rho_i^*$, note that $e_i-x^*\in T$ and $J=\SfC A$, so
\begin{align}
    J^\top(e_i-x^*)
    =A^\top(e_i-x^*)
    =a_i-\overline a^*.
\end{align}
Consequently, by the definition in
\cref{eq:tau-lim.sigma-star}, 
\begin{align}
    \rho_i^*
    &=(a_i-\overline a^*)^\top
      \Sigma^*(a_i-\overline a^*)
    =\sigma_i^*.  
\end{align}
Hence, in vector form, $v^*=\Sigma^*\sigma^*$.

Substituting this identity into
\cref{eq:proof-critical-large-tau-taylor} and using
$\Sigma^*\SfC\sigma^*=\Sigma^*\sigma^*$ gives
\begin{align}
    \Sigma^*
    \left(
        b(s)-b^*+\frac{s^2}{2}\SfC\sigma^*
    \right)
    =O(s^4).
\end{align}
The vector in parentheses lies in $T$, where $\Sigma^*$ is invertible.
Therefore, plugging $s=\tau^{-1}$,
\begin{align}
    b^\tau
    =\SfC\log x^*
    -\frac{1}{2\tau^2}\SfC\sigma^*
    +O(\tau^{-4}).
\end{align}
Multiplying by $\tau$ and recalling $\tau b^\tau = u^\tau=Jh^\tau$ gives the centered
payoff expansion
\begin{align}
    Jh^\tau
    =\tau\SfC\log x^*
    -\frac{1}{2\tau}\SfC\sigma^*
    +O(\tau^{-3}).
    \label{eq:proof-critical-large-tau-centered-payoff}
\end{align}

It remains to translate this centered identity into the two formulations
in the main text. By $J=\SfC A$,
\Cref{eq:proof-critical-large-tau-centered-payoff} is equivalent to
\begin{align}
    \SfC\left(
        Ah^\tau+\frac{1}{2\tau}\sigma^*
        -\tau\log x^*
    \right)
    =O(\tau^{-3}).
\end{align}
The vector in parentheses is the local virtual payoff
$\ClF^\tau(h^\tau)$ from \cref{eq:tau-lim.local-virtual-payoff}. Its
centered part is of order $O(\tau^{-3})$, so there is a scalar $c^\tau$
such that every component equals $c^\tau+O(\tau^{-3})$. This is
\cref{eq:tau-lim.virtual-indifference}, the local counterpart of the
virtual-payoff representation in
\cref{obs:TlL-interpretation}~\ref{enum:virtual-game-2}.
Finally, subtracting components $i$ and $j$ in
\cref{eq:proof-critical-large-tau-centered-payoff} gives
\cref{eq:tau-lim.large-tau-expansion}. This proves part
\ref{enum:tau-lim.large-tau}.
\end{proof}

\begin{proof}[Proof of \cref{prop:k=1}] 
We have 
$L^{1,\eta}(x) 
= 
\sum_{j\in S} M^1(e_j \mid x) \cdot P ^\eta (e_j)
= 
\sum_{j \in S} \Pi_{ij} x_j$ 
where $\Pi_{ij} \Is P_i^\eta(e_j) \in (0,1)$ is the $\eta$-logit choice probability of action $i\in S$ at $e_j$. 
The equilibrium condition $x = L^{1,\eta}(x)$ then reads $x_i = \sum_{j \in S} \Pi_{ij} x_j$ or $x = \Pi x$. 
Because $\Pi \in \BbRg^{n\times n}$ can be seen as a transition probability matrix of an irreducible Markov chain, there exists a unique $x^* \in X$ satisfying $x^* = \Pi x^*$, and $x^*$ is proportional to the positive eigenvector of $\Pi$ associated with the Perron--Frobenius eigenvalue $1$. 
Furthermore, \eqref{eq:sLD} reduces to a linear dynamical system $\Dtx = \Pi x - x = (\Pi - I) x$. 
It follows via standard arguments that $X$ is forward invariant under the dynamic and $x^*$ is globally asymptotically stable. 
\end{proof}

\begin{proof}[Proof of \cref{prop:k=2}]
Assume $k = 2$ and $n = 2$. 
Write $y = x_1$ and let $z \in\{0,1,2\}$ denote how many times an action-$1$ player is drawn in a sample of size $k = 2$. 
Then, $\BbP(z = 2) = x_1^2 = y^2$, 
$\BbP(z = 1) = 2 x_1 x_2 = 2y (1 - y)$, and  
$\BbP(z = 0) = x_2^2 = (1 - y)^2$. 
Let $\rho_z \in (0,1)$ be the choice probability of action $1$ for each realization of $z$. 
That is, 
$\rho_2 = P_1^\eta(e_1)$, 
$\rho_1 = P_1^\eta(\frac{1}{2}(e_1 + e_2))$, 
and 
$\rho_0 = P_1^\eta(e_2)$. 
Define 
\begin{align}
    f(y) 
    & 
    \Is 
    L_1^{k,\eta}(y, 1 - y) - y \\ 
    & = \rho_2 \times y^2 + 2 \rho_1 \times y(1 - y) + \rho_0 \times (1 - y)^2 
    - y
    \\
    & = (\rho_2 - 2 \rho_1 + \rho_0) y^2 + \left( 2(\rho_1 - \rho_0) - 1 \right) y + \rho_0. 
    \label{eq:k=2.f(x)}
\end{align}
Then, an SLE is a solution to $f(y) = 0$ satisfying $y \in (0,1)$. 
Also, $f(0) = \rho_0 > 0$ and $f(1) = \rho_2 - 1 < 0$. 
Because $f$ is a quadratic function, there is a unique $y^* \in (0,1)$ that solves $f(y^*) = 0$. 
Furthermore, $y^*$ is globally asymptotically stable because $\Dty = f(y) > 0$ for $y \in [0,y^*)$ and $\Dty = f(y) < 0$ for $y \in (y^*,1]$. 

\begin{remark}
While the sampling logit choice rule is a quadratic function for general $n \ge 2$ under $k = 2$, uniqueness does not follow in general for $n \ge 3$. 
See, e.g., \cite{Saburov-Yusof-2018} for an explicit construction of a quadratic positive stochastic operator with three fixed points on the simplex when $n = 3$. 
As a concrete example, the linear population game $A = \diag{(1,6,6)}$ can give rise to three fixed points. 
\end{remark}
\end{proof}

\begin{proof}[Proof of \cref{prop:global-uniqueness-in-s-t-coordination}]
Write $y=x_1$ and let $z\in\{0,\ldots,k\}$ be the number of action-$1$ observations in a sample of size $k$. 
Since $F_1(y) - F_2(y) = s y - t(1 - y) = (s + t) y - t$, for each $z$, the inferred payoff difference is 
\begin{align}
    d_z = (s+t) \frac z k - t \ge -t, 
\end{align}
which is increasing in $z$. 
The corresponding logit choice probability of action $1$ is  
\begin{align}
    \rho_z^\eta
    = p(d_z) 
    \quad \text{where}\quad 
     p(d) \Is \frac{1}{1+\exp(-d/\eta)}. 
     \label{eq:def.p.app}
\end{align}
With $\{\rho_z^\eta\}_{z = 0}^k$, the sampling logit choice probability is the Bernstein polynomial
\begin{align}
    \ell_\eta(y)
    \Is L_1^{k,\eta}(y,1-y)
    =\sum_{z=0}^k \binom{k}{z}
      y^z(1-y)^{k-z}\rho_z^\eta. 
    \label{eq:coord-Bernstein-map}
\end{align}
An SLE is a zero of $f_\eta(y)\Is\ell_\eta(y)-y$, and \eqref{eq:sLD} is $\Dty = f_\eta(y)$. 

For every $\eta>0$,
$f_\eta(0)=\rho_0^\eta>0$ and $f_\eta(1)=\rho_k^\eta-1<0$, implying all zeros of $f_\eta$ must lie in $(0,1)$, as mentioned in \cref{sec:model}. 

\UlHeading{The $k = 1$ case}
Condition \eqref{eq:global-uniqueness-s-t} holds automatically because $s > (k - 1) t = 0$. 
Solving $y = \Pi_{11} y + \Pi_{21} (1 - y) = \rho_1^\eta y + \rho_{0}^\eta (1 - y)$, 
the unique fixed point is
\begin{align}
    y_\eta^*=\frac{\rho_0^\eta}{\rho_0^\eta+1-\rho_1^\eta} 
    = \frac{1 + \Rme^{-s /\eta}}{1 + \Rme^{-(s - t)/\eta} + 2\Rme^{-s /\eta}} \in (0,1)
    ,
\end{align}
which is globally asymptotically stable and converges to $1$ as $\eta \Toz$ because $s>t$. 

\UlHeading{The $k = 2$ case}
Condition \eqref{eq:global-uniqueness-s-t} also holds automatically in this case: $s > (k - 1)t = t$. 
Existence, uniqueness, and global asymptotic stability follow from \cref{prop:k=2}. 
It remains to identify the zero-noise limit. 
Since $(d_0,d_1,d_2)=(-t,(s-t)/2,s)$, we have
$(\rho_0^\eta,\rho_1^\eta,\rho_2^\eta)\to(0,1,1)$ as $\eta\Toz$. 
Consequently, for every fixed $y\in(0,1)$, 
\begin{align}
    \lim_{\eta \Toz}
    f_\eta(y)
    = y^2+2y(1-y)-y
    = y(1-y)>0
\end{align}
from \cref{eq:k=2.f(x)}. 
Since $f_\eta$ is positive below its unique zero $y_\eta^*$, it follows that
$y<y_\eta^*$ for every fixed $y\in(0,1)$ and all sufficiently small $\eta>0$. 
Thus, $y_\eta^*\to1$. 

\UlHeading{The $k \ge 3$ case} 
Note that $f_\eta$ is concave if $\ell_\eta$ is concave. 
Together with its signs at $y \in \{0,1\}$, concavity of $f_\eta$ implies the uniqueness of its zero in $(0,1)$. 
We show the concavity of $\ell_\eta$ under condition \eqref{eq:global-uniqueness-s-t} below. 

Set $h=(s+t)/k$. Then, $h \ge (s + t) \cdot t / (s + t) = t$ from condition \eqref{eq:global-uniqueness-s-t}, and this implies $d_1 = h - t \ge0$ and $d_z \ge 0$ for all $z \ge 1$. 
For any $d\ge0$, 
$1/2 \le p(d) < 1$ and 
\begin{align}
    p''(d)
    =\frac{1}{\eta^2}p(d)
    \left(1-p(d)\right)
    \left(1-2p(d)\right)
    \le0.
    \label{eq:p''}
\end{align}
Thus, $p$ is concave on $[0,\infty)$. 
Since $d_z\ge0$ for $z\ge1$ and the points $\{d_z\}$ are equally spaced, this concavity implies
\begin{align}
    & \rho_{z+2}^\eta-2\rho_{z+1}^\eta+\rho_z^\eta\le0
    & \forall z=1,\ldots,k-2.
    \label{eq:discrete-concavity-q}
\end{align}
The remaining inequality for $z=0$ also holds.  
In fact, since $p'$ is even and decreasing on $[0,\infty)$, and $|d_1 - u|\le |d_1  + u| = d_1 +u$ for every $u\in[0,h]$, we have 
\begin{align}
    \rho_1^\eta-\rho_0^\eta
    =\int_0^h p'(d_1-u)\,\mathrm{d}u
    &
    =\int_0^h p'(|d_1-u|)\,\mathrm{d}u
    \\ 
    & 
    \ge \int_0^h p'(d_1+u)\,\mathrm{d}u
    =\rho_2^\eta-\rho_1^\eta 
\end{align}
as we have $d_1 = h-t \ge 0$, $d_1 - h = d_0$ and $d_1 + h = d_2$, implying condition \eqref{eq:discrete-concavity-q} for $z = 0$. 
The Bernstein differentiation formula then gives \citep[cf.][p.1750]{Salant-Cherry-ECTA2020}
\begin{align}
    \ell_\eta''(y)
    =k(k-1)\sum_{z=0}^{k-2}\binom{k-2}{z}
      y^z(1-y)^{k-2-z}
      \left(\rho_{z+2}^\eta-2\rho_{z+1}^\eta+\rho_z^\eta\right)
    \le0.
\end{align}
Thus, $f_\eta$ is concave, implying that it has exactly one zero $y_\eta^*\in(0,1)$. 
Moreover, $f_\eta$ is positive below this zero and negative above it, proving global asymptotic stability. 

\UlHeading{Zero-noise limit for the $k \ge 3$ case} 
It remains to identify the $\eta \Toz$ limit for $k \ge 3$. 
We first observe that there is no interior fixed point in this limit by showing that $f_0(y) = \ell_0(y) - y > 0$ for all $y\in (0,1)$. 
Specifically, if $s>(k-1)t$, then as $\eta\Toz$, $\rho_0^\eta\to0$ and $\rho_z^\eta\to1$ for all $z\ge1$, and hence, using the binomial theorem, 
\begin{align}
    & f_0(y) = \ell_0(y) - y =1-(1-y)^k - y = (1 - y)(1 - (1 - y)^{k - 1}) >0
    & \forall y\in(0,1).
\end{align}
If $s=(k-1)t$, then $\rho_0^\eta\to0$, $\rho_1^\eta=1/2$, and $\rho_z^\eta\to1$ for $z\ge2$. 
In this case, 
\begin{align}
    \ell_0(y)
    &=\frac{k}{2}y(1-y)^{k-1}
      +\sum_{z=2}^k\binom{k}{z}y^z(1-y)^{k-z} 
    =1-(1-y)^k-\frac{k}{2}y(1-y)^{k-1}
\end{align}
from the binomial theorem. 
Then, for all $y\in(0,1)$, setting $r=1-y$, we have 
\begin{align}
    f_0(y) = \ell_0(y)-y
    &=1-r^k-\frac{k}{2}(1-r)r^{k-1}-(1-r) \\
    &=r-r^k-\frac{k}{2}(1-r)r^{k-1} \\
    &=r(1-r)\left(
        \sum_{j=0}^{k-2}r^j-\frac{k}{2}r^{k-2}
      \right) \\
    &>r(1-r)\left(k-1-\frac{k}{2}\right)r^{k-2}
    >0. 
    \label{eq:coord-boundary-limit-map}
\end{align}
The third equality uses
$r-r^k=r(1-r^{k-1})=r(1-r)\sum_{j=0}^{k-2}r^j$. 
The first strict inequality uses
$\sum_{j=0}^{k-2}r^j>(k-1)r^{k-2}$ for $r\in(0,1)$, and the second uses $k\ge3$. 
Thus, in either case $f_0(y) > 0$ for all $y \in (0,1)$, implying that every limit point of $y_\eta^*$ must belong to $\{0,1\}$. 

We now show that $\lim_{\eta\Toz} y_\eta^* = 0$ cannot happen, implying $\lim_{\eta\Toz} y_\eta^* = 1$. 
Since every term in \eqref{eq:coord-Bernstein-map} is nonnegative, the fixed-point equation $y_\eta^*=\ell_\eta(y_\eta^*)$ implies, by retaining only the term for $z=1$, that
\begin{align}
    y_\eta^*
    \ge k y_\eta^*(1-y_\eta^*)^{k-1}\rho_1^\eta.
\end{align}
Since $y_\eta^*>0$, dividing both sides by $y_\eta^*$ yields
\begin{align}
    1\ge k\rho_1^\eta(1-y_\eta^*)^{k-1}.
    \label{eq:bound-y}
\end{align}
Note that $\rho_1^\eta = p(d_1)$ and $d_1=(s + t)/k - t = (s-(k-1)t)/k$.
If $s>(k-1)t$, then $d_1>0$ and hence $\rho_1^\eta=p(d_1)\to1$ as $\eta\Toz$. 
If $s=(k-1)t$, then $d_1=0$ and $\rho_1^\eta=p(0)=1/2$ for every $\eta>0$. 
If $y_\eta^*\to0$, the right-hand side of the preceding bound \eqref{eq:bound-y} therefore converges to $k>1$ in the strict case and to $k/2 \ge 3/2 > 1$ in the equality case. 
Both conclusions contradict the bound \eqref{eq:bound-y}, so $y_\eta^*\not\to0$. 
Consequently, $y_\eta^*\to1$. 

\UlHeading{Nonuniqueness for the $s < (k- 1)t$ case} 
Finally, suppose that condition \cref{eq:global-uniqueness-s-t} does not hold: $s<(k-1)t$. 
Then, we must have $k\ge3$ since $s>t$, as well as $d_1<0$ and $d_{k-1}>0$. 
In the limit as $\eta\Toz$, therefore, $\rho_0^\eta,\rho_1^\eta\to0$ and $\rho_{k-1}^\eta,\rho_k^\eta\to1$. 
Let $\ell_0$ denote the pointwise limit of \eqref{eq:coord-Bernstein-map}. 
Writing $\rho_z^0\Is\lim_{\eta\Toz}\rho_z^\eta$, we have $0\le\rho_z^0\le1$, $\rho_0^0=\rho_1^0=0$, and $\rho_{k-1}^0=\rho_k^0=1$. 
Since $y^z\le y^2$ for $z\ge2$ and $(1-y)^{k-z}\le(1-y)^2$ for $z\le k-2$, the binomial theorem gives, for every $y\in[0,1]$,
\begin{align}
    \ell_0(y)
    &=\sum_{z=2}^k\binom{k}{z}y^z(1-y)^{k-z}\rho_z^0 
    \le y^2\sum_{z=2}^k\binom{k}{z} 
    =Cy^2, \\
    1-\ell_0(y)
    &=\sum_{z=0}^{k-2}\binom{k}{z}y^z(1-y)^{k-z}\left(1-\rho_z^0\right) 
    \le (1-y)^2\sum_{z=0}^{k-2}\binom{k}{z} 
    =C(1-y)^2,
\end{align}
where both truncated binomial sums equal $C\Is2^k-k-1>0$. 
Choose $0<\epsilon<\min\{1/2,1/C\}$. 
Setting $a=\epsilon$ and $b=1-\epsilon$, we have $0<a<b<1$ and
\begin{align}
    f_0(a) &= \ell_0(a)-a \le Ca^2-a<0, \\
    f_0(b) &= \ell_0(b)-b \ge (1-b)-C(1-b)^2>0.
\end{align}
Since $f_\eta(y)=\ell_\eta(y)-y$ converges pointwise to $f_0(y)=\ell_0(y)-y$, the strict inequalities $f_\eta(a)<0<f_\eta(b)$ hold for all sufficiently small $\eta>0$. 
Moreover, for every $\eta>0$, $f_\eta(0)=\rho_0^\eta>0$ and $f_\eta(1)=\rho_k^\eta-1<0$. 
The intermediate value theorem therefore yields one zero of $f_\eta$ in each of $(0,a)$, $(a,b)$, and $(b,1)$. 

It remains to show that $f_\eta$ has no other zeros. 
The identity 
\begin{align}
    y
    =\sum_{z=0}^k \binom{k}{z}
      y^z(1-y)^{k-z} \frac{z}{k},
\end{align}
allows us to rewrite
\begin{align}
    f_\eta(y)
    =\sum_{z=0}^k \binom{k}{z}
      y^z(1-y)^{k-z}c_z^\eta
    \quad \text{where} \quad 
    c_z^\eta\Is \rho_z^\eta-\frac{z}{k}.
\end{align}
For $y\in(0,1)$, set $r=y/(1-y)>0$. Then
\begin{align}
    f_\eta(y)
    =(1-y)^k \Tlf(r)
    \quad 
    \text{where}
    \quad
    \Tlf(r) \Is \sum_{z=0}^k\binom{k}{z}c_z^\eta r^z.
    \label{eq:tlf}
\end{align}
Since $y\mapsto r=y/(1-y)$ is a bijection from $(0,1)$ to $(0,\infty)$ and $(1-y)^k>0$, the zeros of $f_\eta$ in $(0,1)$ correspond to the positive roots of $\Tlf$, with the same multiplicities. 

We resort to Descartes' rule of signs to bound the number of positive roots of $\Tlf$. To examine the signs of the coefficients $\{c_z^\eta\}_{z = 0}^k$ in $\Tlf$, define, using $p$ from \cref{eq:def.p.app},
\begin{align}
    g_\eta(u)\Is p\bigl((s+t)u-t\bigr)-u,
\end{align}
so that $c_z^\eta=g_\eta(z/k)$. 
Its second derivative is
\begin{align}
    g_\eta''(u)
    =(s+t)^2p''\bigl((s+t)u-t\bigr).
\end{align}
Because $p''(d)>0$ for $d<0$ and $p''(d)<0$ for $d>0$, $g_\eta'$ is strictly increasing on $(0,u^*)$ and strictly decreasing on $(u^*,1)$, where $u^*=t/(s+t)$. 
$g_\eta'$ therefore has at most two zeros. 
If $g_\eta$ had four distinct zeros in $[0,1]$, Rolle's theorem would give a distinct zero of $g_\eta'$ between each pair of consecutive zeros, for a total of three, which is impossible. 
Thus, $g_\eta$ has at most three distinct zeros. 
The coefficients $\{c_z^\eta\}_{z=0}^k=\{g_\eta(z/k)\}_{z=0}^k$ can therefore have at most three sign changes (after zero terms are omitted), because otherwise continuity would force $g_\eta$ to have at least four distinct zeros. 
Since the factors $\binom{k}{z}$ are positive, the coefficients of $\Tlf$ have the same sign changes as $\{c_z^\eta\}_{z=0}^k$. 
Descartes' rule of signs therefore bounds the number of positive roots of $\Tlf$ by three, counting multiplicities, and the same bound applies to the zeros of $f_\eta$ in $(0,1)$. 
As the preceding argument gives three distinct zeros for all sufficiently small $\eta>0$, there are exactly three SLE for such $\eta$. 
\end{proof}

\section{Bilingual game}
\label{app:bilingual} 

Depending on the game, noise can change the geometry of the phase portrait in a substantial way. 
We take the \emph{bilingual game} \citep{Galesloot-Goyal-JME1997,Goyal-Janssen-JET1997} as an example. 
The game models language coordination with a bilingual option, and has been an important subject of study in the literature on learning and contagion in games \citep[e.g.,][]{Immorlica-etal-EC2007,Oyama-Takahashi-JET2015,Arigapudi-JME2020,Arigapudi-JME2024,Arigapudi-JEDC2024}. 

Consider the linear population game with base payoff matrix 
\begin{align}
    & 
        A = 
        \begin{bmatrix}
            1 + g & 0 & 1 + g \\ 1 & 1 & 1 \\ 1 + g - c & 1 - c & 1 + g - c
        \end{bmatrix}
    & 
        \left(0 < g < 1,\ \  0 < c < \tfrac{g}{1 + g} \right). 
        \label{eq:A.BL}
\end{align}
The game \eqref{eq:A.BL} is a variant of the bilingual game in which agents choose among three actions: adopting technology $1$, adopting technology $2$, or using a compatible interface that allows interaction with both technologies. 
Technology $1$ yields a higher coordination payoff than technology $2$: when two agents using technology $1$ meet, they obtain $1+g$, whereas coordination on technology $2$ yields payoff $1$, with $0<g<1$. 
The compatibility option enables interaction with either technology but incurs a cost $c>0$. 
Under the parameter restriction $0<c<\tfrac{g}{1+g}$, the compatibility option is viable but strictly costly. 
It allows agents to hedge against coordination risk while remaining dominated by successful coordination on the superior technology~$1$.

\Cref{fig:BL} depicts the phase portraits of the four dynamics. 
Under the best response dynamic \eqref{eq:BRD}, trajectories eventually reach the superior pure state $e_1$ (\cref{fig:BL-BR}), except for those starting close to $e_2$. 
In particular, states with a substantial share of the compatibility action tend to move first toward higher use of action $3$, and subsequently toward action $1$, reflecting the transitional role of the compatibility option. 
Sampling noise yields a sharp selection (\cref{fig:BL-SBR}), just as in Young's game. 
On the other hand, both \eqref{eq:LD} and \eqref{eq:sLD} admit interior stationary states (\cref{fig:BL-LD,fig:BL-SLD}), a feature that changes the geometry of the flow more substantially and contrasts with \eqref{eq:BRD} and \eqref{eq:sBRD}. 
Developing general conditions distinguishing these two forms of behavior is an interesting direction for future research.

\begin{figure}
    \centering
    \begin{subfigure}[c]{.49\textwidth}
        \includegraphics[width=.85\hsize]{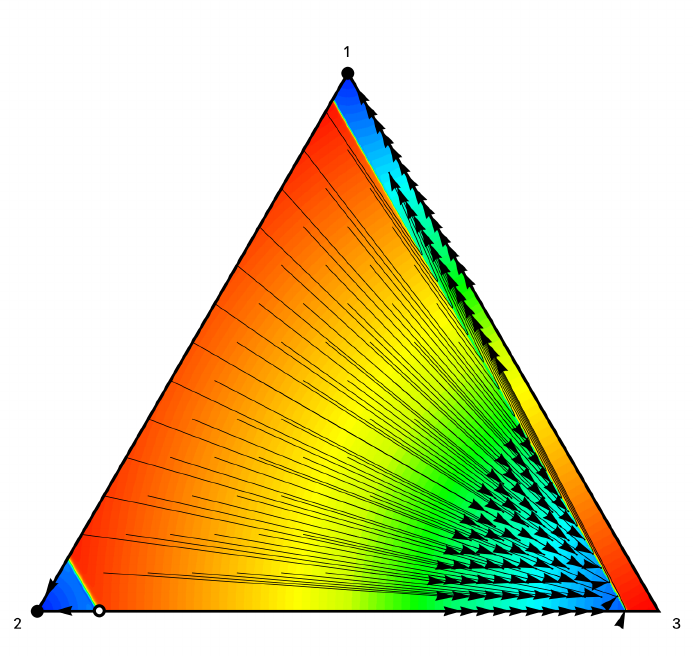}
        \caption{Best response}
        \label{fig:BL-BR}
    \end{subfigure}
    \begin{subfigure}[c]{.49\textwidth}
        \includegraphics[width=.85\hsize]{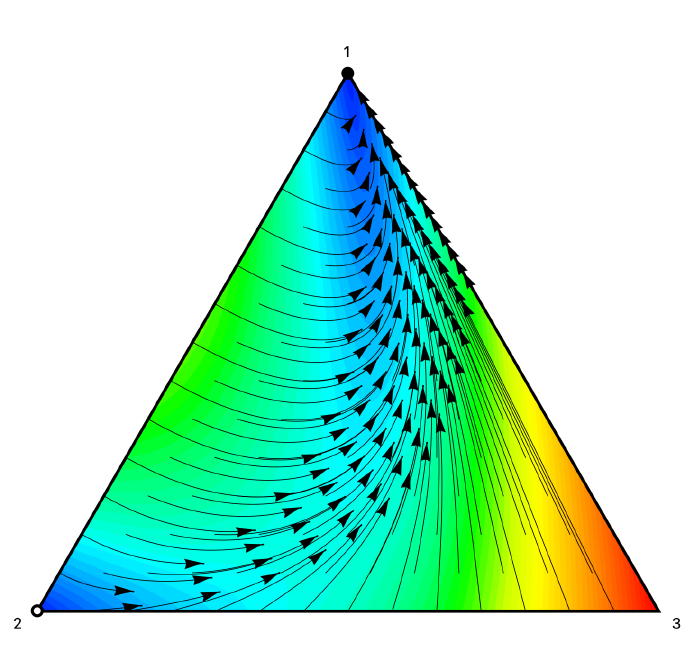}
        \caption{Sampling best response ($k = 2$)}
        \label{fig:BL-SBR}
    \end{subfigure}

    \begin{subfigure}[c]{.49\textwidth}
        \includegraphics[width=.85\hsize]{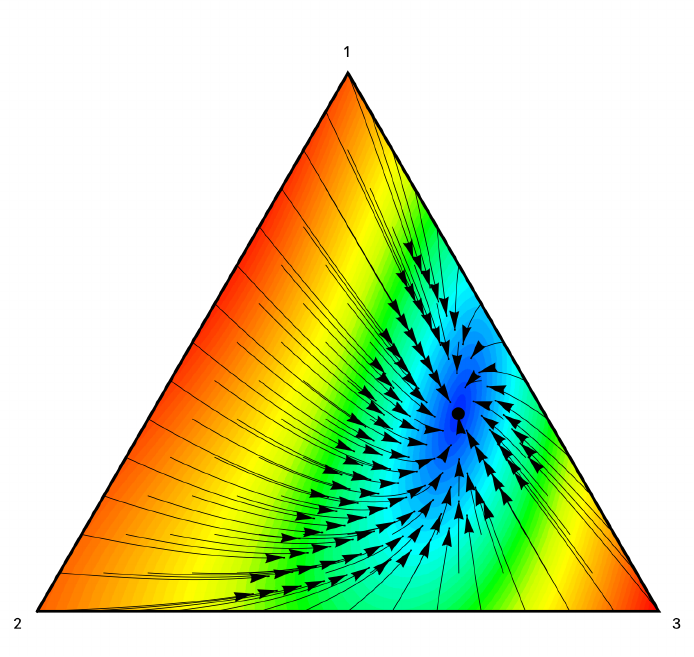}
        \caption{Logit ($\eta = 0.3$)}
        \label{fig:BL-LD}
    \end{subfigure}
    \begin{subfigure}[c]{.49\textwidth}
        \includegraphics[width=.85\hsize]{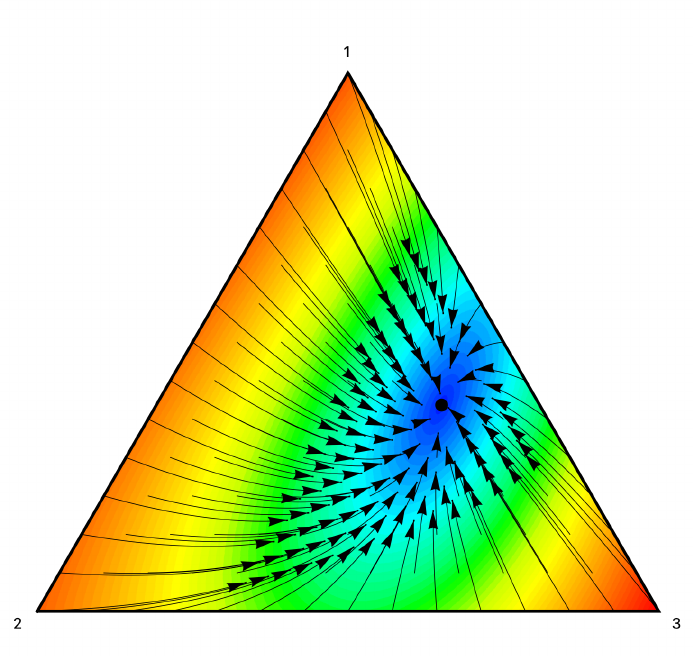}
        \caption{Sampling logit ($k = 2$, $\eta = 0.3$)}
        \label{fig:BL-SLD}
    \end{subfigure}

    \caption{Phase diagrams of the four dynamics in the bilingual game \eqref{eq:A.BL} for the case $g = 0.5$ and $c = 0.05$. This figure is generated with the \texttt{Dynamo} software \citep{Franchetti-Sandholm-BT2013}.} 
    \label{fig:BL}
\end{figure}

\begin{figure}
    \centering
    \begin{subfigure}[c]{.49\textwidth}
        \includegraphics[width=.95\hsize]{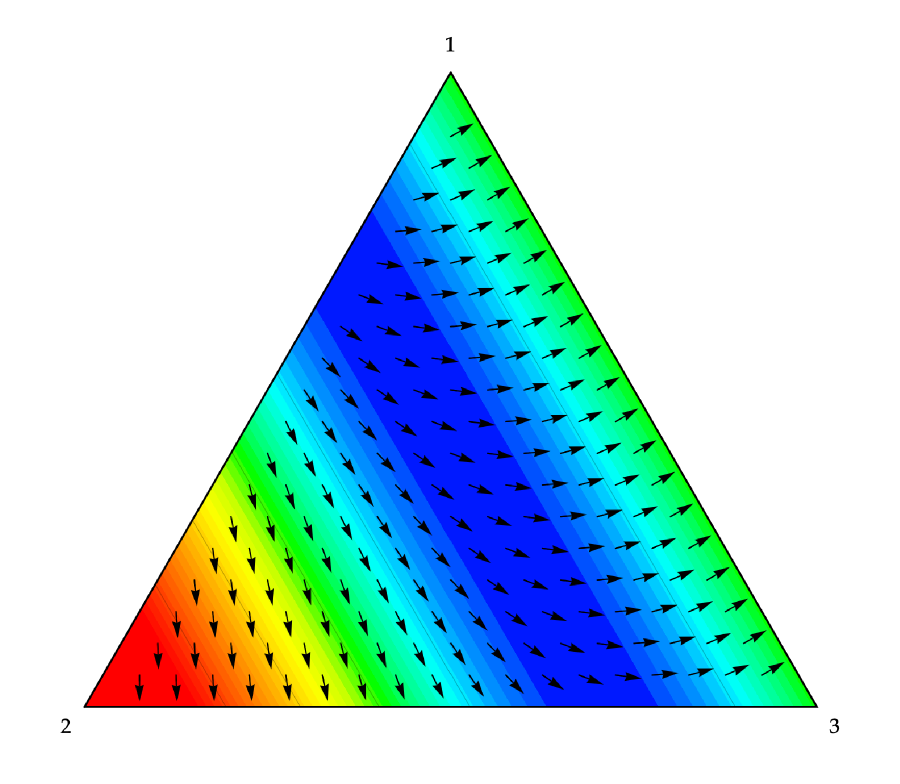}
        \caption{Payoff $F$}
        \label{fig:BL-F}
    \end{subfigure}
	\begin{subfigure}[c]{.49\textwidth}
        \includegraphics[width=.95\hsize]{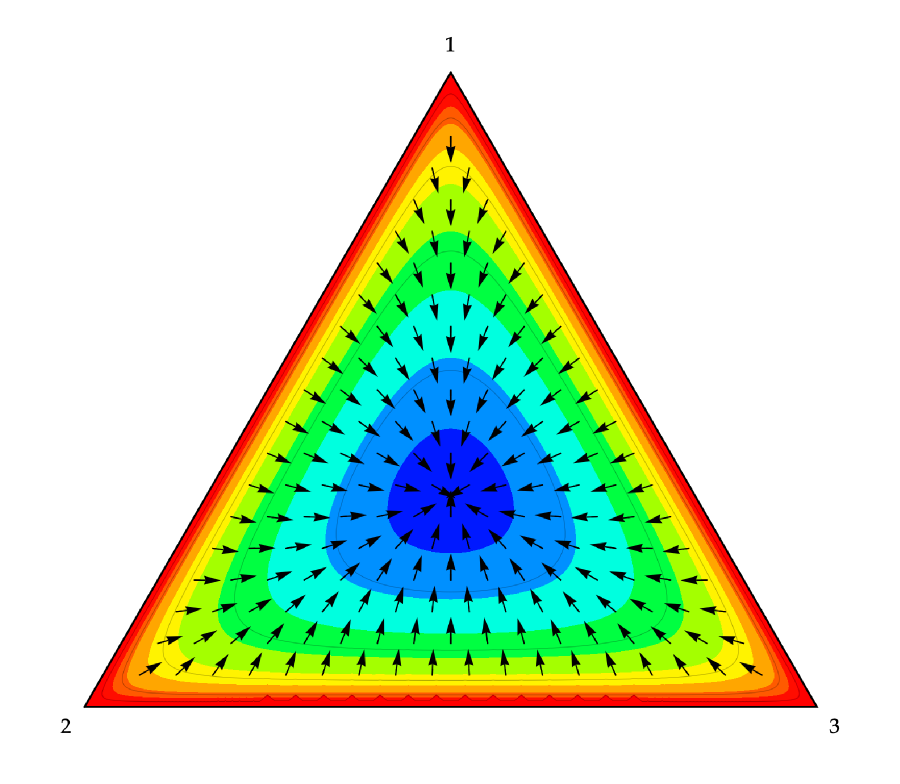}
        \caption{Logit virtual payoff $H$ ($\eta = 0.3$)}
        \label{fig:BL-H}
    \end{subfigure}

    \begin{subfigure}[c]{.49\textwidth}
        \includegraphics[width=.95\hsize]{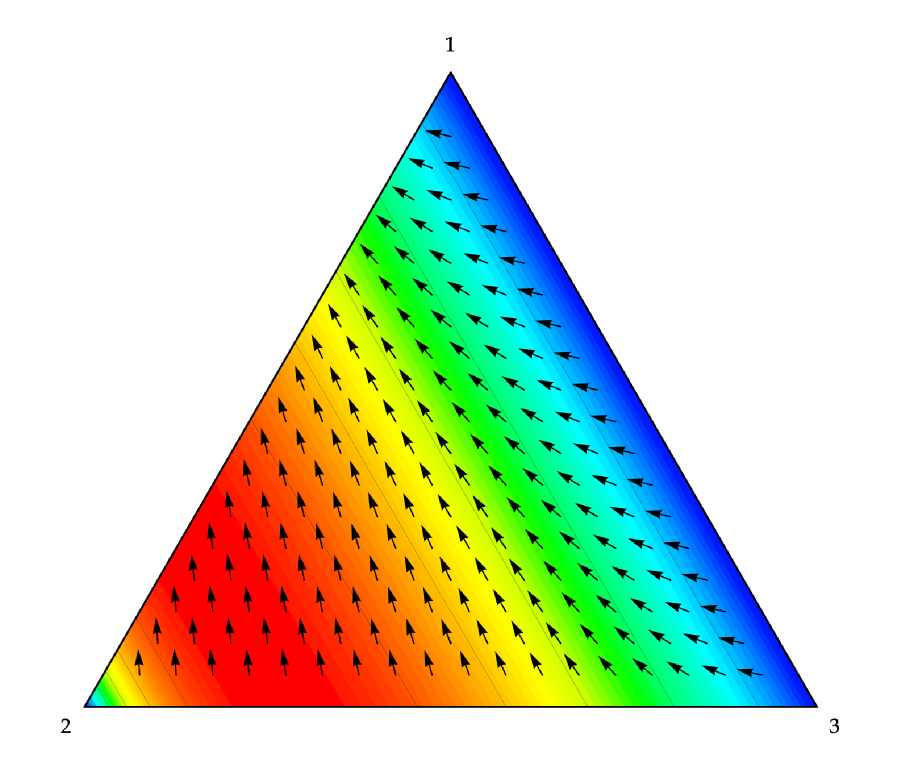}
        \caption{Exact payoff distortion $G$ ($\eta = 0.3$)}
        \label{fig:BL-SL-G}
    \end{subfigure}
    \begin{subfigure}[c]{.49\textwidth}
        \includegraphics[width=.95\hsize]{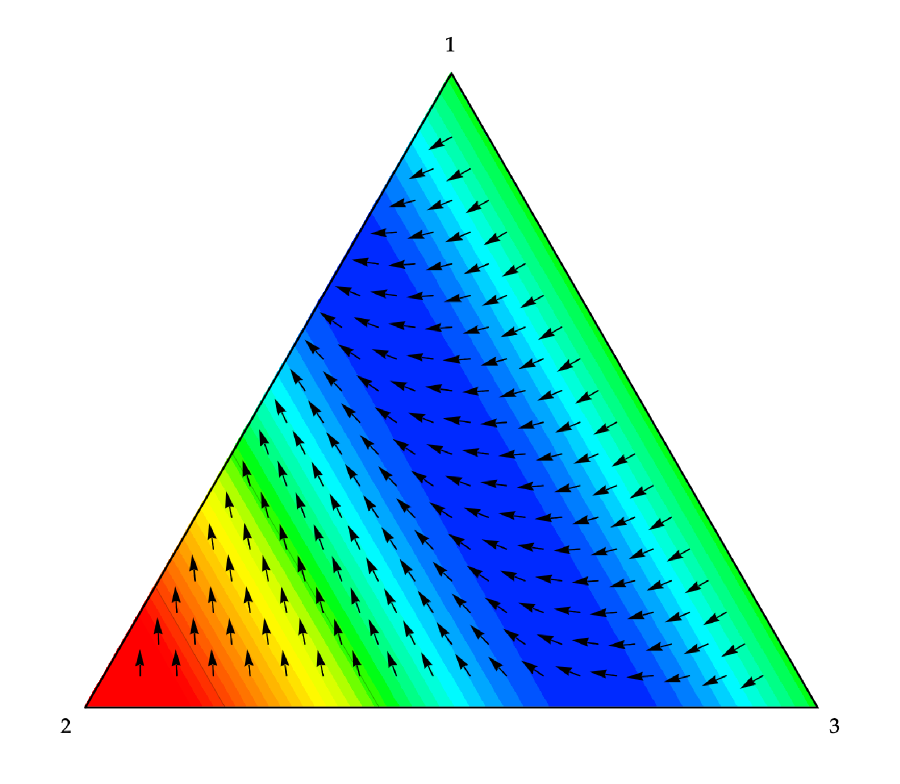}
        \caption{Exact payoff distortion $G$ ($\eta = 0.001$)}
        \label{fig:BL-SBR-G}
    \end{subfigure}
    \caption{Projected payoff fields in the bilingual game \eqref{eq:A.BL} with $g = 0.5$ and $c = 0.05$. Panels~(A) and~(B) show the payoff $F$ and the logit virtual payoff $H$ at $\eta = 0.3$, respectively. Panels~(C) and~(D) show the exact payoff distortion $G$ for $k = 2$ at $\eta = 0.3$ and $\eta = 0.001$, respectively. Arrows show the normalized direction of each field's projection onto the tangent space of $X$, and background contours show its magnitude. The near-zero noise level in panel~(D) approximates the payoff distortion associated with the sampling best response.}
    \label{fig:exact-BL-G} 
\end{figure}

\end{document}